\documentclass[a4paper,11pt,english,final]{article}
\pdfoutput=1
\usepackage{kpfonts}

\usepackage{ifdraft}
\ifdraft{\newcommand{\authnote}[3]{{\color{#3} {\bf  #1:} #2}}}{\newcommand{\authnote}[3]{}}
\newcommand{\snote}[1]{\authnote{ Srini}{#1}{red}}
\newcommand{\anote}[1]{\authnote{ Andras}{#1}{green}}

\newcommand{\argmax}[1]{\underset{#1}{\arg\!\max}}

\usepackage[utf8]{inputenc}
\usepackage[T1]{fontenc}
\usepackage{fullpage}
\usepackage[english]{babel}
\usepackage{verbatim}

\usepackage{tikz}
\usepackage{amssymb}
\usepackage{mathtools}
\usepackage{dsfont}
\usepackage{bbm}
\usepackage{mleftright}\mleftright
\usepackage{thmtools} 
\usepackage{thm-restate}

\usepackage{algorithm}
\usepackage{algcompatible}
\usepackage{enumitem}
\usepackage{float}

\usepackage{caption}
\usepackage{subcaption}

\usepackage{titling}

\newcommand*{\QEDB}{\hfill\ensuremath{\square}}
\newcommand{\eps}{\varepsilon}

\newcommand{\ketbra}[2]{|#1\rangle\! \langle #2|}
\newcommand{\braketbra}[3]{\langle #1|#2| #3 \rangle}
\newcommand{\nrm}[1]{\left\lVert #1 \right\rVert}
\newcommand{\bigO}[1]{\mathcal{O}\left( #1 \right)}
\newcommand{\bigOt}[1]{\widetilde{\mathcal{O}}\left( #1 \right)}
\newcommand{\ctrlA}{\push{\rule{1.5mm}{1.5mm}}}

\newcommand{\pvp}{\vec{p}{\kern 0.45mm}'}
\let\oldnabla\nabla
\renewcommand{\nabla}{\oldnabla\!}

\DeclarePairedDelimiter\bra{\langle}{\rvert}
\DeclarePairedDelimiter\ket{\lvert}{\rangle}
\DeclarePairedDelimiterX\braket[2]{\langle}{\rangle}{#1 \delimsize\vert #2}
\newcommand{\underflow}[2]{\underset{\kern-60mm \overbrace{#1} \kern-60mm}{#2}}

\def\polylog{\mathrm{polylog}}

\def\Pr{\mathrm{Pr}}

\providecommand{\trnorm}[1]{\left\lVert#1\right\rVert_1}

\long\def\ignore#1{}

\newtheorem{theorem}{Theorem}
\newtheorem{corollary}[theorem]{Corollary}
\newtheorem{lemma}[theorem]{Lemma}
\newtheorem{fact}[theorem]{Fact}

\newtheorem{definition}[theorem]{Definition}
\newtheorem{claim}[theorem]{Claim}
\newtheorem{remark}[theorem]{Remark}

\newcommand{\A}{\ensuremath{\mathcal{A}}}

\newcommand{\N}{\ensuremath{\mathbb{N}}}

\newcommand{\F}{\ensuremath{\mathcal{F}}}

\newcommand{\R}{\ensuremath{\mathbb{R}}}
\newcommand{\Z}{\ensuremath{\mathbb{Z}}}

\newcommand{\E}{\mathcal{E}}

\usepackage{tikz}

\usetikzlibrary{backgrounds,fit,decorations.pathreplacing,calc}

\usepackage{qcircuit}

\newenvironment{proof}
{\noindent {\bf Proof. }}
{{\hfill $\Box$}\\	\smallskip}

\usepackage[final]{hyperref}
\hypersetup{
	colorlinks = true,
	allcolors = {blue},
}

\title{
	Optimizing quantum optimization algorithms\\ via faster quantum gradient computation
}
\author{
	András Gilyén\thanks{QuSoft/CWI, Science Park 123, 1098 XG Amsterdam, Netherlands. Supported by ERC Consolidator Grant QPROGRESS.  \texttt{\{gilyen,arunacha\}@cwi.nl} }
	\and	
	Srinivasan Arunachalam$^*$
	\and	
	Nathan Wiebe\thanks{Station Q QuArC, Microsoft Research, USA.  \texttt{nawiebe@microsoft.com} }
}
\date{}

\begin{document}
	
	\maketitle

	\begin{abstract}
		We consider a generic framework of optimization algorithms based on gradient descent. We develop a quantum algorithm that computes the gradient of a multi-variate real-valued function $f:\mathbb{R}^d\rightarrow \mathbb{R}$ by evaluating it at only a logarithmic number of points in superposition. Our algorithm is an improved version of Jordan's gradient computation algorithm~\cite{QuantGrad}, providing an approximation of the gradient $\nabla f$ with quadratically better dependence on the evaluation accuracy of $f$, for an important class of smooth functions. Furthermore, we show that most objective functions arising from quantum optimization procedures satisfy the necessary smoothness conditions, hence our algorithm provides a quadratic improvement in the complexity of computing their gradient. We also show that in a continuous phase-query model, our gradient computation algorithm has optimal query complexity up to poly-logarithmic factors, for a particular class of smooth functions. Moreover, we show that for low-degree multivariate polynomials our algorithm can provide exponential speedups compared to Jordan's algorithm in terms of the dimension $d$.
		
		One of the technical challenges in applying our gradient computation procedure for quantum optimization problems is the need to convert between a probability oracle (which is common in quantum optimization procedures) and a phase oracle (which is common in quantum algorithms) of the objective function $f$. We provide \emph{efficient} subroutines to perform this delicate interconversion between the two types of oracles incurring only a logarithmic overhead, which might be of independent interest. Finally, using these tools we improve the runtime of prior approaches for training quantum auto-encoders, variational quantum eigensolvers (VQE), and quantum approximate optimization algorithms (QAOA).		
	\end{abstract}

\newpage

\tableofcontents

\newpage
\section{Introduction}\label{sec:Introduction} 
In recent years quantum technology has progressed at a fast pace. As quantum computing enters the realm of engineering, it is important to understand how it can provide a speedup for real-world problems.
On the theoretical side, the last two decades have seen many quantum algorithms for various computational problems in number theory~\cite{shor:factoring}, search problems~\cite{grover:search}, formula evaluation~\cite{acrsz:andor}, solving linear systems~\cite{harrow2009quantum}, Hamiltonian simulation~\cite{BerryChildsSim15} and machine learning tasks~\cite{wiebe2015quantum,wiebe:quantumperceptronmodels}.\footnote{See the ``quantum algorithms zoo'': \url{http://math.nist.gov/quantum/zoo/} for a comprehensive list of quantum algorithms for computational problems.} Less attention has been devoted to developing quantum algorithms for discrete and continuous optimization problems which are possibly intractable by classical computers.  Na\"{\i}vely, since Grover's quantum algorithm~\cite{grover:search} quadratically improves upon the classical algorithm for searching in a database, we can simply use it to speed up all discrete optimization algorithms which involve searching for a solution among a set of possible solutions. However, in real-world applications, many problems have \emph{continuous} parameters, where an alternative quantum optimization approach might fit the problem better.

Optimization is a fundamentally important task that touches on virtually every area of science.
Unlike computational problems, quantum algorithms for optimization have not been very well understood. Recently, a handful of quantum algorithms considering specific continuous optimization problems have been developed for: Monte Carlo methods~\cite{montanaro:montecarlo}, derivative-free optimization~\cite{arunachalam:mastersthesis}, least squares fitting~\cite{wiebe2012quantum}, quantum annealing~\cite{tadashi:quantumannealing}, quantum adiabatic optimization~\cite{farhi:adiatic}, optimization algorithms for satisfiability and travelling salesman problem~\cite{hogg:quantumoptimization,arunachalam:mastersthesis} and quantum approximate optimization~\cite{farhi2014qaoa}. Also, very recently, there has been work on quantum algorithms for solving linear and semi-definite programs~\cite{brandao:quantumsdp,van2017quantum,brandaoSDPsAndLearning}.  
However, applying non-Grover techniques to real-word optimization problems has proven challenging, because generic problems usually fail to satisfy the delicate requirements of these advanced quantum~techniques.

In this paper, we consider \emph{gradient-based optimization}, which is a well-known technique to handle continuous-variable optimization problems. We develop an improved quantum algorithm for gradient computation (using non-Grover techniques), which provides a quadratic reduction in query and gate complexity (under reasonable continuity assumptions). Moreover, we show that our new gradient computation algorithm is essentially optimal for a certain class of functions. Finally, we apply our algorithm to improve quantum optimization protocols used for solving important real-world problems, such as quantum chemistry simulation and quantum neural network training.

\subsection{Prior work on quantum gradient methods} \label{sec:priorwork}
Our gradient computation algorithm is based on Jordan's~\cite{QuantGrad} quantum algorithm, which provides an exponential quantum speedup for gradient computation in a black-box model, and similarly to the classical setting it provides a finite precision binary representation of the gradient.
Bulger~\cite{bulger2005BasinHop} later showed how to combine Jordan's algorithm with quantum minimum finding~\cite{durr&hoyer:minimum} to improve gradient-descent methods.

Recently, Rebentrost et al.~\cite{rebentrost:quantumgradientdescent} and Kerenidis and Prakash~\cite{kerenidis:quantumgraddescent} considered a very different approach, 
where they represent vectors as quantum states, which can lead to exponential improvements in terms of the dimension for \emph{specific} gradient-based algorithms.

Rebentrost et al.~\cite{rebentrost:quantumgradientdescent} obtained speedups for first and second-order iterative methods (i.e., gradient descent and Newton's method) for polynomial optimization problems. The runtime of their quantum algorithm achieves poly-logarithmic dependence on the dimension $d$ but scales exponentially with the number of gradient steps~$T$. Kerenidis and Prakash~\cite{kerenidis:quantumgraddescent} described a gradient descent algorithm for the \emph{special case} of quadratic optimization problems. The algorithm's runtime scales linearly with the number of steps $T$ and in some cases can achieve poly-logarithmic dependence in the dimension $d$ as it essentially implements a version of the HHL algorithm~\cite{harrow2009quantum} for solving linear systems. However, their appealing runtime bound requires a very strong access model for the underlying~data.

\subsection{Classical gradient-based optimization algorithms} \label{sec:naivealg}
In this section, we give a brief description of a simple classical gradient-based algorithm for optimization. Consider the multi-variate function $p:\R^d\rightarrow \R$ and assume for simplicity that $p$ is well-behaved, i.e., it is bounded by some absolute constant and differentiable everywhere. The  problem is, given $p:\R^d\rightarrow \R$,~compute
\begin{align}
\label{eq:basicoptimizationproblem}
\text{OPT}=\min \{p(\pmb{x}) : \pmb{x}\in \R^d\}.
\end{align}
A heuristic solution of the optimization problem~\eqref{eq:basicoptimizationproblem} can be obtained by computing the \emph{gradient} \begin{align}\label{eq:gradientformula}
\nabla p= \left(\frac{\partial p}{\partial x_1},\frac{\partial p}{\partial x_2},\ldots,\frac{\partial p}{\partial x_d} \right)
\end{align}
of $p$.
It is a well-known fact in calculus that $p$ decreases the \emph{fastest} in the direction of $-(\nabla p (\pmb{x}))$. 
This simple observation is the basis of gradient-based optimization algorithms. 

Now we describe a simple heuristic gradient-descent algorithm for computing~\eqref{eq:basicoptimizationproblem}: pick a random point $\pmb{x}^{(0)}\in\R^d$, 
compute $\nabla p (\pmb{x}^{(0)})$ and take a $\delta$-step in the direction of $- \nabla p (\pmb{x}^{(0)})$ leading to $\pmb{x}^{(1)}=\pmb{x}^{(0)}-\delta \nabla p (\pmb{x}^{(0)})$ (for some step size $\delta > 0$). Repeat this gradient update for~$T$ steps, obtaining $\pmb{x}^{(T)}$ which has hopefully approached some local minima of~\eqref{eq:basicoptimizationproblem}. 
Finally repeat the whole procedure for $N$ different starting points $\big\{\pmb{x}^{(0)}_1,\ldots,\pmb{x}^{(0)}_N\big\}$ and take the minimum of  $\big\{p(\pmb{x}^{(T)}_1),\ldots,p(\pmb{x}^{(T)}_N)\big\}$ after $T$ gradient steps.

Given the generality of the optimization problem~\eqref{eq:basicoptimizationproblem} and the simplicity of the algorithm, gradient-based techniques are widely used in mathematics, physics and engineering. In practice, especially for well-behaved functions $p$, gradient-based algorithms are known to converge very quickly to a local optimum and are often used, e.g., in state-of-the-art algorithms for deep learning~\cite{ruder:blog}, which has been one of the recent highlights in classical machine~learning.


\subsection{Complexity measure and quantum sampling}
The starting point of this work was the simple observation that 
most quantum optimization procedures translate the objective function to the probability of some measurement outcome, and therefore evaluate it via sampling. To reflect this fact, we use an oracular model to represent our objective function, that is much weaker than the oracle model considered by Jordan~\cite{QuantGrad}.  To be precise, we work with a coherent version of the classical random sampling procedure, i.e., we assume that the function is given by a \emph{probability oracle}:
\begin{equation}\label{eq:proboracle}
U_p:\ket{\pmb{x}}\ket{0}\to \sqrt{p(\pmb{x})}\ket{\pmb{x}}\ket{1}+\sqrt{1-p(\pmb{x})}\ket{\pmb{x}}\ket{0} \quad \text{ for every } \pmb{x},
\end{equation}
where the continuous input variable $\pmb{x}$ is represented as a finite-precision binary encoding of~$\pmb{x}$. 

We address the question: how many queries to $U_p$ suffice to compute the \emph{gradient}~of~$p$?
It is not hard to see that using empirical estimation it suffices to use $\bigO{1/\eps^2}$ samples (obtained by querying $U_p$) in order to evaluate $p(\pmb{x})$ with additive error $\Theta(\eps)$. Provided that $p$ is smooth we can compute an $\eps$-approximation of $\nabla_i p(\pmb{x})=\frac{\partial p}{\partial x_i}$ by performing $\bigOt{1}$ such function evaluations, using standard classical techniques.
Hence, we can compute an $\eps$-approximation of the gradient $\nabla p(\pmb{x})$ with $\bigOt{d}$ function evaluations of precision $\Theta(\eps)$. The simple gradient-descent algorithm described in the previous section uses $TN$ gradient computations, therefore the overall algorithm can be executed using $\bigOt{TNd/\varepsilon^2}$ samples. 

\paragraph{Quantum speedups for the simple gradient descent algorithm.} 
We briefly describe how to improve the query complexity of the simple gradient-descent algorithm, assuming that we have access to a probability oracle~\eqref{eq:proboracle} of a smooth objective function $p$.
First, we improve the complexity of $\eps$-accurate function evaluations to $\bigO{1/\eps}$ using amplitude estimation~\cite{brassard2002quantum}.
Then, similarly to \cite{bulger2005BasinHop,lara2014HybridOpt}, we improve the parallel search for finding a global minimum using the quantum minimum finding algorithm~\cite{durr&hoyer:minimum,van2017quantum}. 
Additionally, we present a quantum algorithm for gradient computation which quadratically improves the algorithm in terms of the dimension $d$.
In particular, this shows that we can speed up the gradient-based optimization algorithm quadratically in almost all parameters, except the number of iterations $T$. The results are summarized below in~Table~\ref{tab:variousQuantumImprovements}:
\begin{table}[H]
	\centering
	\begin{tabular}{l|cccc}
		Method: 	& Simple algorithm	& +Amp. est. 			& +Grover search 			& +\textbf{This paper} \\ \hline
		Complexity:	& $\bigOt{TNd/\eps^2}$					& $\bigOt{TNd/\textcolor{red}{\eps}}$		& $\bigOt{T\sqrt{\textcolor{red}{N}}d/\eps}$	& $\bigOt{T\sqrt{N\textcolor{red}{d}}/\eps}$
	\end{tabular}
	\caption{Quantum speedups for a simple gradient-descent algorithms}
	\label{tab:variousQuantumImprovements}	
\end{table}	

\textbf{Remark about $T$.} Since gradient descent is ubiquitous in optimization, it has been optimized extensively in the \emph{classical} literature, yielding significant reductions in the number of steps $T$, see for example accelerated gradient methods \cite{NesterovAcceleratedGrad1983,BeckIterative,JainAcceleratingSGD}. We think it should be possible to combine some of these classical results with our quantum speedup, because our algorithm outputs a classical description of the gradient, unlike other recent developments on quantum gradient-descent methods \cite{rebentrost:quantumgradientdescent,kerenidis:quantumgraddescent}. However, there could be some difficulty in applying classical acceleration techniques, because they often require unbiased samples of the approximate gradient, which might be difficult to achieve using quantum~sampling.

\subsection{Conversion between probability and phase oracles}
As mentioned earlier, many quantum optimization procedures access to the objective function via a probability oracle \eqref{eq:proboracle}. However, for most of the quantum techniques that we employ, it is more natural to work with a \emph{phase oracle}, acting as
\begin{equation}\label{eq:phaseoracle}
\mathrm{O}_{\!p}:\ket{\pmb{x}}\to e^{i p(\pmb{x})}\ket{\pmb{x}} \quad \text{ for every } \pmb{x}.
\end{equation}
Using Linear Combination of Unitaries (LCU) techniques~\cite{berry:simHamTaylor}, we show that we can efficiently simulate a phase oracle with $\eps$ precision, using $\bigO{\log(1/\eps)}$ queries to the probability oracle~$U_p$. Similarly, we show that under some reasonable conditions, we can simulate the probability oracle $U_p$ with $\eps$ precision, using $\bigO{\log(1/\eps)}$ queries to the phase oracle $\mathrm{O}_{\!p}$. 

For the purposes of our paper this efficient simulation essentially means that we can interchangeably work with probability or phase oracles, using whichever fits our setting best. We are not aware of any prior result that shows this simulation and  we believe that our oracle conversion techniques could be useful for other applications.

One possible application which is relevant for quantum distribution testing~\cite{LiWuEntropyQueryComp} is the following. 
Suppose we are given access to some probability distribution via a quantum oracle
$$ U:\ket{0}\ket{0} \rightarrow \sum_{x\in X} \sqrt{p(x)}\ket{x}\ket{\psi_x}.$$ 
Let $H= \sum_{x\in X}p(x)\ketbra{x}{x}$ be the Hamiltonian corresponding to the probability distribution $p(x)$, then we can implement Hamiltonian simulation $e^{itH}$ for time $t$ with $\eps$ precision making only $\bigO{|t|\log(|t|/\eps)}$ queries to $U$.

\subsection{Improved gradient computation algorithm}\label{subsec:gradientIntro}
Jordan's algorithm for gradient computation~\cite{QuantGrad} uses yet another \emph{fairly strong} input model, it assumes that $f:\mathbb{R}^d\rightarrow \mathbb{R}$ is given by an $\eta$-accurate binary oracle, which on input $\pmb{x}$, outputs $f(\pmb{x})$ binarily with accuracy $\eta$.
Jordan's quantum algorithm outputs an $\eps$-coordinate-wise approximation of $\nabla f$ using a single evaluation of the binary oracle. The algorithm prepares a uniform superposition of evaluation points over a finite region, then approximately implements the $S=\bigO{\sqrt{d}/\eps^2}$-th power of a phase oracle
$$
\mathrm{O}^S_{\!f}:\ket{\pmb{x}}\rightarrow e^{i S f(\pmb{x})}\ket{\pmb{x}},
$$
using a \emph{single} $\Theta(\eps^2/\sqrt{d})$-accurate evaluation of $f$, and then applies an inverse Fourier transformation to obtain an approximation of the gradient. Although this algorithm only uses a single query, the required precision of the function evaluation can be prohibitive. 
In particular, if we only have access to a probability oracle, it would require $\bigO{\sqrt{d}/\eps^2}$ probability oracle queries to evaluate the function with such precision using quantum amplitude estimation~\cite{brassard2002quantum}. In contrast, our new quantum algorithm requires only $\bigOt{\sqrt{d}/\eps}$ queries to a probability oracle. The precise statement can be found in Theorem~\ref{thm:finalScaling}, and below we give an informal statement.
	\begin{theorem}[Informal]\label{thm:informalFinalScaling}
		There is a gate-efficient quantum algorithm, that given probability oracle $U_p$ access to an analytic function $p:\R^d\rightarrow [0,1]$ having bounded partial derivatives at $\pmb{0}$, computes an approximate gradient $\pmb{g}\in\R^d$ such that $\nrm{\pmb{g}-\nabla p(\pmb{0})}_\infty\leq \eps$ with high probability, using $\bigOt{\sqrt{d}/\eps}$ queries to $U_p$.
		We get similar complexity bounds if we are given phase oracle access to the function.
	\end{theorem}
\paragraph{Proof sketch.} The main new ingredient of our algorithm is the use of higher-degree central-difference formulas, a technique borrowed from calculus. We use the fact that for a one-dimensional analytic function $h:\R\rightarrow \R$ having bounded derivatives at $0$, we can use a $\log(1/\eps)$-degree central-difference formula to compute an $\eps\cdot\log(1/\eps)$-approximation of $h'(0)$ using $\eps$-accurate evaluations of $h$ at $\log(1/\eps)$ different points around $0$.
We apply this result to one-dimensional slices of the $d$-dimensional function $f:\R^d\rightarrow \R$. The main technical challenge in our proof is to show that if $f$ is smooth, then for most such one-dimensional slices, the $k$-th order directional derivatives increase by at most an $\bigO{(\sqrt{d})^k}$-factor compared to the partial derivatives of $f$. As we show this implies that it is enough to evaluate the function $f$ with $\bigOt{\eps/\sqrt{d}}$-precision in order to compute the gradient.
After the function evaluations, our algorithm ends by applying a $d$-dimensional quantum Fourier transform providing a classical description of an approximate gradient, similarly to Jordan's algorithm. \QEDB

In the special case when the probability oracle represents a degree-$k$ multivariate polynomial over a finite domain $[-R,R]^d$, for some $R=\Theta(1)$, we obtain an algorithm with query complexity $\bigOt{k\log(d)/\eps}$, see Theorem~\ref{thm:lagrangeAlg}. In the case when $k=\bigO{\log(d)}$ this gives an \emph{exponential speedup} over Jordan's algorithm with respect to the dimension $d$.

\subsection{Smoothness of probability oracles}\label{subsec:SmothProbability}
We show that the seemingly strong requirements of Theorem~\ref{thm:informalFinalScaling} are naturally satisfied by probability oracles arising from typical quantum optimization protocols. In such protocols, probability oracles usually correspond to the measurement outcome probability of some orthogonal projector $\Pi$ on the output state of a parametrized circuit $U(\pmb{x})$ acting on some fixed initial state~$\ket{\psi}$, i.e., 
\begin{equation*}
	p(\pmb{x})=\braketbra{\psi}{U(\pmb{x})^\dagger\Pi U(\pmb{x})}{\psi}.
\end{equation*}
Usually the parametrized circuit can be written as 
\begin{equation*}
	U(\pmb{x})=U_0\prod_{j=1}^{d}\left(e^{i x_j H_j}\right) U_j,
\end{equation*}
where the $U_j$s are fixed unitaries and the $H_j$s are fixed Hermitian operators. We can assume without too much loss of generality that $\nrm{H_j}\leq 1/2$. Under these conditions we can show that  $p(\pmb{x})$ is analytic, and all partial derivatives of $p$ are upper bounded by $1$ in magnitude, i.e., $p$ satisfies the conditions of Theorem~\ref{thm:informalFinalScaling}. For more details, see Lemma~\ref{lemma:unitaryDerivative}-\ref{lemma:derivativeCombination}.

Throughout the paper, when we say that a function is smooth we mean that it satisfies the requirements of our gradient computation algorithm, i.e., it is analytic and has bounded\footnote{To be precise, by this we mean that the function satisfies the requirements of Theorem~\ref{thm:finalScaling}.} partial derivatives.

\subsection{Lower bounds for gradient computation}
An interesting question is whether we can improve the classical $\bigO{d/\varepsilon^2}$-gradient computation algorithm by a super-quadratic factor? At first sight it would very well seem possible considering that our algorithm gains a speedup using the quantum Fourier transform. Indeed, for low-degree multivariate polynomials we can get an exponential speedup. However, we show that in general this soeed-up is \emph{not} possible for smooth non-polynomial functions, and give a query lower bound of $\Omega(\sqrt{d}/\eps)$ for the complexity of a generic quantum gradient computation~algorithm. 

Proving lower bounds for quantum query complexity is a well-studied subject within the area of quantum computing. In general, there are two well-known methods for proving lower bounds in the  quantum  query model, the polynomial method~\cite{Beals:2001} and the adversary method~\cite{Ambainis:AdversaryMethod00,Hoyerleespalek:negative}. Both these methods are also known to give optimal lower bounds for various combinatorial quantum query problems. However, these techniques \emph{crucially rely} on the discrete nature of the problems, i.e., they assume the oracle in the quantum query algorithm is a discrete phase oracle (like in Eq.~\ref{eq:proboracle} with \emph{Boolean} $p$). More generally, most query lower  bound techniques in quantum computing apply to settings where the input unitaries come from a discrete set which might correspond to some discrete computational problem, allowing for the use of the polynomial and adversary method. It is not at all clear how one would use these techniques to prove a lower bound for a family of continuous unitaries.

Moving away from combinatorial problems, in the continuous-input model when one gets arbitrary phase oracles, not many lower bounds are known. In fact the only quantum lower bounds for the continuous-input model that we are aware of, are the so-called ``complexity-theoretic no-cloning theorem" by Aaronson~\cite{aaronson:quantummoney}, and other problems directly related to phase estimation~\cite{bessenlowerbound} \snote{Are there more?}. Apart from these examples we do not have many lower bounds for quantum algorithms in a continuous-input model. More recently the negative-weight adversary method was adapted to the continuous-input setting by Belovs~\cite{belovsGeneralAdv15}. However, this generalized adversary method has not yet been applied to continuous-input settings, as far as the authors are aware. 
It is unclear to us how one would use the generalized adversary method for this problem, without introducing an overly complicated formalism. 
This suggests that proving a lower bound on gradient computation presents additional challenges and calls for introducing using new techniques.

In order to solve this challenge, we revisit the hybrid method\footnote{Actually one can view the adversary method as a generalization of the hybrid method.} which was one of the earliest lower-bound techniques, originally introduced by Bennett et al.~\cite{Bennett:SearchLowerBound97}. In particular we derive an \emph{intuitive} lower bound result, which applies to arbitrary phase oracles. 

\begin{restatable}{theorem}{hybMetArb}\textbf{\emph{(Hybrid method for arbitrary phase oracles)}}\label{thm:arbHybLow}
	Let $G$ be a (finite) set of labels and let $\mathcal{H}:=\mathrm{Span}(\ket{x}\colon x\in G)$ be a Hilbert space. For a function $\tilde{f}:G\rightarrow \R$ let $\mathrm{O}_{\!\tilde{f}}$ be the phase oracle acting on $\mathcal{H}$ such that $$\mathrm{O}_{\!\tilde{f}}:\ket{x}\to e^{i \tilde{f}(x)}\ket{x} \quad \text{ for every } x\in G.$$
	Suppose that $\F$ is a finite set of functions $G\rightarrow \R$, and the function $f_*\colon G\rightarrow \R$ is not in $F$. If a quantum algorithm makes $T$ queries to a (controlled) phase oracle $\mathrm{O}_{\!\tilde{f}}$ (or its inverse) and for all $f\in \F$ can distinguish with probability at least $2/3$ the case $\tilde{f}=f$ from the case $\tilde{f}=f_*$, then
	$$  T \geq \frac{\sqrt{|\F|}}{3}\left/\sqrt{\max_{x\in G}\sum_{f\in \F}\min\left(\left|f(x)-f_*(x)\right|^2,4\right)}\right..$$
\end{restatable}

Our result can be intuitively applied to prove lower bounds for gradient computation and other \emph{natural problems} as well. In fact our technique has already been successfully applied to prove lower bounds for quantum SDP-solvers by van Apeldoorn and Gilyén~\cite{van2018quantum}.

Now we show how to prove our lower bound on gradient computation and present the result in the following theorem, which is an informal version of Theorem~\ref{thm:queryLowerBound}. This lower bound shows that our gradient-computation algorithm is in fact optimal up to poly-logarithmic factors for a specific class of smooth functions.
    
\begin{theorem}[Informal]\label{thm:informalLowerBoundFamily}
	Let $\eps, d> 0$. There exists a family of smooth functions $\F\subseteq \{f:\R^d\rightarrow \R\}$ such that the following holds. Every 
    quantum algorithm $\A$ that makes $T$ queries to the phase oracle $\mathrm{O}_{\!f}$ and satisfies the following:  
	for every $f\in \F$, $\A$  outputs, with probability $\geq 2/3$, an approximate gradient $\pmb{g}\in\R^d$  satisfying
	$$
    	\nrm{\pmb{g}-\nabla f(\pmb{0})}_\infty< \eps,
    $$ 
	needs to make $T=\Omega( \sqrt{d}/\eps)$ queries.
\end{theorem}
\begin{proof}(sketch)
	We exhibit a family of functions $\F$ for which the corresponding phase oracles $\{\mathrm{O}_f:f\in \F\}$ require $\Omega(\sqrt{d}/\eps)$ queries to distinguish them from the constant $0$ function (as shown by Theorem~\ref{thm:arbHybLow}), but the functions in $\F$ can be uniquely identified by calculating their gradient at $\pmb{0}$ with accuracy $\eps$. In particular, this implies that calculating an approximation of the gradient vector for these functions must be at least as hard as distinguishing the phase oracles corresponding to functions in $\F$. 
	We use the following $\R^d\rightarrow \R$ functions: 
	$f_*(\pmb{x}):=0$ and $f_j(\pmb{x}):=2\eps x_j e^{-\nrm{\pmb{x}}^2/2}$ for all $j\in[d]$, and consider the family of functions $\F:=\bigcup_{j\in[d]}\{f_j(\pmb{x})\}$. As we show in Lemma~\ref{lemma:expFunFamily}, for all $\pmb{x}\in \R^d$ we have that 
   	\begin{equation*}
	\sum_{j\in [d]} \left|f_j(\pmb{x})-f_*(\pmb{x})\right|^2\leq\frac{4\eps^2}{e}.
	\end{equation*}
\end{proof}

Using our \emph{efficient} oracle-conversion technique between probability oracles and phase oracles, which incurs an $\bigOt{1}$ overhead, the above lower bound implies an $\widetilde{\Omega}(\sqrt{d}/\eps)$ query lower bound on $\eps$-accurate gradient computation for the probability oracle input model as well.

More recently, Cornelissen~\cite{CornelissenThesis} \anote{(also building on Theorem~\ref{thm:arbHybLow})} managed to show an $\Omega\left(d^{\frac{1}{2}+\frac{1}{p}}/\eps\right)$ lower bound for $\eps$-precise gradient computation in $p$-norm \emph{for every} $p\in [1,\infty]$. More precisely he showed that an algorithm as in Theorem~\ref{thm:informalLowerBoundFamily} that needs to output a gradient $\pmb{g}$ satisfying  $\nrm{\pmb{g}-\nabla f(\pmb{0})}_p< \eps$ for a given $p\in [1,\infty]$, must make $\Omega\left(d^{\frac{1}{2}+\frac{1}{p}}/\eps\right)$ queries. Furthermore, the family of functions in his lower bound satisfies stronger smoothness criteria, making his lower bound even stronger. 
Note that this results shows that our algorithm is essentially optimal for a large class of gradient computation problems. Indeed, if one wants to compute an $\eps$-precise approximation of the gradient in $p$-norm, it suffices to compute an $\eps d^{-1/p}$-precise gradient in $\infty$-norm, which can be achieved with our algorithm making $\bigOt{d^{\frac{1}{2}+\frac{1}{p}}/\eps}$ queries.

\subsection{Significance of our improvement for applications}
Objective functions in continuous optimization problems are typically evaluated either via some arithmetic calculations or via some sampling procedure. 
In the former case the complexity of function evaluation usually has poly-logarithmic dependence on the precision, whereas in the latter case the complexity of the sampling procedure usually has polynomial dependence on the (reciprocal of the) precision. 

It is known that for functions that are evaluated arithmetically, the complexity of function evaluation and the complexity of gradient computation is typically the same up to a constant factor. This is called the ``cheap gradient principle''\footnote{This can be shown by a ``reverse derivation'' argument, which is basically a generalization of the back-propagation idea that is widely used in neural network training.} in the theory of algorithmic differentiation~\cite{evaluatingDerivatives}. Therefore in the arithmetic case Jordan's algorithm typically gives only a constant-factor quantum speedup. 

In the (quantum) sampling case, as we explained in Section~\ref{subsec:gradientIntro}, Jordan's algorithm can have a prohibitive overhead in the runtime due to the need for increased accuracy in function evaluation. This may explain why there have not been many applications of Jordan's original gradient computation algorithm~\cite{QuantGrad} despite its obvious potential for machine learning and other optimization tasks. Note that the dependence on the precision can be crucial, since in applications such as VQE or QAOA it can be natural to aim for precision $\propto 1/d$, where $d$ is the number of independent parameters.

As our lower bound shows, in the sampling case, it is impossible to obtain a super-quadratic speedup for computing an $\eps$-approximation of the gradient of a generic smooth function\footnote{In this statement we ignore polylog factors and consider the case when the gradient is approximated in the norm $\nrm{.}_\infty$. By generic smooth function we mean the class of functions satisfying the requirements of Theorem~\ref{thm:finalScaling}.}. Although the speedup is limited to quadratic in the sampling case, we emphasize that for real-word applications this case is the more relevant one; for arithmetically-evaluated functions classical gradient computation is already quite efficient, and there is typically little room for speedups! In fact, it seems extremely difficult to find an actual (non-query complexity) application, where Jordan's algorithm could result in an exponential speedup.  On the other hand, if one uses a quantum simulation algorithm or a heuristic quantum optimizer as a black-box unitary, then we are limited by quantum sampling, and our methods for gradient computation can result in a significant improvements over both  classical methods and Jordan's original algorithm. Therefore, our work is relevant for quantum optimization algorithms, some of which could potentially give exponential speedups over their classical counterpart.

\paragraph{Approximations in different norms.} Finally, let us compare the complexity of various gradient-computation methods for approximating the gradient in different norms. Note that if one wants to calculate an $\eps$-approximation of the gradient in standard $\ell_2$ norm $\nrm{\cdot}$ with an algorithm that has approximation guarantees with respect to the norm $\nrm{\cdot}_\infty$, then it suffices to take $\eps \rightarrow \eps/\sqrt{d}$ because $\nrm{\cdot}\le \sqrt{d} \nrm{\cdot}_\infty$.
Table~\ref{tab:complexities} shows the number of queries to the probability oracles of the different gradient-computation methods. 

\begin{table}[!ht]
	\centering
	\begin{tabular}{l|ccccc}
		& Classical & Semi-classical & Jordan's & \textbf{Our (smooth)} & \textbf{Our (degree-$k$)} \\ \hline
		$\eps\text{-apx. in } \nrm{\cdot}_\infty$	& $\bigOt{\frac{d}{\eps^2}}$	& $\bigOt{\frac{d}{\eps}}$	& $\bigOt{\frac{\sqrt{d}}{\eps^2}}$	& $\bigOt{\frac{\sqrt{d}}{\eps}}$	& $\bigOt{\frac{k}{\eps}}$	\\
		$\eps$-apx. in $\nrm{\cdot}$		& $\bigOt{\frac{d^2}{\eps^2}}$		& $\bigOt{\frac{d\sqrt{d}}{\eps}}$	& $\bigOt{\frac{d\sqrt{d}}{\eps^2}}$	& $\bigOt{\frac{d}{\eps}}$	& $\bigOt{\frac{k\sqrt{d}}{\eps}}$	
	\end{tabular}
	\caption{Quantum and classical query complexity bounds for gradient computation algorithms achieving $\eps$-precision in the $\nrm{\cdot}_{\infty}$ and $\nrm{\cdot}$ norms, given access to a probability oracle of a function, that is either smooth or a degree-$k$ multivariate polynomial. Our algorithm can take advantage of the polynomial structure, whereas it is not apparent for the other methods listed in the table.
		Note that our algorithm has the best scaling regarding both norms. In case of requiring $\eps$-approximation in the Euclidean norm, even the optimized semi-classical method has better scaling than Jordan's original algorithm.}
	\label{tab:complexities}	
\end{table}	

We also included in the table the complexity of an optimized semi-classical method to make a fair comparison to the quantum gradient algorithm. The semi-classical method uses quantum amplitude estimation to evaluate the function but then calculates the gradient classically coordinate-wise. It uses a high-order central difference formula to compute each partial derivative with $\bigO{\log(1/\eps)}$ function evaluations with roughly $\eps$-precision, similarly to the proof sketch of Theorem~\ref{thm:informalFinalScaling}. 

\subsection{Applications}
We consider three problems to which we apply our quantum gradient descent algorithm. We briefly describe below the problem of quantum variational eigensolvers (VQE)~\cite{peruzzo2014variational,wecker2015progress}, quantum approximate optimization algorithms (QAOA)~\cite{farhi2014qaoa}, and the quantum auto-encoding problem~\cite{wankwok:quantumneuralnetworks,romero:autoencoding}. In each case we show how our gradient computation algorithm can provide a quadratic speedup in terms of the dimension $d$ of the associated problem.

VQE is widely used to estimate the eigenvalue corresponding to some eigenstate of a Hamiltonian.  The main idea behind VQE is to begin with an efficiently parameterizable ansatz to the eigenstate.  For the example of ground state energy estimation, the ansatz state is often taken to be a unitary coupled cluster expansion.  The terms in that unitary coupled cluster expansion are then varied to provide the lowest energy for the groundstate, and the expected energy of the quantum state is mapped to the probability of some measurement outcome, making it accessible to our methods.

QAOA has a similar approach, the core idea of the algorithm is to consider a parametrized family of states such as $\ket{\psi(\pmb{x})} = \prod_{j=1}^d e^{-i x_j H_j} \ket{0}$.  The aim is to tune the parameters of the state in order to minimize some objective function, which can, e.g., represent some combinatorial optimization problem. In particular, if we let $O$ be a Hermitian operator corresponding to the objective function then we wish to find $\pmb{x}$ such that $\bra{\psi(\pmb{x})}H\ket{\psi(\pmb{x})}$ is minimized. For example, if we want to minimize the number of violated constraints of a constraint satisfaction problem, we can choose $O=\sum_{m=1}^M C_m$ to represent the number of violations:  $C_m$ is $1$ if and only if the $m^{\rm th}$ constraint is violated and $0$ otherwise~\cite{farhi2014qaoa}. After proper normalization and using some standard techniques we can map this expectation value to some measurement probability. Thus, from the perspective of our algorithm, QAOA looks exactly like VQE.

The classical auto-encoder paradigm~\cite{azoff:workhorsesneuralnetworks} is an important technique in machine learning, which is widely used for data compression. An auto-encoder is essentially a neural network architecture which is tuned for the following task: given a set of high-dimensional vectors, we would like to learn a low-dimensional representation of the vectors, so that computations on the original data set can be ``approximately'' carried out by working only with the low-dimensional representations. What makes auto-encoding powerful is that it does not assume any prior knowledge about the data set. This makes it a viable technique in machine learning, with various applications in natural language processing, training neural networks, object classification, prediction or extrapolation of information, etc. In this paper, we consider a natural quantum analogue (which was also considered before in the works of~\cite{wankwok:quantumneuralnetworks,romero:autoencoding}) of the auto-encoder paradigm, and show how to use our quantum gradient computation algorithm to quadratically speed up the training of quantum~autoencoders.

\section{Organization of the paper and preliminaries}
In Section~\ref{sec:genericoptimizationalgo}, we give a generic model of quantum optimization algorithms and a detailed description of the classical gradient descent algorithm. In Section~\ref{section:intercovertabilitybetweenoracles}, we describe how to convert a probability oracle to a phase oracle. 
In Section~\ref{section:quantumGradientCalculation} we present our quantum gradient computation algorithm and prove our main Theorem~\ref{thm:finalScaling} regarding its complexity
(we defer some of the technical calculations to Appendix~\ref{apx:centralErrorBounds}).
In Section~\ref{sec:lowerBound}, we present query lower bounds for algorithms that (approximately) compute the gradient of a function. In Section~\ref{sec:applications} we describe some applications. 
We finally conclude with some directions for future research in Section~\ref{section:openproblems}.

\paragraph{Notation.} Let $\pmb{e}_1, \pmb{e}_2, \ldots, \pmb{e}_d \in \R^d$ denote the standard basis vectors. We use bold letters for vectors $\pmb{x}\in\R^d$, in particular we use the notation $\pmb{0}$ for the $0$ vector, and $\pmb{1}$ for the all-$1$ vector $(\pmb{e}_1+ \pmb{e}_2+ \cdots+ \pmb{e}_d)$. By writing $\pmb{y}+r S$ we mean $\{\pmb{y}+r\pmb{v}:\pmb{v}\in S\}$ for a set of vectors $S\subseteq\R^d$, and use the same notation for sets of numbers. For $\pmb{x}\in \R^d$, let $\|\pmb{x}\|_\infty=\max_{i\in [d]} |x_i|$ and $\|\pmb{x}\|=(\sum_{i=1}^d x_i^2)^{1/2}$.  For $M\in \R^{d\times d}$ we use $\nrm{M}$ for denoting the operator norm of $M$.

For the set of numbers $\{1,2,\ldots,d\}$ we use the notation $[d]$. We use the convention $0^0=1$ throughout the paper, and use the notation $\N_0=\N\cup\{0\}$.

When we state the complexity of an algorithm, we use the notation $\bigOt{C}$ to hide poly-log factors in the complexity $C$. In general, we will use $\mathcal{H}$ to denote a finite dimensional Hilbert space. For the $n$-qubit all-$0$ basis state we use the notation $\ket{0}^{\!n}$, or when the value of $n$ is not important to explicitly indicate we simply write $\ket{\vec{0}}$.

\paragraph{Higher-order calculus.} Many technical lemmas in this paper will revolve around the use of higher-order calculus. We briefly introduce some notation here and give some basic definitions.         

	\begin{definition}[Index-sequences]
		For $k\in\N_0$ we call $\alpha\in[d]^k$ a $d$-dimensional length-$k$ index-sequence. For a vector $\pmb{r}\in\mathbb{R}^d$ we define $\pmb{r}^\alpha:=\prod_{j\in [k]} r_{\alpha_j}$.
		Also, for a $k$-times differentiable function, we define $\partial_\alpha f:=\partial_{\alpha_1}\partial_{\alpha_2}\cdots\partial_{\alpha_k}f$. Finally we denote define $|\alpha|=k$.
	\end{definition}
	\begin{definition}[Analytic function]
		We say that the function $f:\mathbb{R}^d\rightarrow \mathbb{R}$ is analytic if	
		\begin{equation}\label{eq:multidimensionalAnalytic}
			f(\pmb{x})=\sum_{k=0}^{\infty}\sum_{\alpha\in[d]^k}\pmb{x}^{\alpha}  \frac{\partial_{\alpha}f(\pmb{0})}{k!}.
		\end{equation}
	\end{definition}
\begin{definition}[Directional derivative]\label{def:dirDer}
		Suppose $f:\mathbb{R}^d\rightarrow \mathbb{R}$ is $k$-times differentiable at~$\pmb{x}\in\mathbb{R}^d$. We define the $k$-th order directional derivative in the direction $\pmb{r}\in\mathbb{R}^d$ using the derivative of a one-parameter function parametrized by $\tau\in\R$ along the ray in the direction of $\pmb{r}$:
		$$
		\partial_{\pmb{r}}^k f(\pmb{x})= \frac{d^k}{(d\tau)^k}f\left(\pmb{x}+\tau \pmb{r}\right).
		$$
\end{definition}
	
Observe that, using the definitions above, one has
\begin{align}\label{eq:rthderivateoff}
		\partial_{\pmb{r}}^k f= \sum_{\alpha\in[d]^k} \pmb{r}^\alpha \cdot \partial_\alpha f.
\end{align}
		In particular for every $i\in[d]$, we have that $\partial_{\pmb{e}_i}^k f=\partial_{i}^k f$.	

Central difference formulas (see, e.g.~\cite{LiGeneralNumDiff}) are often used to give precise approximations of derivatives of a function $h:\R\rightarrow \R$. These formulas are coming from polynomial interpolation, and yield precise approximations of directional derivatives too. Thus, we can use them to approximate the gradient of a high-dimensional function as shown in the following definition. 
	\begin{definition}\label{def:centralDiffernce}
		The degree-$2m$ \emph{central difference approximation} of a function $f:\R^d \rightarrow \R$ is:
		\begin{equation}\label{eq:centalDiffFormula}
		f_{(2m)}\!(\pmb{x}):=\sum_{\underset{\ell\neq 0}{\ell=-m}}^{m} \frac{ (-1)^{\ell-1}}{\ell}\frac{\binom{m}{|\ell|}}{\binom{m+|\ell|}{|\ell|}}f\left(\ell \pmb{x} \right) \approx\nabla f(\pmb{0})\cdot \pmb{x}.
		\end{equation}
		We denote the corresponding \emph{central difference coefficients} for $\ell\in \{-m,\ldots,m\}\backslash \{0\}$ by 
		\begin{equation*}
		a_{\ell}^{(2m)}:= \frac{ (-1)^{\ell-1}}{\ell}\frac{\binom{m}{|\ell|}}{\binom{m+|\ell|}{|\ell|}} \quad \text{ and } a_{0}^{(2m)}:=0
		\end{equation*}
	\end{definition}
In Appendix~\ref{apx:centralErrorBounds} we prove some bounds on the approximation error of the above formulas\footnote{One can read out the coefficients described in Definition~\ref{def:centralDiffernce} from the second row of the inverse of the Vandermonde matrix, as Arjan Cornelissen pointed out to us.} for generic $m$. 
Usually error bounds are only derived for some finite values of $m$, because that is sufficient in practice, but in order to prove our asymptotic results we need to derive more general results.

\section{A generic model of quantum optimization algorithms}\label{sec:genericoptimizationalgo}
	Most quantum algorithms designed for quantum optimization and machine learning procedures have the following core idea: 
	they approximate an optimal solution to a problem by tuning some parameters in a quantum circuit. The circuit usually consists of several simple gates, some of which have tunable real parameters, e.g., the angle of single qubit (controlled) rotation gates.
	Often, if there are enough tunable gates arranged in a nice topology, then there exists a set of parameters that induce a unitary capable of achieving a close to optimal solution. 
	
	In most optimization problems, one can decompose the circuit into three parts each having a different role (see Figure~\ref{fig:tunableCircuits}).
	The circuit starts with a state preparation part which prepares the initial quantum state
	relevant for the problem. We call this part `Prep.' in Figure~\ref{fig:tunableCircuits}. The middle part consist of both tunable parameters $x$ and fixed gates. The tunable and fixed gates are together referred to as `Tuned' in Figure~\ref{fig:tunableCircuits}. Finally, there is a verification circuit that evaluates the output state,
	and marks success if the auxiliary qubit is $\ket{1}$. We call the verification process $V$ in Figure~\ref{fig:tunableCircuits}. The quality of the circuit (for parameter $\pmb{x}$) is assessed by the
	probability of measuring the auxiliary qubit and obtaining~$1 $.
	
	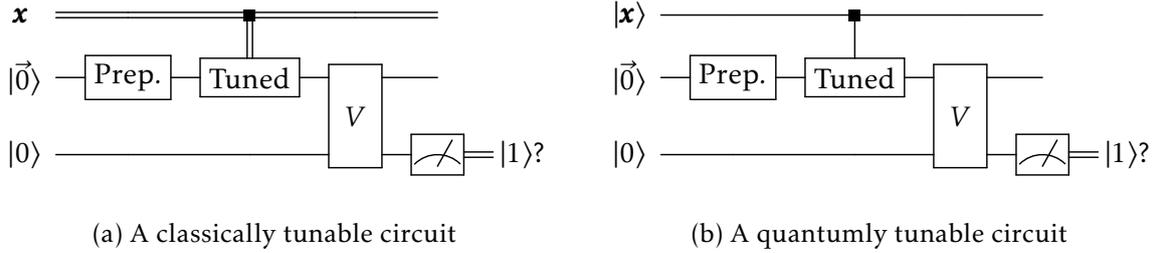
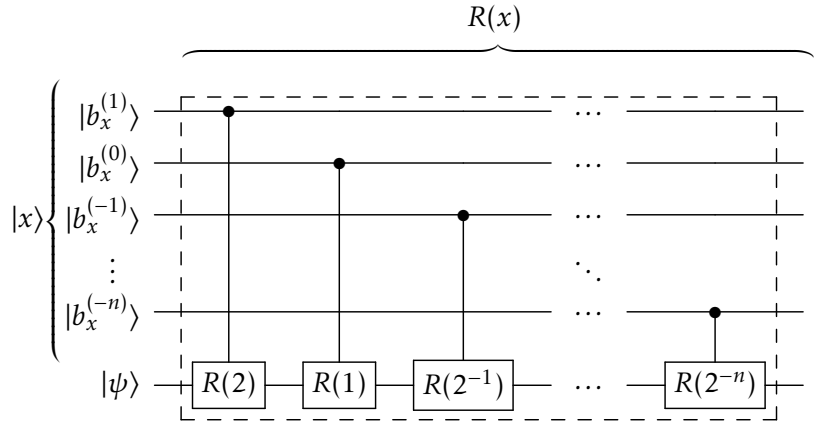
\begin{figure}[!ht]
		\centering
		\begin{subfigure}{.5\textwidth}
			\centering
			\[
			\Qcircuit @C=1.0em @R=1.2em {
				\lstick{\phantom{\ket{}}\pmb{x}\kern1.7mm} & \cw& \ctrlA \cw					& \cw				&\cw\\
				\lstick{\ket{\vec{0}}}	& \gate{\mathrm{Prep.}}	& \gate{\mathrm{Tuned}} \cwx	& \multigate{1}{V}  &\qw \\
				\lstick{\ket{0}}		& \qw 					& \qw							& \ghost{V}			&\meter	&\rstick{\kern-1mm\ket{1}?}\cw
			}
			\]
			\caption{A classically tunable circuit}
			\label{subfig:classicallyTuned}
		\end{subfigure}%
		\begin{subfigure}{.5\textwidth}
			\centering
			\[
			\Qcircuit @C=1.0em @R=1.2em {
				\lstick{\ket{\pmb{x}}}	& \qw					& \ctrlA \qw				& \qw 				&\qw 	& \\
				\lstick{\ket{\vec{0}}}	& \gate{\mathrm{Prep.}}	& \gate{\mathrm{Tuned}}\qwx	& \multigate{1}{V}  &\qw 	& \\
				\lstick{\ket{0}}		& \qw 					& \qw						& \ghost{V}			&\meter	&\rstick{\kern-1mm\ket{1}?}\cw
			}
			\label{Quantum optimization circuit}
			\]
			\caption{A quantumly tunable circuit}
			\label{subfig:quantumlyTuned}
		\end{subfigure}
		\begin{subfigure}{1\textwidth}
			\centering
			\vskip-10mm
			\begin{displaymath}
			\begin{matrix}
			\kern 55mm R(x)\kern -55mm\\ 
			\kern 21mm \overbrace{\kern 83mm}\kern -90mm\\ 
			\ket{x}\left\{\kern10mm\vphantom{\underbrace{\begin{matrix} a \\ a \\ a \\ a \\ a \\ a \end{matrix}}}\right.
			\end{matrix}%
			\begin{array}{l}
			\vphantom{\sum} \\
			\overset{\phantom{\sum_{\sum_{\sum_{\sum}^{\sum}}^{\sum_{\sum}^{\sum}}}^{\sum_{\sum}^{\sum^{\sum}}}}}{
			\Qcircuit @C=5mm @R=6mm {
				\lstick{\ket{b_x^{(1)}}}	& \ctrl{5}		& \qw			& \qw				&\qw&\hdots& &\qw				&\qw\\ 
				\lstick{\ket{b_x^{(0)}}}	& \qw			& \ctrl{4}		& \qw				&\qw&\hdots& &\qw				&\qw\\ 
				\lstick{\ket{b_x^{(-1)}}}	& \qw			& \qw			& \ctrl{3}			&\qw&\hdots& &\qw				&\qw\\ 
				\lstick{\vdots\kern 3mm}	& 		 		& 				& 					&	&\ddots& &					&	\\ 
				\lstick{\ket{b_x^{(-n)}}}	& \qw			& \qw			& \qw				&\qw&\hdots& &\ctrl{1}			&\qw\\
				\lstick{\ket{\psi}}			& \gate{R(2)}	& \gate{R(1)}	& \gate{R(2^{-1})}	&\qw&\hdots& &\gate{R(2^{-n})}	&\qw
				\gategroup{1}{2}{6}{8}{3mm}{--}
			}}\\
			\end{array}		
			\end{displaymath}			
			\caption{A $2^{-n}$ precisely tunable rotation gate $R(x)$ for the fixed point binary parameter $x=b_1b_0.b_{-1}\cdots b_{-n}$.}
			\label{subfig:tunableGate}
		\end{subfigure}%
		\caption{Two different approaches to tunable quantum optimization. The circuit on the top left has classically set parameters $\ket{\pmb{x}}$ (represented as a vector of fixed point binary numbers), 
		whereas the circuit on the top right has parameters $\pmb{x}$ described by an array of qubits $\ket{\pmb{x}}$.
		The black squares connected to the `Tuned' circuit indicate non-trivial control structure for which an example is presented on 
		the bottom figure, showing how to implement a quantumly tunable rotation gate built from simple controlled rotation gates.}
		\label{fig:tunableCircuits}
	\end{figure}
	
    One can think of the tunable circuit as being tuned in a classical way as shown in Figure~\ref{subfig:classicallyTuned} or a quantum way as in Figure~\ref{subfig:quantumlyTuned}. 
	In the classical case, the parameters can be thought of as being manually set. Alternatively, the parameters can be quantum variables represented by qubits. The advantage of the latter is that it allows us to use quantum techniques to speedup optimization algorithms. However the drawback is that it requires more qubits to represent the parameters and requires implementation of additional controlled-gates, see for e.g., Fig.~\ref{subfig:tunableGate}. 
	
	Let us denote by $U(\pmb{x})$ the circuit in Figure~\ref{subfig:classicallyTuned} and the corresponding circuit in Figure~\ref{subfig:quantumlyTuned} as  $U:=\sum_{\pmb{x}}\ketbra{\pmb{x}}{\pmb{x}}\otimes U(\pmb{x})$.
	The goal in these optimization problems is to find the optimal parameters (i.e., $\pmb{x}$) which maximizes the probability of obtaining $1$ after the final measurement, thereby solving the problem
	\begin{align} \label{mainoptimizationproblem}
		\argmax{\pmb{x}} \,\,p(\pmb{x}), \text{ where } p(\pmb{x})=\nrm{\left(I\otimes\ketbra{1}{1}\right)U(\pmb{x})\ket{\vec{0}}\ket{0}}^2.
	\end{align}
    
    A well-known technique to solve continuous-variable optimization problems like the one above is gradient-descent method. In practice, gradient-based techniques is one of the most commonly used algorithms for optimization.
    
    \subsection{Classical gradient ascent algorithm} \label{sec:introtoquantumgradientascent}
    As we discussed earlier, finding globally optimal parameters for optimization problems~\eqref{mainoptimizationproblem} is often hard. 
    Therefore in practice one usually relies on heuristics to find an approximate solution ${\pmb{x}_a}$ such that $p({\pmb{x}_a})$ is close to optimal.
	There are several heuristic optimization techniques that are often applied to handle such problems.
	One of the most common techniques is gradient ascent, which follows a greedy strategy to obtain the optimal solution. It simply follows the path of steepest ascent on the landscape of the objective function to find a solution that is, at least up to small local perturbations, optimal.  Such solutions are called locally optimal.
	In general, gradient-ascent-based algorithms start with a fixed or random setting of initial parameters,
	repeatedly calculate the gradient of the objective function, and takes a step in the direction of maximum ascent until the objective function has converged to a local maximum. If globally optimal parameters are required, a number of random restarts is usually taken and the maxima of all such restarts is reported by the algorithm.  We sketch such an algorithm below.

	\begin{algorithm}[!ht]
		\caption{Na\"{\i}ve gradient ascent using classical techniques} \label{alg:naiiveGradientAscent}
		\begin{algorithmic}[1] 
			\STATEx	{\bf Input:} A tunable circuit $U(\pmb{x})$ with $d$ real parameters. 
			\STATEx	~~~~~~~~~~ Parameters of the algorithm: $N,T,\eps,\delta,\eta$. 
			\STATEx	{\bf Output:} A vector of parameters $\pmb{x}\in\mathbb{R}^d$ such that $\|\nabla p(\pmb{x})\|\approx 0.$
			\STATEx	{\bf Init} $p_{\max}\leftarrow 0$; $\pmb{x}_{\max}\leftarrow \pmb{0}$
			\STATE	{\bf Repeat} $N$ times
			\STATE 	~~~ Choose a random vector of initial parameters $\pmb{x}_0\in\mathbb{R}^d$
			\STATE 	~~~ {\bf For} $t=1$ to $T$
			\STATE 	~~~~~~ Calculate the gradient $\nabla p(\pmb{x})$ as follows:
			\STATE 	~~~~~~ Estimate $p(\pmb{x})$ by the taking $\approx 1/\eps^2$ copies of $U(\pmb{x})\ket{\vec{0}}\ket{0}$ and measuring them.
            \STATE 	~~~~~~ {\bf For every} $i\in[d]$
			\STATE 	~~~~~~~~~ Estimate $p(\pmb{x}+\delta\cdot\pmb{e}_i)$ by the statistics of $\approx 1/\eps^2$ measurements 
			\STATE 	~~~~~~~~~ Estimate $\nabla_i p(\pmb{x})$ by the approximate formula $\left(p(\pmb{x}+\delta\cdot\pmb{e}_i)-p(\pmb{x})\right)/\delta$
			\STATE 	~~~~~~ {\bf Update} $\pmb{x}_t\leftarrow \pmb{x}_{t-1}+\eta \nabla p(\pmb{x})$
			\STATE 	~~~ {\bf If} $p_{\max} \leq p(\pmb{x}_T)$ {\bf then} $p_{\max}\leftarrow p(\pmb{x}_T)$; $\pmb{x}_{\max}\leftarrow \pmb{x}_T$			
			\STATE	{\bf Return} $\left(\pmb{x}_{\max}, p_{\max}\right)$ 
		\end{algorithmic}
	\end{algorithm}	

	Algorithm~\ref{alg:naiiveGradientAscent} gives a high-level overview of how this gradient ascent procedure works. Note that the algorithm above performed a fixed number of gradient steps, however, in general the gradient ascent algorithm continues until it has obtained a solution of ``good'' quality.
	 It is clear that Algorithm~\ref{alg:naiiveGradientAscent} uses the quantum circuit $U(\pmb{x})$ at most $\bigOt{NTd/\eps^2}$ times,
	where $N$ is the number of random initial configurations $\pmb{x}_0$ probed, $T$ is the number of gradient steps, $d$ is the number of parameters or dimension of $\pmb{x}$ and $\eps$ is the evaluation precision of the probability $p(\pmb{x})$ used in the algorithm.

\subsection{Quantum speedups to the classical algorithm}
	Now, let us consider the possible quantum speedups to this na\"{\i}ve gradient ascent algorithm discussed in the previous section. The most basic improvement which works even for classically controlled circuits (Figure~\ref{subfig:classicallyTuned}) is to estimate the probability $p(\pmb{x})$ in Step~5 using quantum amplitude estimation rather than doing repeated measurements and taking the average.
	If one wants to determine the value $p(\pmb{x})$ up to error~$\eps$ for some fixed $\pmb{x}$, the quantum approach uses the circuit $\bigO{1/\eps}$ times,
	whereas the classical statistical method would require $\Omega(1/\eps^2)$ repetitions, due to the additive property of variances of uncorrelated random variables. 
	Although this is a natural improvement, which does not require much additional quantum resources many papers that describe a similar procedure do not mention it. 
		
	Another quantum improvement can be achieved \cite{bulger2005BasinHop,lara2014HybridOpt} by using Grover search, which requires a quantumly controlled circuit like in Figure~\ref{subfig:quantumlyTuned}. 
	Let $P(\pmb{z})$ denote the probability that for a randomly chosen starting point $\pmb{x}_0$ we get $\pmb{x}_T=\pmb{z}$, i.e., we end up with $\pmb{z}$ after performing $T$ gradient steps. Let $\tilde{p}$ be a value such that $P(p(\pmb{z})\geq \tilde{p})\geq 1/N$. If we use $N$ randomly chosen initial points then with high probability at least one initial point will yield a point $\pmb{x}_T$ with $p(\pmb{x}_T)\geq \tilde{p}$.\footnote{I.e., with high probability, we will find a point from the top $1/N$ percentile of the points regarding the objective function $p(\pmb{z})$.}
	If we use the quantum maximum finding algorithm~\cite{durr&hoyer:minimum} or more precisely one if its generalizations \cite{NayakWu,van2017quantum}, 
	we can reduce the number of repetitions to $O(\sqrt{N})$ and still find an a point $\pmb{x}_T$ having $p(\pmb{x}_T)\geq \tilde{p}$ with high probability.
	Due to reversability, we need to maintain all points visited during the gradient ascent algorithm, thereby possibly introducing a significant overhead in the number of qubits used.
	
	However, there is a drawback using Grover search-based techniques. The disadvantage of quantum maximum finding approach over classical methods is the amount of time it takes to reach a local maximum using the gradient descent varies a lot.
	The reason is that classically, once we reached a local maximum we can start examining the next starting point, whereas if we use Grover search we do the gradient updates in superposition
	so we need to run the procedure for the largest possible number of gradient steps. To reduce this disadvantage one could use variable time amplitude amplification techniques introduced by Ambainis~\cite{varAmb}, however, we leave such investigations for future work.
	\anote{Question: Is variable time amplitude amplification compatibile with maximum finding?}
	
 \paragraph{Our contribution}  
	We show a quadratic speedup in $d$ -- the number of control parameters.
	For this we also need to use a quantumly controlled circuit, but the overhead in the number of qubits is much smaller than in the previous Grover type speedup. The underlying quantum technique crucially relies on the quantum Fourier transform as it is based on an improved version of Jordan's gradient computation~\cite{QuantGrad} algorithm. We can optionally combine this speedup with the above mentioned maximum finding improvement, which then gives a quantum algorithm that uses the quantumly controlled circuit (Figure~\ref{subfig:quantumlyTuned})
	$\bigOt{T\sqrt{Nd}/\eps}$ times, and achieves essentially the same guarantees as the classical Algorithm~\ref{alg:naiiveGradientAscent}.
	Therefore we can achieve a quadratic speedup in terms of all parameters except in $T$ and obtain an overall complexity of $\bigOt{T\sqrt{Nd/\eps}}$. For a summary of the speedups see Table~\ref{tab:variousQuantumImprovements}
	
\section{Interconvertibility of oracles}   \label{section:intercovertabilitybetweenoracles}

As we discussed in the previous section, the optimization problem associated with Figure~\ref{fig:tunableCircuits} was to maximize 
$$\max_{\pmb{x}} p(\pmb{x})
=\max_{\pmb{x}}\,\nrm{\left(I\otimes\ketbra{1}{1}\right)U(\pmb{x})\ket{\vec{0}}\ket{0}}^2.
$$
Typically, one should think of $U$ as the unitary corresponding to some parametrized quantum algorithm. Alternatively, we can view $U$ as a \emph{probability oracle} that maps $\ket{\pmb{x}}\ket{\vec{0}}\ket{0}$ to $\ket{\pmb{x}}\left(\sqrt{1-p(\pmb{x})}\ket{\Phi_0}\ket{0}+\sqrt{p(\pmb{x})}\ket{\Phi_1}\ket{1}\right)$, such that the probability of obtaining $1$ on measuring the last qubit is $p(x)$. This measurement probability serves as a ``benchmark score'' for the corresponding unitary $U$ with respect to the vector of parameters $\pmb{x}$.

This oracle model is not commonly used in quantum algorithms, therefore we need to convert it to a different format, so that we can other quantum techniques with this input. In the next subsection we describe different quantum input oracles, and later show how to efficiently convert between these oracles.

\subsection{Oracle access to the objective function}\label{sec:oracleaccessdesc} 
	
	In this subsection we describe our different input oracle models. In every case we will assume the function can be queried at a discrete set of points $X$. Later we will consider functions that act on $\R^d$, in which case we will usually choose $X$ to be finite a $d$-dimensional hypergrid around some point $\pmb{x}_0\in\R^d$. In some of our statements we will only say that we assume oracle access to a function $f:\R^d\rightarrow \R$ and do not specify $X$ for avoiding lengthy statements.

	\begin{definition}[Probability oracle]
		We say that $U_{p}:\mathcal{H}\otimes\mathcal{H}_{\text{aux.}}\rightarrow\mathcal{H}\otimes\mathcal{H}_{\text{aux.}}$ is a probability oracle for the function $p: X \rightarrow [0,1]$, if $\{\ket{x}:x\in X\}$ is an orthonormal basis of the Hilbert space $\mathcal{H}$, and for all $x\in X$ it acts as
		$$
		U_{p}:\ket{x}\ket{\vec{0}}\rightarrow \ket{x}\otimes\left(\sqrt{p(x)}\ket{\psi_{\rm good}(x)}\ket{1}  +\sqrt{1-p(x)}\ket{\psi_{\rm bad}(x)}\ket{0}\right),
		$$
		where $\ket{\psi_{\rm good}(x)}$ and $\ket{\psi_{\rm bad}(x)}$ are arbitrary (normalized) quantum states.
	\end{definition}
	Probability oracles are not commonly used in quantum algorithms; instead most algorithms use amplitude estimation\footnote{In some cases people use sampling and classical statistics to learn this probability, however amplitude estimation is quadratically more efficient. Typically one can improve sampling procedures quadratically using quantum techniques~\cite{montanaro:montecarlo}.} to turn these oracles into a \emph{binary oracle} that output a finite precision binary representation of the probability~\cite{wiebe2015quantum}.
	
	\begin{definition}[Binary oracle]\label{def:bin}
		For $\eta\in \R_+$, we say $B_{f}^\eta:\mathcal{H}\otimes\mathcal{H}_{\text{aux.}}\rightarrow\mathcal{H}\otimes\mathcal{H}_{\text{aux.}}$ is an $\eta$-accurate binary oracle for $f: X \rightarrow \R$, if $\{\ket{x}:x\in X\}$ is an orthonormal basis of the Hilbert space $\mathcal{H}$, and for all $x\in X$ it acts as
		$$B^{\eta}_p:\ket{x}\ket{\vec{0}}\rightarrow \ket{x}\ket{p'(x)},$$
		where $\ket{p'(x)}$ is a fixed-point binary number satisfying $|p'(x)-p(x)|\leq \eta$. We define the cost of one query to $B^{\eta}_f$ as $C(\eta)$.\footnote{The cost function would typically be $\polylog(1/\eta)$ for functions that can be calculated using a classical circuit, however, when the binary oracle is obtained via quantum phase estimation this cost is typically $1/\eta$.}
	\end{definition}
	
	But the conversion (from a probability oracle to binary oracle) has an exponential cost in the number of bits of precision, i.e., in order to obtain $\log(1/\varepsilon)$ bits of precision for the binary oracle, one needs to invoke the probability oracle $\bigO{1/\varepsilon}$ times.
	
	In practice, such binary oracles are often used in quantum algorithms strictly to provide the data needed to do a phase rotation.  This means that it is often simpler to consider our fundamental oracle access model to give the phase directly, i.e., \emph{phase oracle}, which encodes the probability in the phase (instead of explicitly outputing the probability like in the binary oracle).  Such an access model can be implemented using a binary access model, but we focus on phase oracles because they are a weaker oracle access model and can be much less expensive to implement.
	
	\begin{definition}[Phase oracle]
		\label{def:phaseoracle}
		We say that  $\mathrm{O}_{\!f}:\mathcal{H}\otimes\mathcal{H}_{\text{aux.}}\rightarrow\mathcal{H}\otimes\mathcal{H}_{\text{aux.}}$ is a phase oracle for $f: X \rightarrow [-1,1]$, if $\{\ket{x}:x\in X\}$ is an orthonormal basis of the Hilbert space $\mathcal{H}$, and for all $x\in X$ it acts as
		$$
		\mathrm{O}_{\!f}:\ket{x}\ket{\vec{0}}\rightarrow e^{if(x)}\ket{x}\ket{\vec{0}}.
		$$
	\end{definition}
	
	This model of phase oracle  is commonly used  in quantum information theory in particular in the field of quantum query complexity (see e.g. Grover's search~\cite{brassard2002quantum}).  Indeed, any oracle that implements a Boolean function can be cast as a phase oracle by diagonalizing it using a Hadamard transform on the output bit.  
	Hamiltonian simulation can also be thought of as a generalization of phase oracles since dynamical simulation can be thought of as a process that invokes a phase oracle in the eigenbasis of the Hamiltonian.
	Consequently, phase oracles naturally appear in algorithms that utilize quantum simulation (or continuous time quantum walks~\cite{childs2003exponential}) such as the quantum linear systems algorithm~\cite{harrow2009quantum,childs2015quantum} and related algorithms~\cite{wiebe2012quantum}.  

	For technical reasons we assume that we can perform fractional queries as well. In our case a fractional query will be almost as easy to implement as a full query. This is based on the observation that a fractional query for a probability oracle is trivial to implement, and since we simulate our phase queries by probability queries, we naturally get a fractional query implementation almost for free.
	We define the fractional query phase oracle below. 
	\begin{definition}[Fractional query oracle]\label{defn:fractionaloracle}
		Let $r\in[-1,1]$, we say that  $\mathrm{O}_{rf}:\mathcal{H}\otimes\mathcal{H}_{\text{aux.}}\rightarrow\mathcal{H}\otimes\mathcal{H}_{\text{aux.}}$ is a 
		\emph{fractional query}\footnote{
			Note that this fractional query is more general than the fractional query introduced by Cleve et al. \cite{cleveEfficientDiscContQuery}, because we have a continuous phase rather than discrete.
			Thus, the results of \cite{cleveEfficientDiscContQuery} do not give a way to implement a generic fractional query using a simple phase oracle $\mathrm{O}_{\!f}$,
			however one can use some techniques similar to our oracle conversion techniques in order to implement fractional queries~\cite{gilyenBlockMatrices}. } 
		phase oracle for $f: X \rightarrow [-1,1]$, if $\{\ket{x}:x\in X\}$ is an orthonormal basis of the Hilbert space $\mathcal{H}$, and for all $x\in X$ it acts as
		$$
		\mathrm{O}_{rf}:\ket{x}\ket{\vec{0}}\rightarrow e^{i r f(x)}\ket{x}\ket{\vec{0}}.
		$$		
	\end{definition}

\subsection{Conversion between probability and phase oracles.}
    One na\"{\i}ve way to implement a phase oracle is to convert the probability oracle to a binary oracle and then a phase oracle, but this feels morally wrong as this procedure is essentially an analogue$\rightarrow$digital$\rightarrow$analogue conversion. Instead, in this section we show how to convert a probability oracle to a phase oracle \emph{directly} using the so-called {\em Linear Combination of Unitaries} (LCU) technique~\cite{BerryChildsSim15}. To our knowledge this procedure is new. 
    
    Skipping the analogue$\rightarrow$digital$\rightarrow$analogue conversion makes the method conceptually simpler and can make the algorithm  more resource efficient in practice. 
    Our conversion introduces only logarithmic overhead in terms of the precision, which is probably the best we can hope for.
    As shown by several works \cite{BerryChildsSim15,chowdhury2016quantum,van2017quantum}, avoiding phase (amplitude) estimation by using LCU-based techniques can lead to significant speedups. 
    However, it turns out that for our application we do not gain too much compared to the na\"{\i}ve approach in terms of the asymptotic complexity.

	Now we start describing how to convert a probability oracle to a phase oracle.
	First, in order to implement this phase oracle, we need the following observation. Let
	\begin{equation}\label{eq:probOracleImage}
		\ket{\psi(x)}:= U_p\ket{x}\ket{0}^{\!\otimes n} =  \ket{x}\left(\sqrt{p(x)}\ket{\psi_{\rm good}(x)}\ket{1}  +\sqrt{1-p(x)}\ket{\psi_{\rm bad}(x)}\ket{0}\right).
	\end{equation}
	Let us define $\Pi_1:=\left(I\otimes\left(\ketbra{0}{0}^{\otimes n}\right)\right)$, $\Pi_2:=\left(I\otimes\left(I_{n-1}\otimes\ketbra{1}{1}\right)\right)$
	and
	\begin{equation}\label{eq:GUDef}
		G_U:=\left(2\Pi_1-I\right) U_p^\dagger\left(2\Pi_2-I\right) U_p,
	\end{equation}	
	which is a slightly modified version of the Grover operator.
	By the ``2D subspace lemma'' (Lemma~\ref{lemma:2D}, Appendix~\ref{apx:oblivious}) we know\footnote{
		\label{foot:2D}
		Using the notation of Lemma~\ref{lemma:2D} we have $\ket{\psi}:=\ket{x}\ket{0}^{\!\otimes n}$, $\ket{\phi}:=\ket{x}\ket{\psi_{\rm good}(x)}\ket{1}$ and $\ket{\phi^\perp}:=\ket{x}\ket{\psi_{\rm bad}(x)}\ket{0}$.
		The non-trivial assumption of Lemma~\ref{lemma:2D} that we need to satisfy is that $\Pi_1 U^\dagger_p \left(\sqrt{1-p(x)}\ket{\phi}-\sqrt{p(x)}\ket{\phi^\perp}\right)=0.$ 
		To show this, observe that $\mathrm{Im}(\Pi_1)=\mathrm{Span}\{\ket{y}\ket{0}^{\!\otimes n}:y\in X\}$, thus it is enough to show that $\forall y\in X$ we have 
		$\bra{y}\bra{0}^{\otimes n} U^\dagger_p \left(\sqrt{1-p(x)}\ket{\phi}-\sqrt{p(x)}\ket{\phi^\perp}\right)=0$. This holds as can be seen by writing out the state $U_p\ket{y}\ket{0}^{\otimes n}$ using~\eqref{eq:probOracleImage}.
	}
	that $\ket{x}\ket{0}^{\!\otimes n}$ lies in a two-dimensional invariant subspace, on which $G_U$ acts as a rotation operator\footnote{
		Lemma~\ref{lemma:2D} essentially states that $U_p\left(2\Pi_1-I\right) U_p^\dagger\left(2\Pi_2-I\right)$ acts a rotation operator on $U_p\ket{x}\ket{0}^{\otimes n}$, which is equivalent to saying that $G_U=\left(2\Pi_1-I\right) U_p^\dagger\left(2\Pi_2-I\right)U_p$ acts as a rotation operator on $\ket{x}\ket{0}^{\otimes n}$.
	} 
	with rotation angle $2\theta(x)$ (see Fig.~\ref{fig:GroverOperator}), where
	$$\theta(x)
	=\arcsin\left(\nrm{\left(I\otimes(I_{n-1}\otimes\ketbra{1}{1})\right)\ket{\psi(x)}}\right)
	=\arcsin\left(\sqrt{p(x)}\right).
	$$
	Therefore, $\ket{x}\ket{0}^{\!\otimes n}$ is a superposition of two eigenstates of $G_U$, with eigenvalues $e^{\pm2i\theta(x)}$. 
	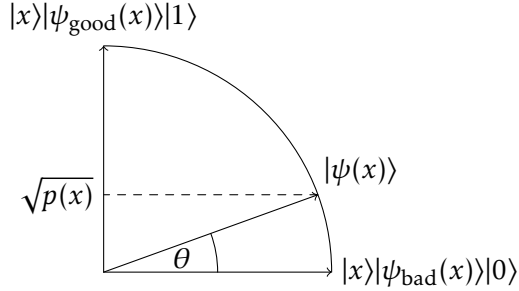
\begin{figure}[H]
		\centering
		\begin{tikzpicture}[scale=3,cap=round]
		\def\axisl{1.0}
		\def\angle{70}
		
		\colorlet{cstart}{magenta}
		\colorlet{chappy}{orange!80!black}
		\colorlet{cunhappy}{cyan}				
		\colorlet{anglecolor}{black!70}	
		
		\begin{scope}[]
		\draw (0,0) ++ (0:10mm) arc (0:90:10mm);			
		\node (psi) at ({\axisl*sin(\angle)},{\axisl*cos(\angle)}) {};	
		\draw[->] (0,0) -- (\axisl,0) node[right] {$\ket{x}\ket{\psi_{\rm bad}(x)}\ket{0}$};
		\draw[->] (0,0) -- (0,\axisl) node[above] {$\ket{x}\ket{\psi_{\rm good}(x)}\ket{1}$};
		\draw[->] (0,0) -- (psi.center) node[label={[label distance=-2mm]{90-\angle}:$\ket{\psi(x)}$}] {} ;	
		\draw (0,0) ++ ({90-\angle}:5mm) arc ({90-\angle}:0:5mm);	
		\draw (0,0) ++ ({(90-\angle)/2}:3.5mm) node {$\theta$};
		\end{scope}
		
		\node[left] (psis) at (0,{cos(\angle)}) {$\sqrt{p(x)}$};
		\node (psi0) at (1,1) {};
		\node (psi1) at (0,1) {};
		\node (psi0h) at (0,1) {};
		
		\draw[dashed]	(psi.center) --  (psis);			
		
		\end{tikzpicture}
		\caption{Geometric illustration of the parameters of the Grover operator $G_U$.}
		\label{fig:GroverOperator}
	\end{figure}

	We will also need the following special case of the LCU lemma before we describe a procedure to convert a probability oracle to a phase oracle.   
    
	\begin{lemma}[{\cite[Lemma 8]{childs2015quantum}}]\label{lemma:LCU}
		Let $M>0$ and $a=\lceil\log_2(2M+1)\rceil$. Let $\beta\in\mathbb{C}^{2M+1}$ satisfy $\nrm{\beta}_1=1$ and $T:\mathcal{H}\rightarrow\mathcal{H}$ be a unitary. Suppose we want to ``implement'' $L=\!\sum_{m=-M}^{M}\beta_{m}T^{m}$.
		Then, we can implement a circuit $C$ such that for all~$\ket{\psi}\in\mathcal{H}$:
		$$
			C: \ket{\psi}\ket{0}^{\!\otimes a}\rightarrow L\ket{\psi}\ket{0}^{\!\otimes a}+\ket{\Phi^\perp}, \text{ where } (I\otimes \ketbra{0}{0}^{\!\otimes a})\ket{\Phi^\perp}=0,
		$$
		using $M$ applications of (controlled) $T$ and $T^\dagger$ and $\bigO{M\log(M)}$ other two-qubit gates.
	\end{lemma}
	\begin{remark}\label{rem:garbageSpan}
		Note that from the proof of \cite[Lemma 8]{childs2015quantum} one can see that $$\ket{\Phi^\perp}\in \mathrm{Span}(T^m\ket{\psi}:m\in -M, \ldots, M) \otimes \mathbb{C}^{2^a}.$$
	\end{remark}
    We now prove the main theorem in this section. 
	\begin{theorem}\label{thm:phaseConv}
		Let $p: X \rightarrow [0,1]$, and suppose $U_p:\mathcal{H}\otimes\mathcal{H}_{\text{aux.}}\rightarrow\mathcal{H}\otimes\mathcal{H}_{\text{aux.}}$ is a probability oracle with an $n$-qubit auxiliary Hilbert space $\mathcal{H}_{\text{aux.}}=\mathbb{C}^{2^n}$.
		Let $\eps\in(0,1/3)$, then we can implement an $\eps$-approximate phase oracle $O$ such that for any phase oracle $\mathrm{O}_p$ and for all $\ket{\psi}\in \mathcal{H}$ 
		$$\nrm{O\ket{\psi}\ket{0}^{\!\otimes (n+a)}-\mathrm{O}_p\ket{\psi}\ket{0}^{\!\otimes (n+a)}}\leq \eps,$$ 
		using $\bigO{\log(1/\varepsilon)}$ applications of $U_p$ and~$U_p^\dagger$, with $a=\bigO{\log\log(1/\eps)}$.
	\end{theorem}
	\begin{proof}
		Our implementation will be based on using the Grover operator $G_U$ discussed before, see Figure~\ref{fig:GroverOperator}. First let us consider the image of $\ket{x}\ket{0}^{\!\otimes n}$, and focus on
		the two dimensional invariant subspace of $G_U$ spanned by $\ket{x}\ket{0}^{\!\otimes n}$ and its image under $G_U$.
		For simplicity initially we fix the value of $x$, and simply denote $p(x),\theta(x)$ by $p, \theta$ for conciseness.
		Similarly let us denote by $G$ the operator induced by $G_U$ on the aforementioned invariant subspace:
		$$G=\left(\begin{array}{cc}
			e^{2i\theta} & 0 \\
			0 & e^{-2i\theta}
			\end{array}\right)
			=e^{2iH}\text{ where } 
			H=\left(\begin{array}{cc}
			\theta & 0 \\
			0 & -\theta
			\end{array}\right)	.
		$$ 
		Recall that $p=\sin^2(\theta)$, therefore using the Taylor expansion of $e^{i\theta}$, we have
		\begin{align*}
			e^{ip}\cdot I &=e^{i\sin^2(\theta)}\cdot I
			=e^{i\sin^2(H)}
			=\sum_{k=0}^{\infty}\frac{i^k}{k!}\sin^{2k}(H).
		\end{align*}
		Using triangle inequality and a some simple calculations we can see that: for all $M\in\mathbb{N}_+$,
		\begin{equation}\label{eq:TaylorClose}
			\nrm{e^{ip}\cdot I-\sum_{k=0}^{M}\frac{i^k}{k!}\sin^{2k}(H)} 
			\leq  \sum_{k=M+1}^{\infty}\frac{1}{k!}
			< \frac{1}{M!}\sum_{\ell=1}^{\infty}2^{-\ell}
			= \frac{1}{M!}.
		\end{equation}
		We can use Stirling's approximation: 
		\begin{equation}\label{eq:StirlingBounds}
		\forall \ell\in\mathbb{N}_+: \sqrt{2\pi \ell}\left(\frac{\ell}{e}\right)^{\!\ell}\leq \ell!\leq e\sqrt{\ell}\left(\frac{\ell}{e}\right)^{\!\ell},
		\end{equation}
		to show that for all $\eps'\in(0,1/3)$ and $M\ge 2\ln(1/\eps')/\ln\ln(1/\eps')$ we have
		\begin{equation}
			1/M!\leq (e/M)^M\leq \eps'. \label{eq:Mbd}
		\end{equation}
		Finally we define $\beta\in\mathbb{C}^{2M+1}$ using the following calculation:
		\begin{align}
			\sum_{k=0}^{M}\frac{i^k}{k!}\sin^{2k}(H)
			&=\sum_{k=0}^{M}\frac{i^k}{k!}\left(\frac{e^{iH}-e^{-iH}}{2i}\right)^{2k} \tag{since $\sin(\theta)=(e^{i\theta}-e^{-i\theta})/(2i)$}\\
			&=\sum_{k=0}^{M}\frac{(-i)^k}{k!}\sum_{\ell=0}^{2k}\binom{2k}{\ell}\frac{(-1)^{\ell}}{2^{2k}}e^{2i(k-\ell)H}\label{eq:trianglePlace}\\
			&=\sum_{k=0}^{M}\frac{(-i)^k}{k!}\sum_{m=-k}^{k}\binom{2k}{k-m}\frac{(-1)^{k-m}}{2^{2k}}e^{2imH}\tag{let $\ell\leftarrow (k-m)$}\\	
			&=\sum_{m=-M}^{M}e^{2imH}\sum_{k=|m|}^{M}\binom{2k}{k-m}\frac{(-1)^m i^k}{k!2^{2k}}\nonumber\\
			&=\sum_{m=-M}^{M}G^m\sum_{k=|m|}^{M}\underset{\beta_m:=}{\underbrace{\binom{2k}{k-m}\frac{(-1)^m i^k}{k!2^{2k}}}}	\tag{since $G=e^{2iH}$ by definition}\\	
			&=:\sum_{m=-M}^{M}G^m\beta_{m}.	\label{eq:betaDef}		
		\end{align}
		Equations \eqref{eq:TaylorClose}, \eqref{eq:betaDef} and \eqref{eq:Mbd} give
		\begin{equation}\label{eq:LCUClose}
			\nrm{e^{ip}\cdot I-\sum_{m=-M}^{M}\beta_m G^m} \leq  \eps',
		\end{equation}
        and by \eqref{eq:trianglePlace}-\eqref{eq:betaDef} we can see (by following how scalar factors in \eqref{eq:trianglePlace} propagate to $\beta$) that
        $$
        \nrm{\beta}_1\leq \sum_{k=0}^{M}\frac{1}{k!}\sum_{\ell=0}^{2k}\binom{2k}{\ell}\frac{1}{2^{2k}}=  \sum_{k=0}^{M}\frac{1}{k!}\leq e,
        $$
		
		This representation makes it possible	to use\footnote{
			Note that in order to use the LCU Lemma~\ref{lemma:LCU} we actually would need to implement a controlled version of $G_U$. Fortunately this is easy to do: as can be seen from equation~\eqref{eq:GUDef}, it is enough to make the two reflection operators controlled.
		}
		the LCU Lemma~\ref{lemma:LCU}. By setting $a := \lceil\log_2(2M+1)\rceil+1$ (i.e., the number of auxiliary qubits), we can implement the unitary $V$ satisfying
		\begin{equation}\label{eq:Vdef}
		V\ket{x}\ket{0}^{\!\otimes n}\ket{0}^{\!\otimes a-1}=\sum_{m=-M}^M \frac{\beta_m}{\nrm{\beta}_1}G_U^m\ket{x}\ket{0}^{\!\otimes n}\ket{0}^{\!\otimes a-1} + \ket{\Phi^\perp}, 
		\text{ and } \left(I\otimes I_n\otimes \ketbra{0}{0}^{\!\otimes a-1}\right)\ket{\Phi^\perp}=0,
		\end{equation}
		where $\ket{\Phi^{\perp}}\in \mathrm{Span}(G_U^m\ket{x}\ket{0}^{\!\otimes n}:m\in -M, \ldots, M) \otimes \mathbb{C}^{2^{(a-1)}}$ is an unnormalized state vector.
		Moreover the implementation of $V$
		uses $\bigO{M}$ invocations of $G_U$ and $G_U^\dagger$ and $\bigO{M\log(M)}$ other two-qubit~gates.
        
        To conclude the proof, we first decrease the amplitude of success from $1/\|\beta_1\|\geq 1/e$ to $\sin(\pi/10)\le 1/e$, which we consequently amplify using $k=2$ oblivious amplitude amplification steps, see Corollary~\ref{cor:oblivious} in Appendix~\ref{apx:oblivious}. For this, let $R$ be a single-qubit unitary $$R:\ket{0}\rightarrow \sin(\pi/10)\nrm{\beta}_1 \ket{0} + \sqrt{1-\sin(\pi/10)^2\nrm{\beta}_1^2} \ket{1},$$ and let $V':=V\otimes R$.	
		Then, using Eq.~\eqref{eq:LCUClose}, it is easy to see that for all $x\in X$
		\begin{equation}\label{eq:2DPhase}
			\nrm{\left(I\otimes(I_n\otimes \ketbra{0}{0}^{\!\otimes a})\right)V'\ket{x}\ket{0}^{\!\otimes (n+a)}-\sin(\pi/10)e^{ip(x)}\ket{x}\ket{0}^{\!\otimes (n+a)}}\leq \sin(\pi/10)\eps'.
		\end{equation}
		Since for $x\neq x'\in X$ the vectors on the left had side of $\eqref{eq:2DPhase}$ are orthogonal to each other\footnote{
			By the definition of $V$ \eqref{eq:Vdef} we can see that $V'\ket{x}\ket{0}^{\!\otimes (n+a)}\in\mathrm{Span}(G_U^m\ket{x}\ket{0}^{\!\otimes n}:m\in -M, \ldots, M) \otimes \mathbb{C}^{2^a}$, so it is enough to show that for $x\neq x'$ and for all $k,k'\in\mathbb{Z}$ we have that $G_U^k\ket{x}\ket{0}^{\!\otimes n}$ is orthogonal to $G_U^{k'}\ket{x'}\ket{0}^{\!\otimes n}$. 
			This is equivalent to saying that $G_U^{k-k'}\ket{x}\ket{0}^{\!\otimes n}$ is orthogonal to $\ket{x'}$.
			As we already observed $\ket{x}\ket{0}^{\!\otimes n}$ lies in a two dimensional invariant subspace of $G_U$, therefore it is enough to show that $\bra{x'}\bra{0}^{\otimes n}G_U\ket{x}\ket{0}^{\!\otimes n}=0$, which follows from the observations of footnote~\ref{foot:2D}.
		}, we get that $V'$ satisfies for all $\ket{\psi}=\sum_{x\in X}c_k \ket{x}\in\mathcal{H}$ that
		\begin{equation}\label{eq:2DPhases}
		\nrm{\left(I\otimes(I_n\otimes \ketbra{0}{0}^{\!\otimes a})\right)V'\ket{\psi}\ket{0}^{\!\otimes (n+a)}-\sin(\pi/10)\mathrm{O}_p\ket{\psi}\ket{0}^{\!\otimes (n+a)}}\leq \sin(\pi/10)\eps'.
		\end{equation}	
		Let $\mathrm{O}$ denote the circuit that uses $k=2$ oblivious amplitude amplification steps on $V'$ using the projectors $\Pi_1:=I\otimes\ketbra{0}{0}^{\!\otimes (n+a)}, \Pi_2:=I\otimes I_n\otimes\ketbra{0}{0}^{\!\otimes a}$. By Corollary~\ref{cor:oblivious} we get that for all $\ket{\psi}\in\mathcal{H}$ $$\nrm{\mathrm{O}\ket{\psi}\ket{0}^{\!\otimes (n+a)}-\mathrm{O}_p\ket{\psi}\ket{0}^{\!\otimes (n+a)}}\leq 10\eps',$$ 
		therefore we can choose $\eps':=\eps/10$ to conclude the proof.
	\end{proof}		

	Now we show a corollary of the above result, which can be relevant for quantum distribution testing~\cite{LiWuEntropyQueryComp}.
	\begin{corollary}\label{cor:probHamSim}
		Suppose we are given access to some probability distribution via a quantum oracle: 
		$$ U:\ket{0}\ket{0} \rightarrow \sum_{x\in X} \sqrt{p(x)}\ket{x}\ket{\psi_x},$$ 
		then we can simulate a Hamiltonian corresponding to the probability distribution
		$$ 
		e^{it\sum_{x\in X}p(x)\ketbra{x}{x}}
		$$
		for time $t\in \R$ with precision $\eps$ making $\bigO{(|t|+1)\log(|t|/\eps)}$ queries to $U$.
	\end{corollary} 
	\begin{proof}
		First we show that we can use the oracle $U$ to construct a probability oracle $U_p$: 
		\begin{align*}
		U_p:\ket{x}\ket{0}\ket{0}\ket{0} \kern5mm
		&\overset{H}{\rightarrow}\kern12mm 
		\ket{x}\ket{0}\ket{+}\ket{0} \\
		&\overset{U}{\rightarrow}\kern12mm 
		\ket{x}\sum_{y\in X} \sqrt{p(y)}\ket{y}\ket{+}\ket{\psi_y}\\
		&\overset{\kern-12mm e^{\pi i \sum_{y\in X}\ket{y}\ket{y}\ketbra{1}{1}\bra{y}\bra{y}}\otimes I\kern-12mm}{\rightarrow} \kern12mm
		\ket{x}\left(\sqrt{p(x)}\ket{x}\ket{-}\ket{\psi_x} + \sum_{y\neq x} \sqrt{p(y)}\ket{y}\ket{+}\ket{\psi_y}\right)\\
		&\overset{H}{\rightarrow} \kern12mm 
		\ket{x}\left(\sqrt{p(x)}\ket{x}\ket{1}\ket{\psi_x} + \sum_{y\neq x} \sqrt{p(y)}\ket{y}\ket{0}\ket{\psi_y}\right)\\
		&\overset{\kern-12mm SWAP\kern-12mm}{\rightarrow} \kern12mm 
		\ket{x}\left(\sqrt{p(x)}\ket{x}\ket{\psi_x}\ket{1} + \sum_{y\neq x} \sqrt{p(y)}\ket{y}\ket{\psi_y}\ket{0}\right).	
		\end{align*}
		As Theorem~\ref{thm:phaseConv} shows we can simulate a fractional phase query $\mathrm{O}^r_p$ where $r:=t/\lceil|t|\rceil$ with precision $\eps/\lceil|t|\rceil$ making $\bigO{\log(t/\eps)}$ queries to $U_p$. Observe that $\lceil|t|\rceil$ consecutive applications of $\mathrm{O}^r_p$ give $\mathrm{O}^t_p$, which is exactly the Hamiltonian simulation unitary that we wanted to implement.
	\end{proof}

	Given the result of Theorem~\ref{thm:phaseConv}, one may wonder if we can go in the other direction, i.e., can we show how to convert a phase oracle to a probability oracle. Indeed, there exists an efficient procedure to implement this. We state the theorem below and provide a proof in Appendix~\ref{apx:oracleConversions}.
	
	\begin{restatable}{lemma}{phaseToProbability}\textbf{\emph{(Converting phase oracles to probability oracles)}}\label{lemma:phaseToProbability}
		Let $\eps,\delta\in(0,1/2)$, and suppose $p:X\rightarrow \left[\delta ,1-\delta\right]$. Suppose we have access to a phase oracle
		$\mathrm{O}_{p}$
		then using $\bigO{\log(1/\eps)/\delta}$ invocations of the (controlled) $\mathrm{O}_{p}$ and $\mathrm{O}_{p}^\dagger$~oracle, we can implement a probability oracle
		$$U_{p}:\ket{x}\ket{0}^{\!\otimes k}\ket{0}\rightarrow \ket{x}\otimes\left(\sqrt{p'(x)}\ket{0}^{\!\otimes k}\ket{0}+\sqrt{1-p'(x)}\ket{\Phi^\perp}\ket{1}\right),$$
		where $\left|\sqrt{p'(x)}-\sqrt{p(x)}\right|\leq \eps$ for every $x\in \R^n$.
	\end{restatable}
	
	This result together with Theorem~\ref{thm:phaseConv} shows how to convert phase oracles and probability oracles back and forth	with logarithmic overhead in the precision, assuming  that the probabilities are bounded away from $0$ and $1$, hence proving that these oracles are essentially equivalent.

	One might wonder why $p$ needs to be bounded away from $0$. The reason is that we actually lose some information when we convert from a probability oracle to phase oracle, 
	since the probability is the amplitude squared. One could also try to convert the amplitude to phase preventing this loss, but then one needs to be careful because the amplitude can be complex.
	Thus one needs to put the absolute value of the amplitude to the phase, but unfortunately the absolute value function is not smooth, and therefore it is not clear if one can apply LCU techniques. 
	However, in some cases one can implement an inner product oracle using the so-called Hadamard test (which we discuss in Section~\ref{sec:variational}), which is almost an amplitude oracle.

\subsection{Conversion between Phase oracle, Binary oracle and Fractional Phase oracle}

	In the previous section we have shown how to convert between probability and phase oracles with a logarithmic overhead, when the probability is bounded away from 0,1. In \cite{gilyenBlockMatrices} we show how to implement an $\eps$-precise fractional phase query with logarithmic number of phase queries, when the phase is bounded away from $-\pi,\pi$. Now we argue, that a phase oracle can be converted back and forth to a binary oracle with logarithmic overhead if we define the cost of precision to be reciprocal.

	One can implement the phase oracle $\mathrm{O}_{\!f}$ (in Eq.~\ref{eq:apxLin}) using a constant number of invocations of $B^{\eta}_f$, by using the phase-kickback trick as explained in~\cite{QuantGrad} or using standard phase preparation techniques. In order to implement the phase $2\pi Sf$ within error $\delta$, we need to set $\eta=\Theta(\delta/|S|)$, so the query complexity of simulating a call to $\mathrm{O}_{\!f}$ using calls to $B^{\eta}_f$ is $C(\bigO{|S|})$ for constant $\delta$, which is acceptable if the cost function $C$ is $\polylog$.     
	
	Finally note for completeness that a phase oracle can simulate a Binary oracle, with $m$ digits of precision and error probability at most $\eta$ using $\bigOt{2^m\log(1/\eta)}$ queries with the use of quantum phase estimation. However one needs to be careful about approximation errors; the difficulty is that the outcome of phase estimation is probabilistic, which can complicate the algorithm when we want to use the binary oracle in a coherent fashion.  

\section{Improved quantum gradient computation algorithm}\label{section:quantumGradientCalculation}

\subsection{Overview of Jordan's algorithm}

	Stephen Jordan constructed a surprisingly simple quantum algorithm~\cite{QuantGrad,bulger2005Gradient} that can approximately calculate the $d$-dimensional gradient of a function $f:\mathbb{R}^d\rightarrow \mathbb{R}$ with a \emph{single} evaluation of $f$.
	In contrast, using standard classical techniques, one would use~$d+1$ function evaluations to calculate the gradient at a point $\pmb{x}\in\mathbb{R}^d$: one can first evaluate~$f(\pmb{x})$ and then, for every $i\in[d]$, evaluate $f(\pmb{x}+\delta \pmb{e}_i)$ (for some $\delta>0$) to get an approximation of the gradient in direction $i$ using the standard formula
    $$
    \nabla_i f(\pmb{x})\approx \frac{f(\pmb{x}+\delta \pmb{e}_i)-f(\pmb{x})}{\delta}.
    $$
	
	The basic idea of Jordan's quantum algorithm~\cite{QuantGrad} is simple. First make two observations. Observe that if $f$ is twice differentiable at~$\pmb{x}$, then  $f(\pmb{x}+\pmb{\delta})=f(\pmb{x})+\nabla f \cdot \pmb{\delta} + \bigO{\|\pmb{\delta}\|^2}$,
	which in particular implies that for small~$\|\pmb{\delta}\|$, the function~$f$ is very close to being affine linear.
	The second observation is that, using the value of $f(\pmb{x}+\pmb{\delta})$, one can implement a phase oracle:
	\begin{equation}\label{eq:apxLin}
		\mathrm{O}_{2\pi S f}:\ket{\pmb{\delta}}\rightarrow e^{2\pi i S f(\pmb{x}+\pmb{\delta})}\ket{\pmb{\delta}}\approx e^{2\pi i S f(\pmb{x})}e^{2\pi i S \nabla f \cdot \pmb{\delta}}\ket{\pmb{\delta}}
	\end{equation}
	for a scaling factor $S>0$, where the approximation uses $f(\pmb{x}+\pmb{\delta}) \approx f(\pmb{x})+\nabla f \cdot \pmb{\delta}$ for small $\|\pmb{\delta}\|$. The role of  $S$ is to make make the phases appropriate for the final quantum Fourier transform. 
    
    \paragraph{Sketch of the algorithm.} Assume that all real vectors are expressed upto some finite amount of precision. In order to compute the gradient at $\pmb{x}$, the algorithm starts with a uniform superposition $\ket{\psi}=\frac{1}{\sqrt{|G^d_{\!\pmb{x}}|}}\sum_{\pmb{\delta}\in G^d_{\!\pmb{x}}}\ket{\pmb{\delta}}$ 
	over the points of a sufficiently small discretized $d$-dimensional grid $G^d_{\!\pmb{x}}$ around $\pmb{x}$, and applies the phase oracle $\mathrm{O}_{2\pi S f}$ (in Eq.~\ref{eq:apxLin}) to $\ket{\psi}$. Next, the inverse quantum Fourier transform is applied to the resulting state and each register is measured to obtain the gradient of $f$ at $\pmb{x}$ approximately. Due to approximate linearity of the phase, as described in Eq.~\eqref{eq:apxLin}, applying the Fourier transform will approximately give us the gradient. This algorithm uses $\mathrm{O}_{2\pi S f}$ once and Jordan showed how to implement $\mathrm{O}_{2\pi S f}$ using one sufficiently precise function evaluation.
	
	In order to improve the accuracy of the simple algorithm above, one could use some natural tricks. If $f$ is twice continuously differentiable, it is easy to see that the smaller the grid $G^d_{\!\pmb{x}}$ becomes, the closer the function gets to being linear. 
	This gives control over the precision of the algorithm, however if we ``zoom-in'' to the function using a smaller grid, the difference between nearby function values becomes smaller, making it harder the distinguish them and thus increasing the complexity of the algorithm proportionally.

	Also, it is well known that if the derivative is calculated based on the differences between the points $(f(x\!-\!\delta/2)$, $f(x\!+\!\delta/2))$ rather than $(f(x)$, $f(x+\delta))$,  
	then one gets a better approximation since the quadratic correction term cancels. To mimic this trick, Jordan chose a symmetric grid $G^d_{\!\pmb{x}}$ around~$\pmb{0}$. 
	
	\paragraph{Complexity of the algorithm}For Jordan's algorithm, it remains to pick the parameters of the grid and the constant $S$ in Eq.~\ref{eq:apxLin}. For simplicity,  assume that $\nrm{\nabla f(\pmb{x})}_\infty\leq 1$, and suppose we want to approximate $\nabla f(\pmb{x})$ coordinate-wise up to $\eps$ accuracy, with high success probability. 
	Under the assumption that ``the $2^{\text{nd}}$ partial derivatives of $f$ have a magnitude of approximately $D_2$'', Jordan argues\footnote{We specifically refer to equation (4) in \cite{QuantGrad} (equation (3) in the arXiv version), and the discussion afterwards. Note that our precision parameter $\eps$ corresponds to the uncertainty parameter $\sigma$ in \cite{QuantGrad}.} that choosing $G^d_{\!\pmb{x}}$ to be a $d$-dimensional hypercube 
	with edge length $\ell \approx \frac{\eps}{D_2\sqrt{d}}$ and with $N \approx \frac{1}{\eps}$ equally spaced grid points in each dimension, the quantum algorithm yields an $\eps$-approximate gradient by setting $S=\frac{N}{\ell}\approx \frac{D_2\sqrt{d}}{\eps^2}$.
	Moreover, since the Fourier transform is relatively insensitive to local phase errors it is sufficient to implement the phase $S f(\pmb{x}+\pmb{\delta})$ upto some constant, say $1\%$~accuracy.
	
	During the derivation of the above parameters Jordan makes the assumption, that the third and higher-order terms of the Taylor expansion of $f$ around $\pmb{x}$ are negligible, however it is not clear from his work~\cite{QuantGrad}, how to actually handle the case when they are non-negligible.
	This could be a cause of concern for the runtime analysis, since these higher-order terms potentially introduce a dependence on the dimension $d$.

	Finally, in order to assess the complexity of his algorithm, Jordan considers the Binary oracle input model of Definiton~\ref{def:bin}. This input model captures functions that are evaluated numerically using, say, an arithmetic circuit.
	Typically, the number of one and two-qubit gates needed to evaluate such functions up to $n$ digits precision is polynomial in $n$ and $d$. However, this input model does not fit the quantum optimization framework that we introduced in Section~\ref{sec:genericoptimizationalgo}.

    \paragraph{Our improvements} We  improve on the results of Jordan~\cite{QuantGrad} in a number of ways.
	Jordan~\cite{QuantGrad} argued that evaluating the function on a superposition of grid-points symmetrically arranged around $\pmb{0}$ is analogous to using a simple central difference formula. 
	We also place the grid symmetrically, but we realized that it is possible to directly use central difference formulas, which is the main idea behind our modified algorithm. 
    
    As discussed in Section~\ref{sec:oracleaccessdesc}, we realized that in applications of the gradient descent algorithm for optimization problems, it is  natural to assume access to a phase oracle $\mathrm{O}_{\!f}:\ket{\pmb{x}}\rightarrow e^{i f(\pmb{x})}\ket{\pmb{x}}$ (allowing fractional queries as well -- see Definition~\ref{defn:fractionaloracle}) instead of the Binary access oracle $B^\eta_f$. If we wanted to use Jordan's original algorithm in order to obtain the gradient with accuracy $\eps$, we need to implement the query oracle $\mathrm{O}_{\!f}^S$ by setting $S\approx D_2\sqrt{d}/\eps^2$, which can be achieved using $\lceil S\rceil$ consecutive (fractional) queries. Although it gives a square-root dependence on $d$ it scales as $\bigO{1/\varepsilon^2}$ with the precision. 
	In this work, we employ the phase oracle model and improve the quadratic dependence on $1/\eps$ to essentially linear.
	Additionally, we rigorously prove the square-root scaling with $d$ under reasonable assumptions on the derivatives of $f$. We also show that, for a class of smooth functions, the $\bigOt{\sqrt{d}/\eps}$-query complexity is optimal up to poly-logarithmic factors.
	We describe the algorithm in the next section, but first present our main result, whose proof is deferred to the end of this section.
    
\subsection{Analysis of Jordan's algorithm}\label{subsec:mainAlg}  
	In this section we describe Jordan's algorithm and provide a generic analysis of its behaviour. In the next subsection we combine these results with our finite difference methods. Before describing the algorithm, we introduce appropriate representation of our qubit strings suitable for fixed-point arithmetics.
	\begin{definition}\label{def:labelDefiniton}
		For every $b\in\{0,1\}^n$, let $j^{(b)}\in \{0,\ldots,2^n-1\}$ be the integer corresponding to the binary string $b=(b_1,\ldots,b_n)$.
		We label the $n$-qubit basis state $\ket{b_1}\ket{b_{2}}\cdots\ket{b_n}$ by $\ket{x^{(b)}}$, where 
		\begin{equation*}
			x^{(b)}=\frac{j^{(b)}}{2^n}-\frac{1}{2}+2^{-n-1}.
		\end{equation*}
		We denote the set of corresponding labels as $G_n:=\left\{\frac{j^{(b)}}{2^n}-\frac{1}{2}+2^{-n-1} : j^{(b)}\in \{0,\ldots,2^n-1\}  \right\}$. Note that there is a bijection between $\{j^{(b)}\}_{b\in \{0,1\}^n}$ and $\{x^{(b)}\}_{b\in \{0,1\}^n}$, so we will use $\ket{x^{(b)}}$ and $\ket{j^{(b)}}$ interchangeably in Remark~\ref{rem:FourierTranslation}. In the rest of this section we always label $n$-qubit basis states by elements of~$G_n$.
	\end{definition}
	\begin{definition}
		For $x\in G_n$ we define the Fourier transform of a state $\ket{x}$ as
		\begin{equation*}
			QFT_{G_n}: \ket{x}\rightarrow \frac{1}{\sqrt{2^n}}\sum_{k\in G_n}e^{2\pi i 2^n x k}\ket{k}.
		\end{equation*}		
	\end{definition}
    
    \begin{claim}\label{rem:FourierTranslation}
		This unitary is the same as the usual quantum Fourier transform up to conjugation with a tensor product of $n$ single-qubit unitaries.
	\end{claim}
	\begin{proof}
		For bitstrings $b,c\in\{0,1\}^n$, let $x^{(b)}\in G_n$ and $j^{(b)}\in \{0,\ldots,2^n-1\}$, be as defined in Definition~\ref{def:labelDefiniton}. 
		Then $QFT_{G_n}$ acts on $\ket{j^{(b)}}\equiv \ket{x^{(b)}}$ as
		\begin{align*}
		QFT_{G_n} :  \ket{x^{(b)}} &\rightarrow \frac{1}{\sqrt{2^n}}\sum_{x^{(c)}\in G_n}e^{2\pi i 2^n x^{(b)} x^{(c)}}\ket{x^{(c)}}\\
		&\equiv \frac{1}{\sqrt{2^n}}\sum_{j^{(c)}\in \{0,\ldots,2^n-1\}}e^{2\pi i 2^n \left(\frac{j^{(b)}}{2^n}-\frac{1}{2}+2^{-n-1}\right) \left(\frac{j^{(c)}}{2^n}-\frac{1}{2}+2^{-n-1}\right)}\ket{j^{(c)}}\\
		&\equiv \frac{1}{\sqrt{2^n}}\sum_{j^{(c)}\in \{0,\ldots,2^n-1\}}e^{2\pi i \left(\frac{j^{(b)} j^{(c)}}{2^n}-\left(j^{(b)}+j^{(c)}\right)\left(\frac{1}{2}+2^{-n-1}\right)+\left(2^{n-2}-\frac{1}{2}+2^{-n-2}\right)\right)}\ket{j^{(c)}}.
		\end{align*}	
		Using the usual quantum Fourier transform 
		$$
		QFT_n:\ket{j^{(b)}}\rightarrow  \frac{1}{\sqrt{2^n}}\sum_{j^{(c)}\in \{0,\ldots,2^n-1\}}e^{2\pi i 2^{-n} j^{(b)} j^{(c)}}\ket{j^{(c)}}
		$$ and the phase 
		unitary 
		$$
		U: \ket{j^{(b)}}\rightarrow e^{2\pi i \left(-j^{(b)}\left(\frac{1}{2}+2^{-n-1}\right)+\left(2^{n-2}-\frac{1}{2}+2^{-n-2}\right)/2\right)}\ket{j^{(b)}} \quad \text{ for every } j^{(b)}\in \{0,\ldots,2^n-1\},
		$$ it is easy to see that
		$$
		QFT_{G_n}=U \cdot QFT_n \cdot U.
		$$
		By writing $j^{(b)}$ in binary it is easy to see that $U$ is a tensor product of $n$ phase gates.
	\end{proof}		

	Now we are ready to precisely describe Jordan's quantum gradient computation algorithm.
	\begin{algorithm}[H]
	\caption{Jordan's quantum gradient computation algorithm}\label{alg:generic}
	\begin{algorithmic}[1]
		\STATEx {\bf Registers:} Use $n$-qubit input registers $\ket{x_1}\ket{x_2}\cdots\ket{x_d}$ with each qubit set to $\ket{0}$.
		\STATEx {\bf Labels:} Label the $n$-qubit states of each register with elements of $G_n$ as in Definition~\ref{def:labelDefiniton}.
		\STATEx {\bf Input:} A  function $f:G_n^d\rightarrow \mathbb{R}$ with phase-oracle $\mathrm{O}_{\!f}$ access such that 
		$$ 
			\mathrm{O}_{\!f}^{\pi2^{n+1}} \ket{x_1}\ket{x_2}\!\cdots\!\ket{x_d} = e^{2\pi i 2^n f\!\left(x_1,x_2,\ldots,x_d\right)} \ket{x_1}\ket{x_2}\!\cdots\!\ket{x_d}.
		$$
		\STATE {\bf Init} Apply a Hadamard transform to each qubit of the input registers.\label{line:initHigh}
		\STATE{\bf Oracle call} Apply the modified phase oracle $\mathrm{O}_{\!f}^{\pi2^{n+1}}$ on the input registers. \label{line:oracleCallHigh}
		\STATE{\bf QFT$_{G_n}^{-1}$} Fourier transform each register individually: 
        $$\ket{x}\rightarrow \frac{1}{\sqrt{2^n}}\sum_{k\in G_n}e^{\!-2\pi i 2^n x k}\ket{k}.
        $$\label{line:qftinverse}
		\STATE{\bf Measure} each input register $j$ and denote the measurement outcome by $k_j$.
		\STATE{\bf Output} $\left(k_1,k_2,\ldots,k_d\right)$ as the estimation for the gradient.
	\end{algorithmic}
	\end{algorithm}

	\begin{lemma}\label{lemma:genericJordan}
	Let $N=2^n$, $c\in \R$ and $\pmb{g}\in \R^d$ such that $\nrm{\pmb{g}}_\infty\leq 1/3$. If $f:G_n^d\to \R$ is such that 
	\begin{equation}\label{eq:closeFunctionApx}
		\left|f(\pmb{x})-\pmb{g}\cdot \pmb{x} -c \right|\leq \frac{1}{42\pi N},
	\end{equation}
	for all but a $1/1000$ fraction of the points $\pmb{x}\in G_n^d$,
	then the output of Algorithm~\ref{alg:generic} satisfies: 
    $$\Pr\left[|k_i-g_i|>\! 4/N\right]\leq 1/3 \quad \text{ for every  } i\in[d].$$
	\end{lemma}
	\begin{proof} First, note that $|G_n|=N$ from Definition~\ref{def:labelDefiniton}. Consider the following quantum  states
		$$   \qquad
		\ket{\phi}:=\frac{1}{\sqrt{N^d}}\sum_{\pmb{x}\in G_n^d}e^{2\pi i N f(\pmb{x})}\ket{\pmb{x}} \qquad \text{and} \qquad
		\ket{\psi}:=\frac{1}{\sqrt{N^d}}\sum_{\pmb{x}\in G_n^d}e^{2\pi i N  \left(\pmb{g} \cdot  \pmb{x}+ c\right)} \ket{\pmb{x}}.
		$$
		 Note that $\ket{\phi}$ is the state we obtain in Algorithm~\ref{alg:generic} after line~\ref{line:oracleCallHigh} 
		and $\ket{\psi}$ is its ``ideal version'' that we try to approximate with $\ket{\phi}$.
		Observe that the ``ideal'' $\ket{\psi}$ is actually a product state:
		\begin{align*}
		\ket{\psi}=
		\Big(\frac{1}{\sqrt{N}}\sum_{x_1\in G_n}e^{2\pi i N g_1 \cdot x_1}\ket{x_1}\Big) \otimes\cdots \otimes   \Big(\frac{1}{\sqrt{N}}\sum_{x_d\in G_n}e^{2\pi i N g_d \cdot  x_d}\ket{x_d}\Big).
		\end{align*}
       	
        It is easy to see that after applying the inverse Fourier transform to each register separately (as in line~\ref{line:qftinverse}) to $\ket{\psi}$, we obtain the state
		$$
		\Big(\frac{1}{N}\sum _{x_1,k_1\in G_n^2}e^{2\pi i Nx_1(g_1-k_1)}\ket{k_1}  \Big)\otimes \cdots \otimes \Big(\frac{1}{N}\sum _{x_d,k_d\in G_n^2}e^{2\pi i Nx_d(g_d-k_d)}\ket{k_d}  \Big).
		$$
       
		Suppose we make a measurement and observe $(k_1,\ldots,k_d)$. As shown in the analysis of phase estimation \cite{nielsen2002quantum}, we have the following\footnote{Note that our Fourier transform is slightly altered, but the same proof applies as in \cite[(5.34)]{nielsen2002quantum}. In fact this result can be directly translated to our case by considering the unitary conjugations proven in Remark~\ref{rem:FourierTranslation}.}: for every $i\in [d]$ (for a fixed accuracy parameter $\kappa>1$), the following holds:
		$$
		\Pr\left[|k_i-g_i|> \frac{\kappa}{N}\right]\leq \frac{1}{2(\kappa-1)} \quad \text{ for every } i\in [d].
		$$
		By fixing $\kappa=4$, we obtain the desired conclusion of the theorem, i.e., if we had access to $\ket{\psi}$ (instead of $\ket{\phi}$), then we would get a $4/N$-approximation of each coordinate of the gradient with probability at least $5/6$. It remains to show that this probability does not change more than $1/3-1/6=1/6$ if we apply the Fourier transform to $\ket{\phi}$ instead of $\ket{\psi}$.
		Observe that the difference in the probability of any measurement outcome on these states is bounded by twice the trace distance between $\ket{\psi}$ and $\ket{\phi}$ which is 
		\begin{equation}\label{eq:trl2}
			\trnorm{\ketbra{\psi}{\psi}-\ketbra{\phi}{\phi}}=2\sqrt{1-\left|\braket{\psi}{\phi}\right|^2}\leq 2 \nrm{\ket{\psi}-\ket{\phi}}.
		\end{equation}
		Since the Fourier transform is unitary and does not change the Euclidean distance, it is sufficient to show that $\nrm{\ket{\psi}-\ket{\phi}}\leq 1/12$ in order to conclude the theorem. Let $S\subseteq G_n^d$ denote the set of points satisfying  Eq.~\eqref{eq:closeFunctionApx}. We conclude the proof of the theorem by showing $\nrm{\ket{\psi}-\ket{\phi}}^2\leq \left(1/12\right)^2$:
		\begin{align*}
			\nrm{\ket{\phi}\!-\!\ket{\psi}}^2\!
			&=\!\frac{1}{N^d}\sum_{\pmb{x}\in G_n^d}\left|e^{2\pi i N f(\pmb{x})}-e^{2\pi i N (\pmb{g}\cdot\pmb{x}+c)}\right|^2\\
			&=\!\frac{1}{N^d}\!\sum_{\pmb{x}\in S}\left|e^{2\pi i N f(\pmb{x})}-e^{2\pi i N (\pmb{g}\cdot\pmb{x}+c)}\right|^2
			\!\!+\!\frac{1}{N^d}\!\!\sum_{\pmb{x}\in G_n^d\setminus S}\!\left|e^{2\pi i N f(\pmb{x})}-e^{2\pi i N (\pmb{g}\cdot\pmb{x}+c)}\right|^2\\
			&\leq \!\frac{1}{N^d}\!\sum_{\pmb{x}\in S}\left|2\pi N f(\pmb{x})-2\pi N (\pmb{g}\cdot\pmb{x}+c)\right|^2
			\!\!+\!\frac{1}{N^d}\!\!\sum_{\pmb{x}\in G_n^d\setminus S}\!4 \tag{$|e^{iz}-e^{iy}|\leq |z-y|$}\\
			&=\!\frac{1}{N^d}\!\sum_{\pmb{x}\in S}(2\pi N )^2\left|f(\pmb{x})-(\pmb{g}\cdot\pmb{x}+c)\right|^2
			\!\!+4\frac{|G_n^d\setminus S|}{N^d} \\
			&\leq\!\frac{1}{N^d}\!\sum_{\pmb{x}\in S}\left(\frac{1}{21}\right)^{\!\!2}+\frac{4}{1000} \tag{by the assumptions of the theorem}\\	
			&\leq\frac{1}{441}+\frac{1}{250}<\frac{1}{144}=\left(\frac{1}{12}\right)^{\!\!2}.
		\end{align*}
		\vskip-4mm
	\end{proof}
	   
	In the following theorem we assume that we have access to (a high power of) a phase oracle of a function $f$ that is very well approximated by an affine linear function $\pmb{g}\cdot \pmb{z}+ c$ on a hypergrid with edge-length $r\in\R$ around some $\pmb{y}\in\R^d$. We show that if the relative precision of the approximation is precise enough, then Algorithm~\ref{alg:generic} can compute an approximation of $\pmb{g}$ (the ``gradient'') with small query and gate complexity.    
    
	\begin{theorem}\label{thm:genericJordan}
		Let $c\in\R$, $r,\rho,\eps<M\in\R_+$, and $\pmb{y},\pmb{g}\in\R^d$ such that $\nrm{g}_\infty\leq M$.
		Let $n_\eps:=\left\lceil\log_2(4 /(r\eps)))\right\rceil$, $n_M:=\left\lceil\log_2(3rM)\right\rceil$ and $n:=n_\eps+n_M$.
		Suppose $f:\left(\pmb{y}+r G_n^d\right)\to \R$ is such that 
		\begin{equation*}
			\left|f(\pmb{y}+r \pmb{x})-\pmb{g}\cdot r\pmb{x}-c\right|\leq \frac{\eps r}{8\cdot 42\pi}
		\end{equation*}
		for all but a $1/1000$ fraction of the points $\pmb{x}\in G_n^d$.
		If we have access to a phase oracle $\mathrm{O}:\ket{\pmb{x}}\to e^{2\pi i 2^{n_\eps} f(\pmb{y}+r \pmb{x})}\ket{\pmb{x}}$ acting on $\mathcal{H}=\mathrm{Span}\{\ket{\pmb{x}}:\pmb{x}\in G_n^d\}$, then we can calculate a vector $\tilde{\pmb{g}}\in\R^d$ such that 
		$$\Pr\left[\,\nrm{\tilde{\pmb{g}}-\pmb{g}}_\infty>\! \eps\right]\leq \rho,$$
		with $\bigO{\log\big(\frac{d}{\rho}\big)}$ queries to $\mathrm{O}$ and with gate complexity $\bigO{d\log\big(\frac{d}{\rho}\big)\log\big(\!\frac{M}{\eps}\!\big)\log\log\big(\frac{d}{\rho}\big)\log\log\big(\!\frac{M}{\eps}\!\big)}$.
	\end{theorem}
	\begin{proof}
		Let $N_M:=2^{n_M}$, $N:=2^n$, and $h(\pmb{x}):=\frac{f(\pmb{y}+r \pmb{x})}{N_M}$,  then 
		$\big|h(x)-\pmb{g}\frac{r}{N_M}\pmb{x}-\frac{c}{N_M}\big|\leq \frac{\eps r}{8\cdot 42\pi N_M}\leq \frac{1}{42\pi N}$.
		Note that $\mathrm{O}=\mathrm{O}_h^{2\pi N}$, therefore Algorithm~\ref{alg:generic} yields and output $\tilde{\pmb{g}}$,
	 	which, as shown by Lemma~\ref{lemma:genericJordan}, is such that that for each $i\in [d]$ with probability at least $2/3$ we have $\big|\tilde{g}_i-\frac{r}{N_M} g_i\big|\leq \frac{4}{N}$. Thus also $\big|\frac{N_M}{r}\tilde{g}_i- g_i\big|\leq \frac{4N_M}{rN}\leq \eps$. By repeating the procedure $\bigO{\log(d/\rho)}$ times and taking the median coordinate-wise we get a vector $\tilde{\pmb{g}}_{\text{med}}$, such that $\nrm{\tilde{\pmb{g}}_{\text{med}}-\pmb{g}}_\infty\leq \eps$ with probability at least $(1-\rho)$.
	 	
	 	The gate complexity statement follows from the fact that the complexity of Algorithm~\ref{alg:generic} is dominated by that of the $d$ independent quantum Fourier transforms, each of which can be approximately implemented using $\bigO{n\log n}$ gates. We repeat the procedure $\bigO{\log(d/\rho)}$ times, which amounts to $\bigO{d\log(d/\rho)n\log n}$ gates. At the end we get $d$ groups of numbers each containing $\bigO{\log(d/\rho)}$ numbers with $n$ bits of precision. We can sort each group with a circuit having $\bigO{\log(d/\rho)\log\log(d/\rho)n\log n}$ gates.\footnote{Note that using the median of medians algorithm~\cite{BFPRT73} we could do this step with $\bigO{\log(d/\rho)n}$ time complexity, but this result probably does not apply to the circuit model, which is somewhat weaker than e.g. a Turing machine.} So the final gate complexity is $\bigO{d\log(d/\rho)\log\log(d/\rho)n\log n}$,
		which gives the stated gate complexity by observing that $n=\bigO{\log(M/\eps)}$.
	\end{proof}

\subsection{Improved quantum gradient algorithm using higher-degree methods}\label{subsec:finiteDiff}
	As Theorem~\ref{thm:genericJordan} shows Jordan's algorithm works well if the function is very close to linear function over a large hyprecube. However, in general even highly regular functions tend to quickly diverge from their linear approximations. To tackle this problem we borrow ideas from numerical analysis and use higher-degree finite-difference formulas to extend the range of approximate linearity. 
	
	We will apply Jordan's algorithm to the finite difference approximation of the gradient rather than the function itself. 
	We illustrate the main idea on a simple example: 
	suppose we want to calculate the gradient at $\pmb{0}$, then we could use the $2$-point approximation $\left(f(\pmb{x})-f(-\pmb{x})\right)/2$ instead of $f$, which has the advantage that it cancels out even order contributions.
	The corresponding phase oracle $\ket{\pmb{x}}\to e^{2^n\pi i (f(\pmb{x})-f(-\pmb{x})) }\ket{\pmb{x}}$ is also easy to implement as the product of the oracles:
	$$+ \text{ phase oracle: }\ket{\pmb{x}}=e^{2^n\pi i f\left(\pmb{x}\right) }\ket{\pmb{x}}\kern5mm\text{ and }\kern4mm
	  - \text{ phase oracle: }\ket{\pmb{x}}=e^{\!-2^n\pi i f\left(-\pmb{x}\right) }\ket{\pmb{x}}.$$

	Now we describe the general central difference approximation formula. There are a variety of other related formulas \cite{LiGeneralNumDiff}, but we stick to the central difference because the absolute values of the coefficients using this formula scale favorably with the approximation degree. 
	Since we only consider central differences, all our approximations have even degree, which is sufficient for our purposes as we are interested in the asymptotic scaling. 
	Nevertheless, it is not difficult to generalize our approach using other formulas \cite{LiGeneralNumDiff} that can provide odd-degree approximations as well.
    
    In the following lemma, we show that if $f:\R\rightarrow \R$ is $(2m+1)$-times continuously differentiable, then the central-difference formula in Eq.~\eqref{eq:centalDiffFormula} is a good approximation to $f'(0)$. Eventually we will generalize this to the setting where $f:\R^d\rightarrow \R$.
    
	\begin{restatable}{lemma}{oneDCentralDiff}\label{lemma:lagrangeBound}
		Let $\delta\in\R_+$, $m\in\N$ and suppose $f:[-m\delta,m\delta]\rightarrow\R$ is $(2m+1)$-times differentiable. Then
		\begin{equation} \label{eq:nthFiniteDiff}
		\left|f'(0)\delta -f_{(2m)}\!(\delta) \right|
		=\left|f'(0)\delta - \sum_{\ell=-m}^{m} a_{\ell}^{(2m)} f(\ell\delta)\right|
		\leq e^{-\frac{m}{2}} \nrm{f^{(2m+1)}}_\infty |\delta|^{2m+1},
		\end{equation}		
		where $\nrm{f^{(2m+1)}}_\infty:=\sup_{\xi\in[-m\delta,m\delta]}|f^{(2m+1)}(\xi)|$ and $a_{\ell}^{(2m)}$ is defined in Definition~\ref{def:centralDiffernce}. Moreover
		\begin{equation} \label{eq:diffCoeffs}
		\sum_{\ell=0}^{m} \left|a_{\ell}^{(2m)}\right| < \sum_{\ell=1}^{m} \frac{1}{\ell}\leq \ln(m)+1. 
		\end{equation}	
	\end{restatable}	
 	This lemma shows that for small enough $\delta$ the approximation error in \eqref{eq:centalDiffFormula} is upper bounded by a factor proportional to $\delta^{2m+1}$. If $\nrm{f^{(2m+1)}}_\infty\leq c^m$ for all $m$ and we choose $\delta\leq 1/c$, then the approximation error becomes exponentially small in $m$, motivating the use of higher-degree methods in our modified gradient computation algorithm. 
 	We generalize this statement to higher dimensions in Appendix~\ref{apx:centralErrorBounds}, 
 	which leads to our first result regarding our improved algorithm: (for the definition of the directional partial derivative $\partial_{\pmb{r}}^{2m+1} f$ see Definition~\ref{def:dirDer} )
 	
	\begin{restatable}{theorem}{lagrangeAlg}\label{thm:lagrangeAlg}
	Let $m\in\Z_+$, $R\in\R_+$ and $B\geq 0$. Suppose $f:[-R,R]^d\rightarrow \mathbb{R}$ is given with phase oracle access. If $f$ is $(2m+1)$-times differentiable and for all $\pmb{x}\in[-R,R]^d$ we have that
	$$
	|\partial_{\pmb{r}}^{2m+1} f(\pmb{x})|\leq B \quad \text{ for }\pmb{r}=\pmb{x}/\nrm{\pmb{x}},
	$$
	then we can compute an approximate gradient $\pmb{g}$ such that $\nrm{\pmb{g}-\nabla f(\pmb{0})}_\infty\leq \eps$ with probability at least $(1-\rho)$, using 
	$
	\bigO{\left(\max\left(\frac{\sqrt{d}}{\eps}\sqrt[2m]{\frac{B\sqrt{d}}{\eps}}, \frac{m}{\eps R}\right)\log(2m)+m\right)\log\left(\frac{d}{\rho}\right)}
	$
	 phase queries.
	\end{restatable}
	Suppose that $R=\Theta(1)$, and $f$ is a multi-variate polynomial of degree $k$. Then for $m=\lceil k/2\rceil$ we get that $B=0$, as can be seen by using \eqref{eq:rthderivateoff}, therefore the above result gives an $\bigOt{\frac{k}{\eps}\log\left(\frac{d}{\rho}\right)}$ query algorithm.
	If $2\leq k =\bigO{\log(d)}$, then this result gives an exponential speedup in terms of the dimension $d$ compared to Jordan's algorithm.
	For comparison note that other recent results concerning quantum gradient descent also work under the assumption that the objective function is a polynomial, for more details see Section~\ref{sec:priorwork}.
	
	However, as we argue in Appendix~\ref{apx:centralErrorBounds} for non-polynomial functions we can have $B\approx d^{m}$ even under strong regularity conditions. This then results in an $\bigOt{\frac{d}{\eps}}$ query algorithm, achieving the desired scaling in $\eps$ but failing to capture the sought $\sqrt{d}$ scaling. In order to tackle the non-polynomial case we need to introduce some smoothness conditions.

\subsection{Smoothness conditions and approximation error bounds}\label{subsec:smoothBounds}

	In this section we show how to improve the result of Theorem~\ref{thm:lagrangeAlg}, assuming some smoothness condition. The calculus gets a bit involved, because we need to handle higher-dimensional analysis. In order to focus on the main results, we keep this section concise and move the proof of some technical results to Appendix~\ref{apx:centralErrorBounds}. 
	We show that under reasonable smoothness assumptions, the complexity of our quantum algorithm is $\widetilde{O}(\sqrt{d}/\varepsilon)$ and in the next section show that for a specific class of smooth functions this is in fact optimal up to polylog factors.

   	We recommend the reader to take a look at the statements of lemmas presented in Appendix~\ref{apx:centralErrorBounds} to get a little more familiar with the main ideas behind the proofs. In Appendix~\ref{apx:centralErrorBounds} we prove the following result about analytic\footnote{The functions examined in the following two theorems are essentially Gevrey class $G^{\frac12}$ functions~\cite{GevreyClass}.} functions:
	\begin{restatable}{theorem}{analyticBound}\label{thm:analyticBound}
		If $R\in\R_+$, $f:\mathbb{R}^d\!\rightarrow \mathbb{R}$ is analytic and for all $k\in \mathbb{N}, \alpha\in[d]^k$ we have 
		$$
		|\partial_\alpha f(\pmb{0})|\leq c^k k^{\frac{k}{2}},
		$$
		then 
		$$
		|\nabla f(\pmb{0})\pmb{y}- f_{(2m)}(\pmb{y})|\leq \sum_{k=2m+1}^{\infty}\left(8 R c m \sqrt{d}\right)^{\!\!k},
		$$
		for all but a $1/1000$ fraction of points $\pmb{y}\in R\cdot G_n^d$.		
	\end{restatable}
	
	We can use this result to analyze the complexity of Algorithm~\ref{alg:generic} when applied to functions evaluated with using a central difference formula. In particular it makes it easy to prove the following theorem, which is one of our main results.
	\anote{State it more as a constructive result, rather than existential result.}
	\begin{theorem}\label{thm:finalScaling}
		Let $\pmb{x}\in \R^d$, $\eps\leq c\in\R_+$ be fixed constants and suppose $f\!:\mathbb{R}^d\!\rightarrow \mathbb{R}$ is analytic\footnote{
		For convenience we assume in the statement that $f$ can be evaluated at any point of~$\mathbb{R}^d$, but in fact we only evaluate it inside a finite ball around $\pmb{x}$. It is straightforward to translate the result when the function is only accessible on an open ball around $\pmb{x}$. However, a finite domain imposes restrictions to the evaluation points of the function. If $\pmb{x}$ lies too close to the boundary, this might impose additional scaling requirements and thus potentially increases the complexity of the derived algorithm. Fortunately in our applications it is natural to assume that $f$ can be evaluated at distant points too, so we don't need worry about this detail.
		} and satisfies the following: for every $ k\in\!\mathbb{N}$ and $\alpha\in [d]^k$
		$$
		|\partial_\alpha f(\pmb{x})|\leq c^k k^{\frac{k}{2}}.
		$$
		There is a quantum algorithm that works for all such functions, and outputs an $\eps$-approximate gradient $\tilde{\nabla}f(\pmb{x})\in\R^d$ such that 
		$$
		\nrm{\nabla f(\pmb{x})-\tilde{\nabla} f(\pmb{x})}_{\infty}\leq \eps,
		$$    
		with probability at least $1-\delta$, using $\bigOt{\frac{c\sqrt{d}}{\eps}\log\left(\frac{d}{\delta}\right)}$ queries to the phase oracle $\mathrm{O}_f$.
	\end{theorem}   
	\begin{proof}
		Let $g(\pmb{y}):=f(\pmb{x}+\pmb{y})$.
		By Theorem~\ref{thm:analyticBound} we know that for a uniformly random $\pmb{y}\in r\cdot G_n^d$ we have 
		$|\nabla g(\pmb{0})\pmb{y}- g_{(2m)}(\pmb{y})|\leq \sum_{k=2m+1}^{\infty}\left(8 r c m \sqrt{d}\right)^{\!\!k}$
		with probability at least $999/1000$. Now we choose $r$ such that this becomes smaller that $\frac{\eps r}{8\cdot 42 \pi}$.
		Let $r^{-1}:=9cm\sqrt{d}\Big( 81 \cdot 8 \cdot  42\pi cm\sqrt{d}/\eps\Big)^{1/(2m)}$, then we get $8 r c m \sqrt{d}=\frac{8}{9}\Big( 81 \cdot 8 \cdot  42\pi cm\sqrt{d}/\eps\Big)^{-1/(2m)}$ and so
		
        \begin{align*}
			\sum_{k=2m+1}^{\infty}\left(8 r c m \sqrt{d}\right)^{\!\!k}
			&=\left(8 r c m \sqrt{d}\right)^{\!\!2m+1}\sum_{k=0}^{\infty}\left(8 r c m \sqrt{d}\right)^{\!\!k}\\
			&\leq \frac{\eps}{81 \cdot 8 \cdot  42\pi cm\sqrt{d}}\Big( 81 \cdot 8 \cdot  42\pi cm\sqrt{d}/\eps\Big)^{\!\!\frac{-1}{2m}}\sum_{k=0}^{\infty}\left(\frac{8}{9}\right)^{\!\!k} \tag{by our choice of $r$}\\	
			&= \frac{\eps}{9cm\sqrt{d}\cdot 8 \cdot 42\pi }\Big( 81\cdot 8 \cdot 42\pi cm\sqrt{d}/\eps\Big)^{\!\!\frac{-1}{2m}} \tag{since $\sum_{k=0}^{\infty}\left(\frac{8}{9}\right)^{\!\!k}= 9$}\\ 
			&= \frac{\eps r}{8\cdot 42 \pi}.
		\end{align*}
		By Theorem~\ref{thm:genericJordan} we know that we can compute an approximate gradient with $\bigO{\log(d/\delta)}$ queries to $\mathrm{O}^S_{g_{(2m)}}$, where $S=\bigO{\frac{1}{\eps r}}$. Observe that 
		
		$$
		\mathrm{O}^S_{g_{(2m)}}\ket{\pmb{y}}=e^{i S g_{\!(2m)}\!\left(\pmb{y}\right)} \ket{\pmb{y}}=e^{i S \sum_{\ell=-m}^{m} a^{\!(2m)}_\ell g(\ell \pmb{y})} \ket{\pmb{y}}.
		$$
		Using the relation between $f$ and $g$, it is easy to see that the number of (fractional) phase queries to $\mathrm{O}_f$ we need in order to implement a modified oracle call $\mathrm{O}^S_{g_{(2m)}}$ is
		\begin{align}\label{eq:analyticCost}
		\sum_{\ell=-m}^{m}\left\lceil \left|a^{\!(2m)}_\ell\right| S\right\rceil
		&\leq 2m+ S\sum_{\ell=-m}^{m} a^{\!(2m)}_\ell 
		\overset{\eqref{eq:diffCoeffs}}{\leq} 2m+ S \left(2\log(m)+2\right) .
		\end{align}
		Thus $\mathrm{O}^S_{g_{(2m)}}$ can be implemented using $\bigO{\frac{\log(m)}{\eps r} + m}$ (fractional) queries to $\mathrm{O}_{\!f}$.
		By choosing $m=\log(c\sqrt{d}/\eps)$ the query complexity becomes\footnote{If we strengthen the $c^k k^{\frac{k}{2}}$ upper bound assumption on the derivatives to $c^k$, then we could improve the bound of Theorem~\ref{thm:analyticBound} by a factor of $k^{-k/2}$. Therefore in the definition of $R^{-1}$ we could replace $m$ by $\sqrt{m}$ which would quadratically improve the log factor in \eqref{eq:analyticCaseScaling}.} 
		\begin{equation}\label{eq:analyticCaseScaling}
			\bigO{\frac{c\sqrt{d}}{\eps}m\log(m)}=\bigO{\frac{c\sqrt{d}}{\eps}\log\left(\frac{c\sqrt{d}}{\eps}\right)\log\log\left(\frac{c\sqrt{d}}{\eps}\right)}.
		\end{equation}
	\end{proof}
	\indent The above achieves, up to logarithmic factors, the desired $1/\eps$ scaling in the precision parameter and also the $\sqrt{d}$ scaling with the dimension.
	This improves the results of \cite{QuantGrad} both quantitatively and qualitatively.
	
	We also show that the query complexity for this problem is almost optimal, by proving a lower bound in Section~\ref{sec:lowerBound} which matches the above upper bound up to log factors.
	
    \subsection{Most quantum optimization problems are ``smooth''}
	We now show that the condition on the derivatives in Theorem~\ref{thm:finalScaling} is fairly reasonable, i.e., a wide range of probability oracles that arise from quantum optimization problems satisfy this condition. In particular, consider the function $p:\R^d\rightarrow \R$ that we looked at (see Eq.~\eqref{mainoptimizationproblem}) during the discussion of a generic model of quantum optimization algorithms:
	\begin{equation*}
		p(\pmb{x})=\braketbra{\pmb{0}}{U(\pmb{x})^\dagger (\ketbra{1}{1}\otimes I) U(\pmb{x})}{\pmb{0}}.
	\end{equation*}
    We will now show that for every $k\in\mathbb{N}$ and index-sequence $\alpha\in[d]^k$, we have\footnote{This essentially means that the function $p(\pmb{x})$ in the probability oracle is in the Gevrey class $G^{0}$.} $|\partial_{\alpha} p(\pmb{x})|\leq 2^k$ when $U(\pmb{x})$ is a product of $d$ (controlled) rotations
	\begin{equation*}
		\mathrm{Rot}(x_j)= 
		 \left( {\begin{array}{cc}
		   \cos(x_j) & \sin(x_j) \\
		   -\sin(x_j) & \cos(x_j)
		  \end{array} } \right)
		  =\exp\left[ix_j \left( {\begin{array}{cc} 0 & -i \\ i & 0 \end{array} } \right)\right]
		  =e^{ix_j \sigma_y}
	\end{equation*}
and other fixed unitaries. In order to prove this, we first use Lemma~\ref{lemma:unitaryDerivative} to show that $\nrm{\partial_{\alpha} U(\pmb{x})}\leq 1$,
	which by Lemma~\ref{lemma:derivativeCombination} implies that $\nrm{\partial_{\alpha} \left(U(\pmb{x})^\dagger (\ketbra{1}{1}\otimes I) U(\pmb{x})\right)}\leq 2^k$, hence proving the claim.
	In fact, we prove slightly stronger statements, so that these lemmas can be used later in greater generality.
		
	\begin{lemma}\label{lemma:unitaryDerivative}
		Suppose $\gamma\geq0$ and
		\begin{equation*}
			U(\pmb{x})=U_0\prod_{j=1}^{d}\left(P_j \otimes e^{i x_j H_j}+(I-P_j)\otimes I\right) U_j,
		\end{equation*}
		where $\nrm{U_0}\leq 1$ and for every $j\in[d]$ we a have that $\nrm{U_j}\leq 1$, $P_j$ is an orthogonal projection and $H_j$ is Hermitian with $\nrm{H_j}\leq \gamma$ .
		Then for every $k\in\mathbb{N}$ and $\alpha\in[d]^k$, we have that $\nrm{\partial_\alpha U(\pmb{x})}\leq \gamma^k$.
	\end{lemma}	
\begin{proof}
		We have that 
		\begin{align*}
			\partial_\alpha U(\pmb{x})=U_0\prod_{j=1}^{d}\left(P_j \otimes \left(iH_j\right)^{|\{\ell\in[k]:\alpha_\ell=j\}|}e^{i x_j H_j}+0^{|\{\ell\in[k]:\alpha_\ell=j\}|}(I-P_j)\otimes I\right) U_j,
		\end{align*}		 
		therefore
		\begin{align*}
			\nrm{\partial_\alpha U(\pmb{x})}
			&=\nrm{U_0\prod_{j=1}^{d}\left(P_j \otimes \left(iH_j\right)^{|\{\ell\in[k]:\alpha_\ell=j\}|}e^{i x_j H_j}+0^{|\{\ell\in[k]:\alpha_\ell=j\}|}(I-P_j)\otimes I\right) U_j}\\
			&\leq\prod_{j=1}^{d}\nrm{\left(P_j \otimes \left(iH_j\right)^{|\{\ell\in[k]:\alpha_\ell=j\}|}e^{i x_j H_j}+0^{|\{\ell\in[k]:\alpha_\ell=j\}|}(I-P_j)\otimes I\right) }\\
			&\leq\prod_{j=1}^{d}\gamma^{|\{\ell\in[k]:\alpha_\ell=j\}|} = \gamma^k.
		\end{align*}
	\end{proof}
	\begin{lemma}\label{lemma:derivativeCombination}
		Suppose that $A(\pmb{x}), B(\pmb{x})$ are linear operators parametrized by $\pmb{x}\in\mathbb{R}^d$.
		If for all $k\in\N_0$ and $\alpha\in[d]^k$ we have that $\nrm{\partial_\alpha A}\leq \gamma^k $ and $\nrm{\partial_\alpha B}\leq \gamma^k $, then
		for all $k\in\N_0$ and $\alpha\in[d]^k$ we get that $\nrm{\partial_\alpha (AB)}\leq (2\gamma)^k $.
	\end{lemma}
	\begin{proof}
		For an index-sequence $\alpha=(\alpha_1,\alpha_2,\ldots,\alpha_k)\in[d]^k$ and a set $S=\{i_1<i_2< \ldots < i_\ell\}\subseteq [k]$ consisting of positions of the index-sequence, we define $\alpha_S:=\left(\alpha_{i_1},\alpha_{i_2},\ldots,\alpha_{i_\ell}\right)\in [d]^{|S|}$ 
		to be the index-sequence where we only keep indexes corresponding to positions in $S$; also let $\overline{S}:=[k]\setminus S$.
		It can be seen that 
		\begin{align*}
			\partial_{\alpha}(AB)=\sum_{S\subseteq [k]} \partial_{\alpha_S}A \partial_{\alpha_{\overline{S}}}B,
		\end{align*}
		therefore 
		\begin{align*}
			\nrm{\partial_{\alpha}(AB)}\leq\sum_{S\subseteq [k]} \nrm{\partial_{\alpha_S}A \partial_{\alpha_{\overline{S}}}B}\leq \sum_{S\subseteq [k]} \gamma^{|S|}\gamma^{k-|S|}=(2\gamma)^k.
		\end{align*}
	\end{proof}	
	\begin{remark}
		Finally note that the unitary in Lemma~\ref{lemma:unitaryDerivative} is an analytic function of its parameters, therefore the probability that we get by taking products of such unitaries and some fixed matrices/vectors is also an analytic function.
	\end{remark}
	
    \section{Lower bounds on gradient computation}\label{sec:lowerBound}
    
	In this section we prove that the number of phase oracle queries required to compute the gradient for some of the smooth functions satisfying the requirement of Theorem~\ref{thm:finalScaling} is $\Omega\left(\sqrt{d}/\eps\right)$, showing that Theorem~\ref{thm:finalScaling} is optimal up to log factors. As Lemma~\ref{lemma:phaseToProbability} shows, a probability oracle can be simulated using a logarithmic number of phase oracle queries, therefore this lower bound also translates to probability oracles.\footnote{A probability oracle can only represent functions which map to $[0,1]$, whereas the range of the function $f$ we use in the lower bound proof is an subinterval of $[-1,1]$. However, by using the transformed function $g:=(2+f)/4$ we get a function which has a range contained in $[1/2,3/4]$ so it can in principle be represented by a probability oracle. Moreover for a function with a range contained in $[1/2,3/4]$ we can efficiently convert between phase and probability oracles as shown by Theorem~\ref{thm:phaseConv}~and~Lemma~\ref{lemma:phaseToProbability}.}
	
	We first prove a useful Theorem providing a lower bound on the number of queries needed in order to distinguish a particular phase oracle from a family of other phase oracles. This result is of independent interest, and has already found an application for proving lower bounds on quantum SDP-solving~\cite{van2018quantum}.

	Then we prove our lower on gradient computation by constructing a family of functions which can be distinguished from the trivial phase oracle $I$ by $\eps$-precise gradient computation, but for which our hybrid-method based lower bound shows that distinguishing the phase oracles from $I$ requires $\Omega\left(\sqrt{d}/\eps\right)$ queries.
    
    \subsection{Hybrid method for arbitrary phase orcales}
    
    Now we turn to proving our general lower bound result based on the hybrid method, which was originally introduced for proving a lower bound for quantum search\footnote{One might wonder why do not we make a reduction to a search problem, e.g., by considering a function which has non-zero gradient only at some marked coordinates. We expect that this approach is not going to lead to a good lower bound, because the phase oracle is too strong, and by calculating the gradient one can actually solve exact counting, which problem has a linear lower bound for usual search oracles.} by Bennett et al.~\cite{Bennett:SearchLowerBound97}, and can be viewed as a simplified version of the adversary method \cite{Ambainis:AdversaryMethod00,Hoyer:AdversaryReview05,LMRSSz:conversion,belovsGeneralAdv15}.
    Our proof closely follows the presentation of the hybrid method in Montanaro's lecture notes \cite[Section 1]{Montanaro:LectureNotes11}.
    
    \hybMetArb*
   	\begin{proof}
		Suppose that $\F=\{f_j\colon j\in \{1,\ldots,d\}\}$, and let $f_0:=f_*$. Let $\A_j$ denote the algorithm which uses phase oracle $\mathrm{O}_{\!f_j}$ and let $\ket{\psi_j}:=\A_j\ket{\vec{0}}$ denote the state of the algorithm before the final measurement. 
		Since we can distinguish the states $\ket{\psi_0}$ and $\ket{\psi_j}$ with probability at least $2/3$, by the Holevo-Helstrom theorem~\cite[Chapter 3.1.1]{WatrousTQI}, it follows that
		$$
		\frac{2}{3}\leq \frac{1}{2} + \frac{1}{4}\nrm{\ketbra{\psi_0}{\psi_0}-\ketbra{\psi_j}{\psi_j}}_1.
		$$
		Since $\nrm{\ketbra{\psi_0}{\psi_0}-\ketbra{\psi_j}{\psi_j}}_1\leq 2 \nrm{\ket{\psi_0}-\ket{\psi_j}}$, we have that $1/3\leq \nrm{\ket{\psi_0}-\ket{\psi_j}}$.
		
		In general the quantum algorithm might use some workspace $\mathcal{W}=\underset{w\in W}{\mathrm{Span}}\{\ket{w}: w\in W\}$ along with the Hilbert space $\mathcal{H}$.
		To emphasize this we introduce the notation $\mathrm{O}'_{\!f}:=\mathrm{O}_{\!f}\otimes I_\mathcal{W}$, and $G':=G\times W$, so that the elements of $\mathcal{H}\otimes \mathcal{W}$ can be labeled by the elements of $G'$. 
		It is well known that in a quantum algorithm all measurements can be postponed to the end of the quantum circuit,
		so we can assume without loss of generality that between the queries the algorithm acts in a unitary fashion.
		Thus we can write $\A=U_T\mathrm{O}'_{\!f}U_{T-1}\mathrm{O}'_{\!f}U_{T-1}\cdots U_1\mathrm{O}'_{\!f}U_0$.
		
		Let us define for $t\in \{0,1,\ldots, T\}$
		$$\ket{\psi_j^{(t)}}:=\left(\prod_{\tau=1}^{t}U_\tau\mathrm{O}'_{\!f_j}\right)U_0\ket{\vec{0}},$$
		the state of algorithm $\A_j$ after making $t$ queries. We now prove by induction that for all $t\in \{0,1,\ldots, T\}$
		\begin{align}
		\label{eq:inductioninequality}
		\nrm{\ket{\psi_j^{(t)}}-\ket{\psi_0^{(t)}}}\leq \sum_{\tau=0}^{t-1} \nrm{(\mathrm{O}'_{\!f_j}-\mathrm{O}'_{\!f_0})\ket{\psi_0^{(\tau)}}}.
		\end{align}
		For $t=0$ the left-hand side is $0$, so the base case holds. Let us assume that \eqref{eq:inductioninequality} holds for $t-1$, we prove the inductive step as follows:
		\begin{align*}
		\nrm{\ket{\psi_j^{(t)}}-\ket{\psi_0^{(t)}}}
		&=\nrm{U_t\mathrm{O}'_{\!f_j} \ket{\psi_j^{(t-1)}}-U_t\mathrm{O}'_{\!f_0}\ket{\psi_0^{(t-1)}}}	\\
		&=\nrm{\mathrm{O}'_{\!f_j} \ket{\psi_j^{(t-1)}}-\mathrm{O}'_{\!f_0}\ket{\psi_0^{(t-1)}}}	 \tag{since norms are unitarily invariant}  	\\		
		&=\nrm{\mathrm{O}'_{\!f_j} \left(\ket{\psi_j^{(t-1)}}-\ket{\psi_0^{(t-1)}}+\ket{\psi_0^{(t-1)}}\right)-\mathrm{O}'_{\!f_0}\ket{\psi_0^{(t-1)}}}	   	\\	
		&\leq\nrm{\mathrm{O}'_{\!f_j} \left(\ket{\psi_j^{(t-1)}}-\ket{\psi_0^{(t-1)}}\right)} + \nrm{(\mathrm{O}'_{\!f_j} -\mathrm{O}'_{\!f_0})\ket{\psi_0^{(t-1)}}}	   	\tag{triangle inequality}\\	
		&=\nrm{\ket{\psi_j^{(t-1)}}-\ket{\psi_0^{(t-1)}}} + \nrm{(\mathrm{O}'_{\!f_j} -\mathrm{O}'_{\!f_0})\ket{\psi_0^{(t-1)}}}	   	\\	
		&\leq \sum_{\tau=0}^{t-1} \nrm{(\mathrm{O}'_{\!f_j}-\mathrm{O}'_{\!f_0})\ket{\psi_0^{(\tau)}}}. \tag{by the induction hypothesis}
		\end{align*}

		Since $\ket{\psi_j}=\ket{\psi_j^{(T)}}$, we additionally have that
		\begin{equation*}
		1/9\leq \nrm{\ket{\psi_0}-\ket{\psi_j}}^2 
		\leq \left(\sum_{\tau=0}^{T-1} \nrm{(\mathrm{O}'_{\!f_j}-\mathrm{O}'_{\!f_0})\ket{\psi_0^{(\tau)}}}\right)^{\!\!2}
		\leq T\sum_{\tau=0}^{T-1} \nrm{(\mathrm{O}'_{\!f_j}-\mathrm{O}'_{\!f_0})\ket{\psi_0^{(\tau)}}}^2,
		\end{equation*}
		where the last inequality uses the Cauchy-Schwarz inequality.
		Averaging the above inequality over the different oracles, we have 
		\begin{equation}\label{eq:TSquareGen}
		1/9\leq \frac{T}{d}\sum_{\tau=0}^{T-1}\sum_{j\in [d]} \nrm{(\mathrm{O}'_{\!f_j}-\mathrm{O}'_{\!f_0})\ket{\psi_0^{(\tau)}}}^2
		\leq \frac{T^2}{d}\max_{\tau}\sum_{j\in [d]} \nrm{(\mathrm{O}'_{\!f_j}-\mathrm{O}'_{\!f_0})\ket{\psi_0^{(\tau)}}}^2.
		\end{equation}
		We now upper bound the right-hand side of Eq.~\eqref{eq:TSquareGen} for an arbitrary pure state $\ket{\psi}$ to conclude the proof of the theorem.
		\begin{align*}
		\sum_{j\in [d]} \nrm{(\mathrm{O}'_{\!f_j}-\mathrm{O}'_{\!f_0})\ket{\psi}}^2
		&=\sum_{j\in [d]} \nrm{\left(\sum_{x\in G'}\ketbra{x}{x}\right)(\mathrm{O}'_{\!f_j}-\mathrm{O}'_{\!f_0})\left(\sum_{x'\in G'}\ketbra{x'}{x'}\right)\ket{\psi}}^2\\
		&=\sum_{j\in [d]} \nrm{\sum_{x\in G'}\ketbra{x}{x}(\mathrm{O}'_{\!f_j}-\mathrm{O}'_{\!f_0})\ketbra{x}{x}\ket{\psi}}^2 \tag{since $\braketbra{x}{\mathrm{O}'_{\!f_j}}{x'}=0$ for $x\neq x'$}\\			
		&=\sum_{x\in G'} \left|\braket{x}{\psi}\right|^2 \sum_{j\in [d]} \left|\bra{x}(\mathrm{O}'_{\!f_j}-\mathrm{O}'_{\!f_0})\ket{x}\right|^2\\					
		&\leq \max_{x\in G'} \sum_{j\in [d]} \left|\bra{x}(\mathrm{O}'_{\!f_j}-\mathrm{O}'_{\!f_0})\ket{x}\right|^2\\					
		&= \max_{x\in G} \sum_{j\in [d]} \left|e^{if_j(x)}-e^{if_0(x)}\right|^2 \tag{note the $G'\rightarrow G$ change}\\				
		&\leq \max_{x\in G} \sum_{j\in [d]} \min\left(\left|f_j(x)-f_0(x)\right|^2,4\right) \tag{since $|e^{iz}-e^{iy}|\leq \min\left(|z-y|,2\right)$}
		\end{align*}
		Combining this upper bound with Eq.~\eqref{eq:TSquareGen}, we have
		\begin{equation*}
		\frac{1}{9}\leq \frac{T^2}{d}\max_{x\in G} \sum_{j\in [d]} \min\left(\left|f_j(x)-f_0(x)\right|^2,4\right) ,
		\end{equation*}
		which in turn gives us the desired lower bound 
		\begin{equation*}
			T\geq \frac{\sqrt{d}}{3}\left/\sqrt{\max_{x\in G} \sum_{j\in [d]} \min\left(\left|f_j(x)-f_0(x)\right|^2,4\right)}\right..
		\end{equation*}
		Finally note that a controlled phase oracle is also a phase oracle, and the inverse oracles have the same operator distance as the non-inverted versions, therefore the above lower bound holds even if we allow controlled phase oracles or inverse oracle calls.
	\end{proof}

    \subsection{A family of functions for proving the gradient computation lower bound}
    
	Now we prove our lower on gradient computation by constructing a family of functions $\F$ for which the corresponding phase oracles $\{\mathrm{O}_f:f\in \F\}$ require $\Omega(\sqrt{d}/\eps)$ queries to distinguish them from the constant $0$ function (as shown by Theorem~\ref{thm:arbHybLow}), but the functions in $\F$ can be uniquely identified by calculating their gradient at $\pmb{0}$ with accuracy $\eps$. In particular, this implies that calculating an approximation of the gradient vector for these functions must be at least as hard as distinguishing the phase oracles corresponding to functions in $\F$.
    
    \begin{lemma} \label{lemma:expFunFamily}
    	Let $d\in \N$, $\eps,c\in\R_+$ and let us define the following $\R^d\rightarrow \R$ functions: 
    	$f_*(\pmb{x}):=0$ and $f_j(\pmb{x}):=2\eps x_j e^{-c^2\nrm{\pmb{x}}^2/2}$ for all $j\in[d]$.  Consider the family of functions $\F:=\bigcup_{j\in[d]}\{f_j(\pmb{x})\}$, then for all $\pmb{x}\in \R^d$ we have that
    	\begin{equation*}
    		\sum_{j\in [d]} \left|f_j(\pmb{x})-f_*(\pmb{x})\right|^2\leq\frac{4\eps^2}{ec^2}.
    	\end{equation*}
    \end{lemma}
	\begin{proof}
   		\begin{align*}
		\sum_{j\in [d]} \left|f_j(\pmb{x})-f_*(\pmb{x})\right|^2
		&= \sum_{j\in [d]} \left|2\eps x_j e^{-c^2 \nrm{\pmb{x}}^2/2}\right|^2 \\			
		&= 4\eps^2 \nrm{\pmb{x}}^2 e^{-c^2\nrm{\pmb{x}}^2}\\			
		&\leq \frac{4\eps^2}{ec^2}. \tag{using $ze^{-z}\leq 1/e$ with $z:=c^2\nrm{\pmb{x}}^2$}
		\end{align*}
	\end{proof}
    
    Now we prove bounds on the partial derivatives of the above functions to determine their smoothness.
    
    \begin{lemma}\label{lemma:normalDerivatives}
    	Let $d,k$ be positive integers, $c\in\mathbb{R}_+$ and $\pmb{x}\in\mathbb{R}^d$. Then, the function $f_j(\pmb{x}):=cx_j e^{-\frac{c^2\nrm{\pmb{x}}^2}{2}}$ 
    	satisfies the following: for every index-sequence $\alpha\in[d]^k$, the derivative of $f$ is bounded by $\left|\partial_{\alpha}f(\pmb{0})\right|\leq c^k k^{\frac{k}{2}}$. Moreover $\nabla f_j(\pmb{0})=c\pmb{e}_j$.
    \end{lemma}
    \begin{proof}
		Observe that
		\begin{equation}\label{eq:normalProduct}
			f(\pmb{x})= cx_j e^{-\frac{c^2 x_j^2}{2}}\prod_{i\neq j}^d e^{-\frac{c^2 x_i^2}{2}} ,
		\end{equation}
		and as one can see from the Taylor series $e^{-\frac{(cx)^2}{2}}=\sum_{\ell=0}^{\infty}\left(-\frac{1}{2}\right)^{\!\ell}\frac{(cx)^{2\ell}}{\ell!}$ we have for $k\geq 0$
		\begin{equation}\!\! 
		\left.\partial^k_{i}e^{-\frac{c^2 x_i^2}{2}}\right|_{x_i=0}=\left\{
			\begin{array}{ll}
				(-\frac{1}{2})^{\ell}c^{2\ell} \frac{(2\ell)!}{\ell!}	& k \!=\! 2\ell \\
							0											& k \!=\! 2\ell + 1
			\end{array}\right. , 
		\quad
		\left.\partial^k_{j} cx_j e^{-\frac{c^2 x_j^2}{2}}\right|_{x_j=0}=\left\{
			\begin{array}{ll}
				0															& k \!=\! 2\ell \\
				(-\frac{1}{2})^{\ell}c^{2\ell+1} \frac{(2\ell+1)!}{\ell!}	& k \!=\! 2\ell + 1
			\end{array}\right.\!\! . 
		\end{equation}
        Also observe that, for $\ell\geq 0$ 
		\begin{equation}\label{eq:normaloddDerivatives}
			\frac{(2\ell)!}{\ell!}\leq(2\ell)^{\ell} 
			\qquad\text{ and }\qquad 
			\left(\frac{1}{2}\right)^{\!\ell} \frac{(2\ell+1)!}{\ell!}\leq (2\ell+1)^{\ell+1/2}.
		\end{equation}
		The statements of the lemma follow by combining the results \eqref{eq:normalProduct}-\eqref{eq:normaloddDerivatives}.
    \end{proof}  
    
    Now use the above lemmas combined with the hybrid method Theorem~\ref{thm:arbHybLow} to prove our general lower bound result which implies the earlier stated informal Theorem~\ref{thm:informalLowerBoundFamily}.
    
    \begin{restatable}{theorem}{querylowerbound}\label{thm:queryLowerBound}
   		Let $\eps, c,d> 0$ such that $2\eps\leq c$ and for an arbitrary finite set $G\subseteq\R^d$ let
        $$\mathcal{H}=\underset{\pmb{x}\in G}{\mathrm{Span}}\{\ket{\pmb{x}}: \pmb{x}\in G\}.
        $$
   		Suppose $\A$ is a $T$-query quantum algorithm (assuming query access to phase oracle $\mathrm{O}_{\!f}\!:\ket{\pmb{x}}\!\rightarrow\! e^{if(\pmb{x})}$, acting on $\mathcal{H}$) for analytic functions $f:\mathbb{R}^d\rightarrow\mathbb{R}$ satisfying
        $$|\partial_{\alpha}f(\pmb{0})|\leq c^k k^{\frac{k}{2}} \quad \text{ for all } k\in\mathbb{N}, \alpha\in[d]^k,$$ 
        such that $\A$ computes an $\eps$-approximation $\tilde{\nabla} f(\pmb{0})$ of the gradient at $\pmb{0}$ 
        (i.e., $\nrm{\tilde{\nabla} f(\pmb{0})-\nabla f(\pmb{0})}_\infty\!< \eps$), 
        succeeding with probability at least~$2/3$. Then $T> \frac{c\sqrt{d}}{4\eps}$.
   	\end{restatable}
   	\begin{proof}
		Inspired by Lemma~\ref{lemma:expFunFamily}, we first define a set of ``hard'' functions, which we will use to prove our  lower bound
        Let $f_*:=f_0:=0$ and $f_j(\pmb{x}):=2\eps x_j e^{-c^2\nrm{\pmb{x}}^2/2}$ for all $j\in[d]$.  Consider the family of functions $\F:=\bigcup_{j\in[d]}\{f_j(\pmb{x})\}$.
   		By Lemma~\ref{lemma:normalDerivatives}, every $f\in \F$ satisfies $|\partial_{\alpha} f(\pmb{0})|\leq c^k k^{\frac{k}{2}}$ for all $k\in\mathbb{N}$ and $\alpha\in[d]^k$.
   		
   		Suppose we are given a phase oracle $\mathrm{O}_{\!f}$ acting on $\mathcal{H}$,
   		such that $\mathrm{O}_{\!f}=\mathrm{O}_{\!f_j}:\ket{\pmb{x}}\rightarrow e^{if_j(\pmb{x})}\ket{\pmb{x}}$ for some unknown $j\in \{0,\ldots,d\}$.
   		Since $\nabla f_0(\pmb{0})=\pmb{0}$ and $\nabla f_j(\pmb{0})=2\eps \pmb{e}_j$, using the $T$-query algorithm $\A$ in the theorem statement, one can determine the $j\in \{0,\ldots,d\}$ for which $f_j=f$ with success probability at least $2/3$. In particular we can distinguish the case $f=f_*$ from $f\in\F$, and thus by Theorem~\ref{thm:arbHybLow} and Lemma~\ref{lemma:expFunFamily} we get that
   		\begin{equation*}
	   		T\geq\sqrt{d}\frac{c}{\eps}\sqrt{\frac{e}{36}}>\frac{c\sqrt{d}}{4\eps}.
   		\end{equation*}
	\end{proof}
		
	\subsection{Lower bound for more regular functions}	\label{sec:regularFunctionsDiscussion}	
	Note that the functions for which we apply our results in this paper tend to satisfy a stronger condition than our lower bound example in Theorem~\ref{thm:queryLowerBound}.
	They usually satisfy\footnote{Without the $k^{\frac{k}{2}}$ factor -- i.e., they are of Gevrey class $G^0$ instead of $G^{\frac12}$.} $\left|\partial_{\alpha}f(\pmb{x}_0)\right|\leq c^k$. 
	We conjectured that the same lower bound holds for this subclass of functions as well.
	
	Very recently, Cornelissen~\cite{CornelissenThesis} managed to prove this conjecture\anote{(also building on Theorem~\ref{thm:arbHybLow})}.
	Moreover, he showed an $\Omega\left(d^{\frac{1}{2}+\frac{1}{p}}/\eps\right)$ lower bound for $\eps$-precise gradient computation in $p$-norm \emph{for every} $p\in [1,\infty]$. Note that these results shows that our algorithm is essentially optimal for a large class of gradient computation problems.

	Now we argue heuristically as to why Jordan's algorithm should not be able to calculate the gradient with significantly fewer queries for the above mentioned class of functions.	
	Algorithm~\ref{alg:generic} performs a Fourier transform, after applying a phase oracle that puts a phase $\sim \tilde{f}(\pmb{x})/\eps$ to the state $\pmb{x}\in G_n^d$.
	We can prove that for $n\geq \log(3c/\eps)$ the Fourier transform will provide the coordinates of $\nabla f(\pmb{0})$ up to error $\bigO{\eps}$ with high probability, 
	given that for most of the points $\pmb{x}\in G_n^d$ we have $|\tilde{f}(\pmb{x})/\eps-\nabla f(\pmb{0})\cdot \pmb{x}/\eps|\ll 1$. 
	Using a fractional phase oracle we can prepare phases of the form $\tilde{f}(\pmb{x})/\eps=\sum_{\pmb{y}\in S}\lambda_{\pmb{y}}f(\pmb{x})$ for some $S\subseteq G_n^d$,
	and $(\lambda_{\pmb{y}})\in[-1,1]^{|S|}$, where $S$ (and possibly $(\lambda_{\pmb{y}})$) might depend on $\pmb{x}$.
	The query complexity is thus driven by $|S|$. 
	
	Let us assume that $\nabla f(\pmb{0})= c\cdot\pmb{1}$, and observe that since $G_n^d$ is symmetric around $\pmb{0}$ we get that the typical value of $c\cdot\pmb{1}\cdot \pmb{x}=c\sum_{i=1}^{d}x_i$ is $\Theta(c\sqrt{d})$, 
	as can be seen for example by the central limit theorem. If we also assume that $|f|\leq 1$, then by the triangle inequality we see that $|S|=\Omega(\frac{c\sqrt{d}}{\eps})$.
	
	To conclude we still need to show that there exists an analytic function $f:\mathbb{R}^d\rightarrow \mathbb{R}$, which has $|f|\leq 1$ and $\nabla f(\pmb{0})=c\cdot \pmb{1}$ while it also satisfies for all $k\in\mathbb{N}$ and $\alpha\in[d]^k$ that $|\partial_\alpha f|\leq c^{k}$. At first sight this set of requirements seem slightly contradicting, but we found a very simple example of such a function:
	$$
		f(\pmb{x}):=\sin(cx_1+cx_2+\ldots+cx_d) \,\text{ for which }\, \partial_\alpha f(\pmb{x})=c^{|\alpha|}\sin^{(|\alpha|)}(cx_1+cx_2+\ldots+cx_d).
	$$
	The above argument also shows that placing the grid $G_n^d$ symmetrically around $\pmb{0}$ is of crucial importance. For example if the midpoint would be shifted by say $\delta\pmb{1}$, 
	then the typical magnitude of $\pmb{1}\cdot \pmb{x}=\pmb{1}\cdot \left(\delta\pmb{1}+\pmb{x}^{(\text{symm.})}\right)$ would be $\Theta(\delta d + \sqrt{d})$, 
	which would give rise to an elevated lower bound when $\delta\gg 1/\sqrt{d}$.

\section{Applications}\label{sec:applications}
In this section, we first consider variational quantum eigensolvers and QAOA algortihms, which can be treated essentially identically using our techniques. Then we consider the training of quantum autoencoders, which requires a slightly different formalism. We show that our gradient descent algorithms can be applied to these problems by reducing such problems to a probability maximization problem. For each application our quantum gradient computation algorithm yields a quadratic speedup in terms of the dimension. 

\subsection{Variational quantum eigensolvers}\label{sec:variational}
In recent years, variational quantum eigensolvers and QAOA~\cite{peruzzo2014variational,wecker2015progress,farhi2014qaoa} are favoured methods for providing low-depth quantum algorithms for solving important problems in quantum simulation and optimization. Current quantum computers are limited by decoherence, hence the option to solve optimization problems using  very short circuits can be enticing even if such algorithms are polynomially more expensive than alternative strategies that  could possibly require long gate sequences.  Since these methods are typically envisioned as being appropriate only for low-depth applications, comparably less attention is paid to the question of what their complexity would be, if they were executed on a fault-tolerant quantum computer. In this paper, we consider the case that these algorithms are in fact implemented on a fault-tolerant quantum computer and show that the gradient computation step in these algorithms can be performed quadratically faster compared to the earlier approaches that were tailored for pre-fault-tolerant applications. 

Variational quanutm eigensolvers (VQEs) are widely used to estimate the eigenvalue corresponding to some eigenstate of a Hamiltonian.  The idea behind these approaches is to begin with an efficiently parameterizable ansatz to the eigenstate.  For the example of ground state energy estimation, the ansatz state is often taken to be a unitary coupled cluster expansion.  The terms in that unitary coupled cluster expansion are then varied to provide the lowest energy for the groundstate.  For excited states a similar argument can be applied, but minimizing a free-energy rather than ground state energy is the most natural approach.

For simplicity, let us focus on the simplest (and most common) example of groundstate estimation.  Consider a Hamiltonian of the form $H=\sum_j a_j U_j$ where $U_j$ is a unitary matrix, $a_j>0$ and $\sum_{j}a_j=1$.  This assumption can be made without loss of generality by renormalizing the Hamiltonian and absorbing signs into the unitary matrix.  Let the state $\ket{\psi(\pmb{x})}$ for $\pmb{x}\in \mathbb{R}^d$ be the variational state prepared by the Prep. and Tuned circuits in Fig.~\ref{subfig:quantumlyTuned}.  Our objective function is then to estimate
\begin{equation}
\pmb{x}_{\rm opt}= \underset{\pmb{x}}{{\rm argmin}}\left(\bra{\psi(\pmb{x})} \sum_j a_j U_j\ket{\psi(\pmb{x})}\right),
\end{equation}
which is real valued because $H$ is Hermitian.

In order to translate this problem to one that we can handle using our gradient descent algorithm, we construct a verifier circuit that given $\ket{\psi(\pmb{x})}$ sets an ancilla qubit to $1$ with probability $p=(1+\bra{\psi(\pmb{x})}H\ket{\psi(\pmb{x})})/2$. This is possible since $\nrm{H}\leq 1$ due to the assumption that $\sum_j a_j =1$.  This requires us to define new unitary oracles that are used to implement the Hamiltonian.

\begin{align}
{\rm prepareW}\ket{0}&= \sum_{j}\sqrt{a_j}\ket{j},\label{eq:prepareW}\\
{\rm selectH}&= \sum_j \ketbra{j}{j}\otimes U_j.\label{eq:selectH}
\end{align}

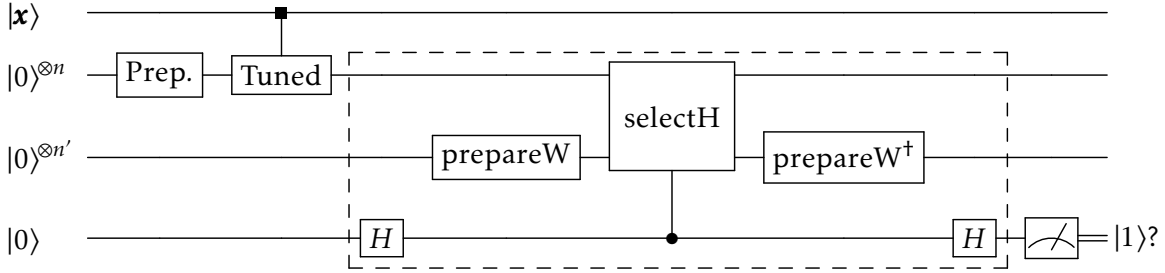
\begin{figure}[ht]
\[
\Qcircuit @C=1.0em @R=1.2em {
\lstick{\ket{\pmb{x}}^{\phantom{\!\otimes n'}}} 			& \qw				&\ctrlA \qw				&\qw							&\qw					&\qw							&\qw							&\qw	&\qw&\qw\\
\lstick{\ket{0}^{\!\otimes n\phantom{'}}}  &\gate{{\rm Prep.}}	&\gate{{\rm Tuned}}\qwx	&\push{\rule{0mm}{3mm}}\qw	&\qw					&\multigate{1}{{\rm selectH}}	&\qw							&\qw	&\qw&\qw\\
\lstick{\ket{0}^{\!\otimes n'}} &\qw				&\qw					&\qw							&\gate{{\rm prepareW}}	&\ghost{{\rm selectH}}			&\gate{{\rm prepareW}^\dagger}	&\qw	&\qw&\qw&\\
\lstick{\ket{0}^{\phantom{\!\otimes n'}}} 				&\qw				&\qw					&\gate{H}						&\qw					&\ctrl{-1}						&\qw							&\gate{H}&\meter&\rstick{\kern-1mm\ket{1}?}\cw
\gategroup{2}{4}{4}{8}{3mm}{--}	
}
\]
\caption{Circuit for converting groundstate energy to a probability for VQE.  The dashed box denotes the verifier circuit, $V$, in Fig.~\ref{subfig:quantumlyTuned} which corresponds here to the Hadamard test circuit.  Probability of measuring $1$ is $1/2 - \bra{\psi(\pmb{x})} H \ket{\psi(\pmb{x})}/2$. Note here Prep. circuit is the identity gate, which is kept in the circuit only for clarity.  Note that in this contexts, the final ${\rm prepareW}^{\dagger}$ can be omitted. \label{fig:vqecirc}}
\end{figure}
We can then use the unitaries \eqref{eq:prepareW}-\eqref{eq:selectH} to compute the query complexity of performing a single variational step in a VQE algorithm.

\begin{corollary}
Let ${\rm Tuned} = \prod_{j=1}^d e^{-i H_j x_j}$ for $\|H_j\|=1$ and $\pmb{x}\in \mathbb{R}^d$ and let ${\rm prep} =I$.  If $H=\sum_{j=1}^M a_j H_j$ for unitary $H_j$ with $a_j\ge 0$ and $\sum_j a_j=1$ then the number of queries to prepareW, selectH and Tuned needed to output a qubit string $\ket{\pmb{y}}$ such that $\left|\nabla \bra{\psi(\pmb{x})}H \ket{\psi(\pmb{x})} - \pmb{y}\right|\le \eps$ with probability at least $2/3$ is $\bigOt{\sqrt{d}/\eps}$.
\end{corollary}
\begin{proof}
First we argue that the circuit in Fig.~\ref{fig:vqecirc} outputs the claimed probability.  We then convert this into a phase oracle, use our results from Jordan's algorithm and demonstrate that $c\in \bigO{1}$ for this problem to show the claimed complexity.

First, if we examine the gates in Fig.~\ref{fig:vqecirc} we note that the prep and Tuned oracles by definition prepare the state $\ket{\psi(\pmb{x})}$.  In this context the prep circuit is the identity.  While this could be trivially removed from the circuit, we retain it to match the formal description of the model that we give earlier.  Under the assumption that $\sum_j a_j =1$ note that
\begin{align}
\bra{0}\bra{\psi(\pmb{x})}{\rm prepareW}^\dagger ({\rm selectH}){\rm prepareW}\ket{0}\ket{\psi(\pmb{x})} &= \sum_j \sum_k \sqrt{a_ja_k}\braket{k}{j}\otimes \bra{\psi(\pmb{x})}U_j \ket{\psi(\pmb{x})}.\nonumber\\
&=\bra{\psi(\pmb{x})}\sum_j a_j U_j \ket{\psi(\pmb{x})}=\bra{\psi(\pmb{x})}H \ket{\psi(\pmb{x})}.\label{eq:Hamexp}
\end{align}
 Then because prepareW is unitary it follows that controlling the selectH operation enacts the controlled ${\rm prepareW}^\dagger ({\rm selectH}){\rm prepareW}$ operation.  

The claim regarding the probability then follows directly from the Hadamard test, which we prove below for completeness.  Let $\Lambda(U)$ be a controlled unitary operation.  Then
\begin{align}
H\Lambda(U)H\ket{0}\ket{\psi(\pmb{x})}&=H(\ket{0}\ket{\psi(\pmb{x})}+\ket{1}U\ket{\psi(\pmb{x})})/\sqrt{2}.\nonumber\\
&=\ket{0}\left(\frac{(1+U)\ket{\psi(\pmb{x})}}{2} \right)+\ket{1}\left(\frac{(1-U)\ket{\psi(\pmb{x})}}{2} \right)
\end{align}
Thus it follows from Born's rule that the probability of measuring the first register to be $1$ is $(1-{\rm Re}(\bra{\psi}U\ket{\psi}))/2$.  We then have from combining this result with~\eqref{eq:Hamexp} and recalling that $H$ is Hermitian gives us that the probability of measuring $1$ in the output of the circuit in Fig.~\ref{fig:vqecirc} is $1/2 - \bra{\psi}H\ket{\psi}/2$ as claimed.  Thus we have an appropriate probability oracle for the approximate groundstate energy expectation.

Each query to the circuit of Fig.~\ref{fig:vqecirc} requires $\bigO{1}$ queries to prepareW and selectH.  Thus the probability oracle can be simulated at cost $\bigO{1}$ fundamental queries.  Now if we remove the measurement from the circuit we see that we can view the transformation as a circuit of the form 
\begin{equation}
U\ket{0}= \sqrt{1/2 - \bra{\psi(\pmb{x})}H\ket{\psi(\pmb{x})}/2}\ket{\psi_{\rm good}}\ket{1}+\sqrt{1/2 + \bra{\psi(\pmb{x})}H\ket{\psi(\pmb{x})}/2}\ket{\psi_{\rm bad}}\ket{0}.\label{eq:Uvqe}
\end{equation}
We then see that the unitary in~\eqref{eq:Uvqe} is exactly of the form required by Theorem~\ref{thm:phaseConv}.  We then have that for any $\delta\in (0,1/3)$ we can simulate a $\delta$--approximate query to the phase oracle analogue of $U$ using $\bigO{\log(1/\delta)}$ applications of $U$.  Since $U$ requires $\bigO{1}$ fundamental queries, the phase oracle can be implemented using $\bigO{1}$ fundamental queries to selectH and prepareW.

From Theorem~\ref{thm:finalScaling} it then follows that we can compute
\begin{equation}
\nabla( 1/2 - \bra{\psi(\pmb{x})}H\ket{\psi(\pmb{x})}/2) = -\nabla\bra{\psi(\pmb{x})}H\ket{\psi(\pmb{x})}/2,
\end{equation}
within error $\eps/2$ and error probability bounded above by $1/3$ using $\bigOt{c\sqrt{d}/\eps}$ applications of selectH and prepareW.  Our result then immediately follows if $c\in \bigOt{1}$.  This is equivalent to proving that for some $c\in \bigOt{1}$ we have that $|\partial_{\alpha_1}\cdots\partial_{\alpha_k} \bra{\psi}H\ket{\psi}|\le c^k$ holds for all $\alpha \in [d]^k$.

We prove that for this application $c\le 2$.  
To see this note that for any index sequence $\alpha \in [d]^k$
\begin{align}
|\partial_{\alpha}\bra{\psi(\pmb{x})} H\ket{\psi(\pmb{x})}| &\le \left\|\partial_{\alpha}\left(\prod_{j=d}^1 e^{iH_j x_j} H \prod_{j=1}^d e^{-iH_j x_j}  \right)\right\|\nonumber\\
&\le \sum_{p=1}^M |a_p| \left\|\partial_{\alpha}\left(\prod_{j=d}^1 e^{iH_j x_j} H_k \prod_{j=1}^d e^{-iH_j x_j}  \right)\right\|.
\end{align}
Lemma~\ref{lemma:unitaryDerivative} directly implies that
\begin{equation}
\left\|\partial_\alpha \prod_{j=d}^1 e^{iH_j x_j}\right\| =1,\label{eq:derivP1}
\end{equation}
and similarly because $H_k$ is unitary and Hermitian for each $k$ Lemma~\ref{lemma:unitaryDerivative} also implies,
\begin{equation}
\left\|H_k\partial_\alpha \prod_{j=1}^d e^{-iH_j x_j}\right\| = 1.\label{eq:derivP2}
\end{equation}
Finally, Lemma~\ref{lemma:derivativeCombination} in concert with Eq.~\eqref{eq:derivP1} and~\eqref{eq:derivP2} then imply
\begin{equation}
\sum_{p=1}^M |a_p| \left\|\partial_{\alpha}\left(\prod_{j=d}^1 e^{iH_j x_j} H_k \prod_{j=1}^d e^{-iH_j x_j}  \right)\right\| \le \sum_{p=1}^M |a_p| 2^k = 2^k.
\end{equation}
\end{proof}

To compare the complexity for this procedure in an unambiguous way to that of existing methods, we need to consider a concrete alternative model for the cost.  For classical methods, we typically assume that the $a_p$ are known classically as are the $U_j$ that appear from queries to selectH.  For this reason, the most relevant aspect to compare is the number of queries needed to the Tuned oracle.  The number of state preparations needed to estimate the gradient is clearly $\bigOt{d/\eps^2}$ using high-degree gradient methods on the empirically estimated gradients~\cite{wecker2015progress} if we assume that the gradient needs to be computed with constant error in $\|\cdot\|_\infty$.  In this sense, we provide a quadratic improvement over such methods when the selectH and prepareW oracles are sub-dominant to the cost of the state preparation algorithm.

The application of this method to QAOA directly follows from the analysis given above.  There are many flavours of the quantum approximate optimization algorithm (QAOA)~\cite{farhi2014qaoa}.  The core idea of the algorithm is to consider a parametrized family of states such as $\ket{\psi(\pmb{x})} = \prod_{j=1}^d e^{-i x_j H_j} \ket{0}$.  The aim is to modify the state in such a way as to maximize an objective function.  In particular, if we let $O$ be a Hermitian operator corresponding to the objective function then we wish to find $\pmb{x}$ such that $\bra{\psi(\pmb{x})}H\ket{\psi(\pmb{x})}$ is maximized.  For example, in the case of combinatorial optimization problems the objective function is usually expressed as the number of satisfied clauses: $O=\sum_{\alpha=1}^m C_{\alpha}$ where $C_{\alpha}$ is $1$ if and only if the $\alpha^{\rm th}$ clause is satisfied and $0$ otherwise~\cite{farhi2014qaoa}.  Such clauses can be expressed as sums of tensor products of Pauli operators, which allows us to express them as Hermitian operators.  Thus, from the perspective of our algorithm, QAOA looks exactly like variational quantum eigensolvers except in that the parameterization chosen for the state may be significantly different from that chosen for variational quantum eigensolvers.

\subsection{Quantum auto-encoders}	\label{section:trainingautoencoding}
Classically, one application of neural networks is \emph{auto-encoders}, which are networks that encode information about a data set into a low-dimensional representation. Auto-encoding was first introduced by Rumelhart et al.~\cite{rumelhart:autoencodingstarter}. Informally, the goal of an auto-encoding circuit is the following: suppose we are given a set of high-dimensional vectors, we would like to learn a representation of the vectors hopefully of low dimenension, so that computations on the original data set can be ``approximately'' carried out by working only with the low-dimensional vectors. More precisely the problem in auto-encoding is: Given $K < N$ and $m$ data vectors $\{v_1,\ldots,v_m\}\subseteq \R^N$, find an encoding map $\E:\R^N\rightarrow \R^K$ and decoding map $\mathcal D:\R^K\rightarrow \R^N$ such that the average squared distortion $\nrm{v_i-(\mathcal D\circ \E)(v_i)}^2$ is minimized:\footnote{There are other natural choices of dissimilarity functions that one might want to minimize, for a comprehensive overview of the classical literature see~\cite{baldi:autoencoding}.} 
\begin{align}\label{eq:distortionmeasure}
	\min_{\E,\mathcal D}\sum_{i\in [m]}\frac{\nrm{v_i-(\mathcal D\circ \E)(v_i)}^2}{m}.
\end{align}

What makes auto-encoding interesting is that it does not assume any prior knowledge about the data set. This makes it a viable technique in machine learning, with various applications in natural language processing, training neural networks, object classification, prediction or extrapolation of information, etc.

Given that classical auto-encoders are `work-horses' of classical machine learning~\cite{azoff:workhorsesneuralnetworks}, it is also natural to consider a quantum variant of this paradigm.
Very recently such quantum auto-encoding schemes have been proposed by Wan Kwak et al.~\cite{wankwok:quantumneuralnetworks} and independently by Romero et al.~\cite{romero:autoencoding}.
Inspired by their work we provide a slightly generalized description of quantum auto-encoders by 'quantizing' auto-encoders the following way: 
we replace the data vectors $v_i$ by quantum states $\rho_i$ and define the maps $\E,\mathcal D$ as quantum channels transforming states back and forth between the Hilbert spaces $\mathcal{H}$ and $\mathcal{H}'$. 
A natural generalization of squared distortion for quantum states $\rho,\sigma$ that we consider is $1-F^2(\rho, \sigma)$,\footnote{Note that some authors (including \cite{romero:autoencoding}) call $F'=F^2$ the fidelity. The distortion measure we use here is $P(\rho,\sigma)=\sqrt{1-F^2(\rho, \sigma)}$, which is called the purified (trace) distance \cite{tomamichelSomoothEntropies}.} giving us the following minimization problem
\begin{equation}
\min_{\E,\mathcal D}\sum_{i\in [m]}\frac{1-F^2\left(\rho_i,(\mathcal D\circ \E)(\rho_i)\right)}{m}.
\end{equation}
Since $F^2\left(\ketbra{\psi}{\psi},\sigma\right)=\braketbra{\psi}{\sigma}{\psi}$
in the special case when the input states are pure states  $\rho_i=\ketbra{\psi_i}{\psi_i}$, the above minimization
problem is equivalent to the maximization problem
\begin{equation}\label{eq:probMaximization}
\max_{\E,\mathcal D}\sum_{i\in [N]}\frac{\braketbra{\psi_i}{\left[(\mathcal D\circ \E)(\ketbra{\psi_i}{\psi_i})\right]}{\psi_i}}{m}.
\end{equation}
Observe that $\braketbra{\psi}{\left[(\mathcal D\circ \E)(\ketbra{\psi}{\psi})\right]}{\psi}$ is the probability of finding the output state $(\mathcal D\circ \E)(\ketbra{\psi}{\psi})$ in state $\ket{\psi}$ after performing the projective measurement $\{\ketbra{\psi}{\psi},I-\ketbra{\psi}{\psi}\}$.
Thus we can think about this as maximizing the probability of recovering the initial pure state after encoding and decoding, which is a natural measure of the quality of the quantum auto-encoding procedure.

\subsubsection{Training quantum auto-encoders}
Similarly to \cite{wankwok:quantumneuralnetworks,romero:autoencoding} we describe a way to perform this optimization problem in the special case when the input states are $n$-qubit pure states and they are mapped to $k$-qubit states, i.e., $\mathcal{H}$ is the Hilbert space of $n$ qubits and $\mathcal{H}'$ is the Hilbert space of $k<n$ qubits. We also show how our gradient computation algorithm can speedup solving the described optimization problem.

We observe that by adding a linear amount of ancilla qubits we can represent the encoding and decoding channels by unitaries, which makes the minimization conceptually simpler.
Indeed by Stinespring's dilation theorem \cite[Corollary 2.27]{WatrousTQI}, \cite{keylFundamentalsQIT} we know that any quantum channel $\E$ that maps $n$ qubit states to $k$ qubit states can be constructed by adding $2k$ qubits initialized in $\ket{\vec{0}}$ state, then acting with a unitary $U_\E$ on the extended space and then tracing out $k+n$ qubits. Applying this result to both $\E$ and $\mathcal D$ results in a unitary circuit representing the generic encoding/decoding procedure, see Figure~\ref{fig:autoencoderChannels}. (This upper bound on the required number of ancilla qubits for $\mathcal D$ becomes $2n$.) 

\begin{figure}[ht]
\[
\Qcircuit @C=5mm @R=6mm {
\lstick{\makebox[12mm][s]{$\ket{0}^{\!\otimes 2k}\!$}} 	& \qw & \qw 							& \multigate{2}{U_{\E}}	& \qw					& \qw										& \qw	& & \text{Ancillae}\\
\lstick{\makebox[12mm][s]{$\ket{0}^{\!\otimes n-k}\!$}}	& \qw & \multigate{1}{\mathrm{Prep}_{\psi}} & \ghost{U_{\E}}	& \qw					& \qw										& \qw	& & \text{for }\E\\
\lstick{\makebox[12mm][s]{$\ket{0}^{\!\otimes k}\!$}} 	& \qw & \ghost{\mathrm{Prep}_{\psi}} 		& \ghost{U_{\E}}	& \multigate{2}{U_{\mathcal D}}	& \multigate{1}{\mathrm{Prep}^{-1}_{\psi}} 	& \meter& & \kern7mm\text{Result }\ket{0}^{\!\otimes n}\\
\lstick{\makebox[12mm][s]{$\ket{0}^{\!\otimes n-k}\!$}}	& \qw & \qw									& \qw				& \ghost{U_{\mathcal D}}		& \ghost{\mathrm{Prep}^{-1}_{\psi}}  		& \meter& & \kern16mm\text{indicates success}\\
\lstick{\makebox[12mm][s]{$\ket{0}^{\!\otimes n+k}\!$}}	& \qw & \qw									& \qw				& \ghost{U_{\mathcal D}}		& \qw										& \qw	& \kern-6mm\}\kern8mm & \kern10mm\text{Ancillae for }\mathcal D
\gategroup{1}{7}{2}{7}{3mm}{\}}
\gategroup{3}{7}{4}{7}{3mm}{\}}
}
\]
\caption{A unitary quantum auto-encoding circuit: For the input $\ket{\psi}$, the circuit prepares $\ket{\psi}$, applies a purified version of the channels $\E,\mathcal D$ and finally checks by a measurement whether the decoded state is $\ket{\psi}$.}
\label{fig:autoencoderChannels}
\end{figure}
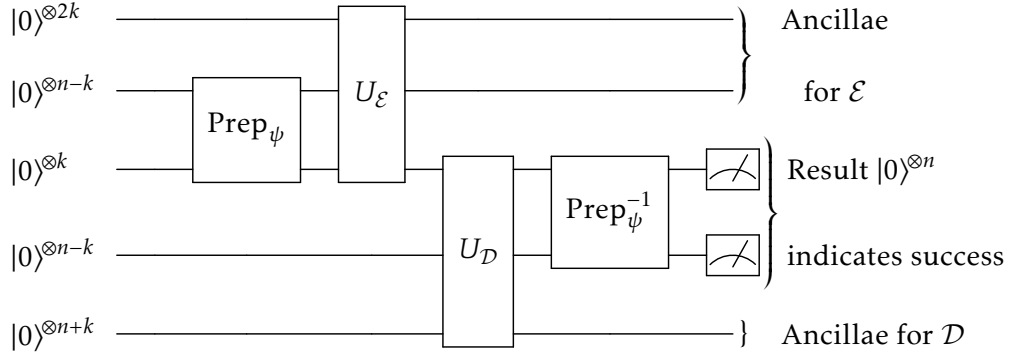

In order to solve the maximization problem \eqref{eq:probMaximization} we could just introduce a parametrization of the unitaries $U_\E,U_\mathcal D$ and search for the optimal parameters using gradient descent. 
Unfortunately a complete parametrization of the unitaries requires exponentially many parameters, which is prohibitive.
However, analogously to, e.g., classical machine learning practices, one could hope that a well-structured circuit can achieve close to optimal performance using only a polynomial number of parameters.
If the circuits $U_\E,U_\mathcal D$ are parametrized nicely, so that Lemma~\ref{lemma:unitaryDerivative} can be applied, then we can use our gradient computation algorithm to speedup optimization.

\begin{figure}[ht] 
	\[
	\kern13mm\Qcircuit @C=1.0em @R=1.2em {
		\lstick{\ket{\pmb{x}}\kern7.4mm} 							& \qw								& \ctrlA \qw						&\qw &\ctrlA \qw						&\qw									&\qw		&\qw\\
		\lstick{\frac{1}{\sqrt{m}}\sum_{i=1}^{m}\ket{m}\kern-1.2mm}	& \ctrlA	\qw						& \qw \qwx							&\qw &\qw \qwx							&\ctrlA	\qw								&\qw		&\qw\\	
		\lstick{\ket{\vec{0}}\kern6.9mm}							& \qw	\qwx						& \multigate{2}{U_\E(\pmb{x})}\qwx	&\qw &\qw \qwx							&\qw	\qwx							&\qw		&\qw\\			
		\lstick{\ket{0}^{\!\otimes n-k}} 							& \multigate{1}{\mathrm{prepS}}\qwx	& \ghost{U_\E(\pmb{x})}				&\qw &\qw \qwx							&\qw\qwx								&\qw		&\qw\\
		\lstick{\ket{0}^{\!\otimes k\phantom{-n}}}					& \ghost{\mathrm{prepS}} 			& \ghost{U_\E(\pmb{x})}				&\qw &\multigate{2}{U_\mathcal D(\pmb{x})} \qwx	&\multigate{1}{\mathrm{prepS}^{-1}}\qwx	&\ctrlo{1}	&\qw\\
		\lstick{\ket{0}^{\!\otimes n-k}}							& \qw								& \qw								&\qw &\ghost{U_\mathcal D(\pmb{x})}				&\ghost{\mathrm{prepS}^{-1}}			&\ctrlo{2}	&\qw\\
		\lstick{\ket{\vec{0}}\kern6.9mm}							& \qw								& \push{\rule{13.7mm}{0.16mm}}\qw	&\qw &\ghost{U_\mathcal D(\pmb{x})}				&\qw									&\qw		&\qw\\					
		\lstick{\ket{0}\kern6.9mm}									& \qw								& \qw								&\qw &\qw								&\push{\rule{17.5mm}{0.16mm}} \qw		&\targ		& \meter	&\rstick{\kern-1mm\ket{1}?}\cw  \\		
	  						 \kern20mm\mbox{Prep.}\gategroup{2}{2}{5}{2}{2mm}{--} &	& \kern20mm\mbox{Tuned}\gategroup{1}{3}{7}{5}{2mm}{--}	& &	& \kern7.5mm\mbox{V}\gategroup{2}{6}{8}{7}{2mm}{--}
	}
	\]		     
	\caption{Quantum circuit which outputs $1$ with probability equal to the objective function \eqref{eq:probMaximization}.
	The structure of the circuit fits the generic model of quantum optimization circuits (Figure~\ref{fig:tunableCircuits}), therefore we can use our gradient computation methods to speedup its optimization.}
	\label{fig:quantumAutoEncoderTraining}
\end{figure}
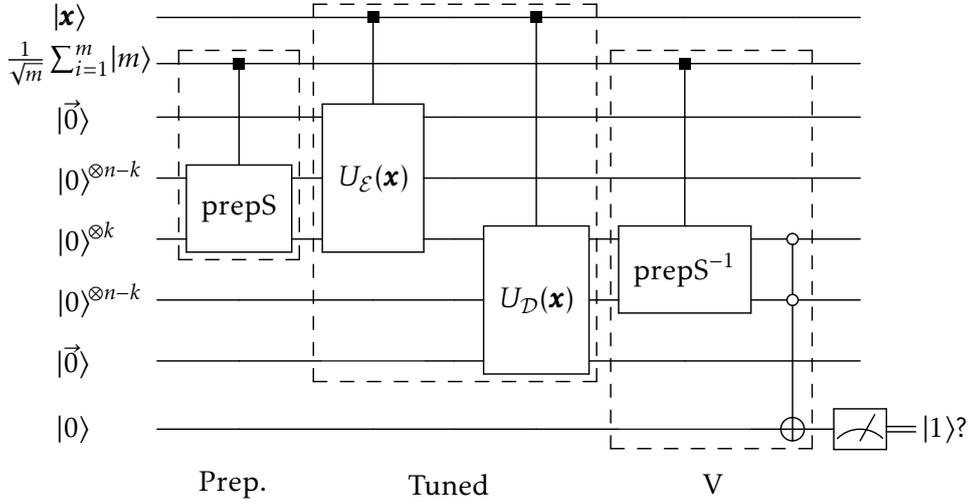

We can do the whole optimization using stochastic gradient descent \cite{JainAcceleratingSGD}, so that in each step we only need to consider the effect of the circuit on a single pure state.
Or if we have more quantum resources available we can directly evaluate the full gradient by preparing a uniform superposition over all input vectors.
In this case the state preparation unitary $\mathrm{Prep}=\sum_{i=1}^{m}\ketbra{i}{i}\otimes \mathrm{Prep}_{\psi_i}$ is a controlled unitary, which controlled on index $i$ would prepare $\ket{\psi_i}$.
Graphically we represent this type of control by a small black square in contrast to the small black circle used for denoting simple controlled unitaries. See the full quantum circuit in Figure~\ref{fig:quantumAutoEncoderTraining}.

Finally, note that in some application it might be desirable to ask for a coherent encoding/decoding procedure, where all the ancilla qubits are returned to the $\ket{\vec{0}}$ state. 
In this case similarly to \cite{wankwok:quantumneuralnetworks,romero:autoencoding} one could define $U_\mathcal D=U_\E^{-1}$ and optimize the probability of measuring $\ket{\vec{0}}$ on the ancilla qubits after applying $U_\E$.

\section{Conclusion and future research}\label{section:openproblems}

We gave a new approach to quantum gradient computation that is asymptotically optimal (up to logarithmic factors) for a class of smooth functions, in terms of the number of queries needed to estimate the gradient within fixed error with respect to the max-norm.  This is based on several new ideas including the use of differentiation formul\ae\kern1mm originating from high-degree interpolatory polynomials.  These high-degree methods quadratically improve the scaling of the query complexity with respect to the approximation quality compared to what one would see if the results from Jordan's work were used. In the case of low-degree multivariate polynomials we showed that our algorithm can yield an exponential speedup compared to Jordan's algorithm or classical algorithms.
We also provided lower bounds on the query complexity of the problem for certain smooth functions revealing that our algorithm is essentially optimal for a class of functions.

While it has proven difficult to find natural applications for Jordan's original algorithm, we provide in this paper several applications of our gradient descent algorithm to areas ranging from machine learning to quantum chemistry simulation.  These applications are built upon a method we provide for interconverting between phase and probability oracles.  The polynomial speedups that we see for these applications is made possible by our improved quantum gradient algorithm via the use of this interconversion process.  
It would be interesting to find applications where we can apply the results for low-degree multivariate polynomials providing an exponential speedup.

More work remains to be done in generalizing the lower bounds for functions that have stronger promises made about the high-order derivatives.  
It would be interesting to see how quantum techniques can speedup more sophisticated higher-level, e.g., stochastic gradient descent methods.
Another interesting question is whether quantum techniques can provide further speedups for calculating higher-order derivatives, 
such as the Hessian, using ideas related to Jordan's algorithm, see e.g. \cite[Appendix D]{jordanPhDThesis}.  Such improvements might open the door for improved quantum analogues of Newton's method and in turn substantially improve the scaling of the number of epochs needed to converge to a local optima in quantum methods.

	\anote{\noindent \underline{Things to keep in mind for possible extensions/applications:}\\
	Quantum neural networks application?\\
	SDP interior point methods?\\}
    
\section*{Acknowledgements}
	The authors thank Ronald de Wolf, Arjan Cornelissen, Yimin Ge and Robin Kothari for helpful suggestions and discussions.
	A.G. and N.W. thank Sebastien Bubeck, Michael Cohen and Yuanzhi Li for insightful discussions about accelerated gradient methods.
	A.G. thanks Joran van Apeldoorn for drawing his attention to the ``cheap gradient principle''. 

    \bibliographystyle{alphaUrlePrint}
	\bibliography{Bibliography}

\appendix
	
	\section{Error bounds on central difference formulas}\label{apx:centralErrorBounds}

	In this appendix we develop error bounds on finite difference formulas using some higher-dimensional calculus.
	The goal is to give upper bounds on the query complexity of gradient computation of $f$ under some smoothness conditions that $f$ satisfies. The main idea is to use Algorithm~\ref{alg:generic} in combination with central difference formulas and analyze the query complexity using some technical lemmas involving higher-dimensional calculus. We first prove Theorem~\ref{thm:lagrangeAlg}, which gives rise to a quantum algorithm that yields potentially exponential speedups for low-degree polynomial functions. 
	The query complexity bound that we can derive for smooth functions using Theorem~\ref{thm:lagrangeAlg}  scales as $\widetilde{O}(d/\varepsilon)$ (which is already an improvement in  $1/\varepsilon$, but worse in $d$ compared to Jordan's algorithm), which we later improve to $\widetilde{O}(\sqrt{d}/\varepsilon)$ via Theorem~\ref{thm:analyticBound}.

	In the proof of the following lemma we will use Stirling's approximation of the factorial:
	\begin{fact}[Stirling's approximation]
	\label{fact:stirling}
	For every $j\in \N_+$, we have
	$$
	\sqrt{2\pi j} \Big(\frac{j}{e}\Big)^j\leq j! \leq e\sqrt{j}\Big(\frac{j}{e}\Big)^j.
	$$
	\end{fact}

	As a first step towards proving Theorem~\ref{thm:lagrangeAlg} and Theorem~\ref{thm:analyticBound} we derive a bound on the following sum of coefficients which appears in finite difference formulas:

	\begin{lemma}\label{lemma:coeffBoundSimple}
		For all $m\in\mathbb{N}_+$ and $k \geq 2m$ we have\footnote{		
			Sometimes we will be interested in bounding $\left|\sum_{\ell=-m}^{m} a_{\ell}^{(2m)}\ell^{k+1}\right|$ rather than the left-hand side of \eqref{eq:allTermBoundSimple}.
			One could ask how good this bound is, do not we lose too much by dismissing the $(-1)^\ell$ cancellations?
			It turns out that the most important case for us is when $k=2m$. In this case our numerical experiments showed that
			the quantity $\left|\sum_{\ell=-m}^{m} a_{\ell}^{(2m)}\ell^{k+1}\right|$ is lower bounded by $(m/e)^{2m}$, showing that the proven upper bound is still qualitatively right.
		}
		\begin{align}\label{eq:allTermBoundSimple}
		\sum_{\ell=-m}^{m} \left|a_{\ell}^{(2m)}\ell^{k+1}\right|
		\leq 6 e^{-\frac{7m}{6}}m^{k+3/2},
        \end{align}
        where $a_{\ell}^{(2m)}$ is defined in Definition~\ref{def:centralDiffernce}
        		\begin{equation*}
			a_{\ell}^{(2m)}= \frac{ (-1)^{\ell-1}}{\ell}\frac{\binom{m}{|\ell|}}{\binom{m+|\ell|}{|\ell|}} \quad \text{ and } a_{0}^{(2m)}=0.
		\end{equation*}
	\end{lemma}
	\begin{proof}
		First we bound the left-hand side of \eqref{eq:allTermBoundSimple} as follows,
		\begin{align}
			\sum_{\ell=-m}^{m} \left|a_{\ell}^{(2m)}\ell^{k+1}\right|
			&= 2\sum_{\ell=1}^{m} \frac{\binom{m}{\ell}}{\binom{m+\ell}{\ell}}\ell^{k} \leq 2m \cdot \max_{\ell\in[m]} \frac{\binom{m}{\ell}}{\binom{m+\ell}{\ell}}\ell^{k}.  \label{eq:coeffMaxSimple}
		\end{align}
			We now upper bound the binomial quantity on the right as follows. For every $\ell\in[m]$, we have
		\begin{align}
			\frac{\binom{m}{\ell}}{\binom{m+\ell}{\ell}}\ell^{k}
			&=\frac{(m!)^2}{(m+\ell)!(m-\ell)!}\ell^{k} \nonumber\\
			&\leq \frac{e^2 m \left(\frac{m}{e}\right)^{\!2m}}{\sqrt{2\pi(m+\ell)}\left(\frac{m+\ell}{e}\right)^{m+\ell}\left(\frac{m-\ell}{e}\right)^{m-\ell}}\ell^{k}	\tag{using\footnote{See below} Fact~\ref{fact:stirling}}\\
			&\leq 3\sqrt{m} \frac{\left(\frac{m}{e}\right)^{\!2m}}{\left(\frac{m+\ell}{e}\right)^{\!\!m+\ell}\left(\frac{m-\ell}{e}\right)^{\!\!m-\ell}}\ell^k \tag{since $e^2/\sqrt{2\pi}\leq 3$}\\
			&= 3\sqrt{m} \frac{\left(\frac{m}{e}\right)^{\!2m}}{\left(\frac{(1+y)m}{e}\right)^{\!\!(1+y)m}\left(\frac{(1-y)m}{e}\right)^{\!\!(1-y)m}} \left(ym\right)^{k} \tag{substitute $y:=\ell/m$}\\							
			&= 3\sqrt{m} \left(\frac{1}{\left(1+y\right)^{1+y}\left(1-y\right)^{1-y}}\right)^{\!\!m} \left(ym\right)^{k} \nonumber\\
			&= 3\sqrt{m} \left(\frac{y^z}{\left(1+y\right)^{1+y}\left(1-y\right)^{1-y}}\right)^{\!\!m} m^k \tag{substitute $z:=k/m$} \\					
			&\leq 3\sqrt{m} \left(\frac{y^2}{\left(1+y\right)^{1+y}\left(1-y\right)^{1-y}}\right)^{\!\!m} m^k \tag{$y\leq 1$ and $z\geq 2$}\\		
			&\leq 3\sqrt{m} \left(e^{-\frac{7}{6}}\right)^{\!\!m} m^k. \tag{by elementary calculs}
		\end{align}	
		\kern200mm\addtocounter{footnote}{-1}
		\footnote{Additionally to Stirling's approximation we also used that $\left(\frac{m-\ell}{e}\right)^{\!m-\ell}\leq (m-\ell)!$, which is true even for $m-\ell=0$.}
		\vskip-6mm
	\end{proof}	    

	Now we are ready to prove Lemma~\ref{lemma:lagrangeBound} from Section~\ref{subsec:finiteDiff}, which we restate here.\oneDCentralDiff*	

\begin{proof} We prove this lemma as follows: first we approximate $f(x)$ using order-$(2m)$ Taylor expansion, and bound the error using Lagrange remainder term. Then we use Lagrange interpolation polynomials to re-express the Taylor polynomial, and use this interpolation formula to approximately compute the derivative of $f$ at $0$ yielding the $(2m)$-th central difference formula. Finally, we use Lemma~\ref{lemma:coeffBoundSimple} to upper bound the difference between $f'(0)\delta$ and the $(2m)$-th central difference formula $f_{(2m)}(\delta)$.

Recall that Taylor's theorem with Lagrange remainder term says that for all $y\in\mathbb{R}$,
		\begin{equation} \label{eq:lagrange}
			f(y) = \underset{:=p(y/\delta)}{\underbrace{\sum_{j=0}^{2m}\frac{f^{(j)}(0)}{j!}y^j}} + \frac{f^{(2m+1)}(\xi)}{(2m+1)!}y^{2m+1}
		\end{equation}			
		for some $\xi \in [0,y]$ (in case $y<0$, then $\xi\in [y,0]$).
		Now let $z:= y/\delta$, we introduce a $\delta$-scaled version of the $(2m)$-th order Taylor polynomial at $0$, which we will use for re-expressing $f'(0)$:
		\begin{equation}\label{eq:taylorLagrangeApx}
			p(z):=-\sum_{j=0}^{2m}\frac{f^{(j)}(0)}{j!}(z\delta)^j=f(z\delta)-\frac{f^{(2m+1)}(\xi)}{(2m+1)!}(z\delta)^{2m+1}.
		\end{equation}
		Because $\deg(p)\leq 2m$ we can use the following Lagrange interpolation formula to represent it~as:
		$$p(z)=\sum_{\ell=-m}^{m}p(\ell)\prod_{\underset{i\neq \ell}{i=-m}}^{m}\frac{z-i}{\ell-i}.$$
		Using the above identity, it is not hard to see that 
		\begin{equation}\label{eq:lagrangeDeriv}
			p'(0)=\sum_{\underset{\ell\neq 0}{\ell=-m}}^{m}p(\ell)\frac{(m!)^2}{-\ell} \frac{(-1)^{\ell}}{(m+|\ell|)!(m-|\ell|)!}
			=\sum_{\underset{\ell\neq 0}{\ell=-m}}^{m}(-1)^{\ell-1} \frac{\binom{m}{|\ell|}}{\binom{m+|\ell|}{|\ell|}}\frac{p(\ell)}{\ell}	
			=\sum_{\ell=-m}^{m}a_{\ell}^{(2m)}p(\ell).
		\end{equation}		
		Observe that by definition $p'(0)=f'(0)\delta$, therefore
		\begin{align*}
			f'(0)\delta
			&=p'(0)\\
			&\overset{\eqref{eq:lagrangeDeriv}}{=}\sum_{\ell=-m}^{m}a_{\ell}^{(2m)}p(\ell)\\
			&\overset{\eqref{eq:taylorLagrangeApx}}{=}\sum_{\ell=-m}^{m}a_{\ell}^{(2m)}\left(f(\ell\delta)-\frac{f^{(2m+1)}(\xi_\ell)}{(2m+1)!}\ell^{2m+1}\delta^{2m+1}\right).
		\end{align*}
		Now we bound the left-hand side of \eqref{eq:nthFiniteDiff} using the above equality:
		\begin{align*}
			\kern-2mm\left|\sum_{\ell=-m}^{m}a_{\ell}^{(2m)} \frac{f^{(2m+1)}(\xi_\ell)}{(2m+1)!}\ell^{2m+1}\delta^{2m+1}\right|
			&\leq \sum_{\ell=-m}^{m}\left|a_{\ell}^{(2m)}\ell^{2m+1}\right| \frac{\nrm{f^{(2m+1)}}_\infty}{(2m+1)!}|\delta|^{2m+1}\\
			&\leq 6 e^{-\frac{7m}{6}}m^{2m+3/2}	\frac{\nrm{f^{(2m+1)}}_\infty}{(2m+1)!} |\delta|^{2m+1} \tag*{(Lemma~\ref{lemma:coeffBoundSimple} with $\!k\!:=\!2m$)\kern-2mm}\\
			&\leq 3 e^{-\frac{7m}{6}}m^{2m+1/2}	\frac{\nrm{f^{(2m+1)}}_\infty}{(2m)!} |\delta|^{2m+1} \\			
			&\leq 3 e^{-\frac{7m}{6}}m^{2m+1/2}	\frac{\nrm{f^{(2m+1)}}_\infty}{\sqrt{4\pi m}(2m/e)^{2m}} |\delta|^{2m+1} \tag{using Fact~\eqref{eq:StirlingBounds}}\\				&\leq e^{-\frac{7m}{6}}\left(\frac{e}{2}\right)^{\!\!2m}\nrm{f^{(2m+1)}}_\infty|\delta|^{2m+1} \tag*{(since $\frac{3}{\sqrt{4\pi}}\leq 1$)}\\
			&\leq e^{-\frac{m}{2}}\nrm{f^{(2m+1)}}_\infty|\delta|^{2m+1} \tag*{(since $e^{-\frac{7m}{6}}\left(\frac{e}{2}\right)^{\!\!2m}\leq e^{-m/2}$)}. 
		\end{align*}
		
		Finally, the first inequality\footnote{We conjecture that the first inequality of \eqref{eq:diffCoeffs} becomes an equality if we take half of the middle term.} in \eqref{eq:diffCoeffs} holds element-wise and the second inequality is standard from elementary calculus, and can be proven using the integral of $1/x$.
	\end{proof}

We now prove a version of Lemma~\ref{lemma:lagrangeBound} but for higher dimensional functions, by making the assumption that all the higher derivatives are bounded. 
	\begin{corollary}\label{cor:lagrangeBound}
		Let $m\in \N$, $B>0$, $\pmb{x}\in\mathbb{R}^d$ and $\pmb{r}:=\pmb{x}/\nrm{\pmb{x}}$. Suppose $f:[-m\nrm{\pmb{x}}_\infty,m\nrm{\pmb{x}}_\infty]^d\rightarrow \mathbb{R}$ is $(2m+1)$-times differentiable and
		$$
		|\partial_{\pmb{r}}^{2m+1} f(\tau\pmb{x})|\leq B \quad \text{ for all  }\tau\in [-m,m],
		$$
		then
		$$\left|f_{(2m)}(\pmb{x})-\nabla f(\pmb{0})\cdot\pmb{x}\right|\leq B e^{-\frac{m}{2}}\nrm{\pmb{x}}^{2m+1}.$$
	\end{corollary}	
   	\begin{proof}
		Consider the function $h(\tau):=f(\tau \pmb{x})$, then
		\begin{align*}
			\left|f_{(2m)}(\pmb{x})-\nabla f(\pmb{0})\cdot\pmb{x}\right|
			&=\left|h_{(2m)}(1)-h'(0)\right|\\
			&\leq e^{-\frac{m}{2}}\sup_{\tau\in [-m,m]}\left|h^{(2m+1)}(\tau)\right|\tag{by Lemma~\ref{lemma:lagrangeBound}} \\
			&= e^{-\frac{m}{2}}\sup_{\tau\in [-m,m]}\left|\partial^{2m+1}_{\pmb{x}}f(\tau \pmb{x})\right|\\
			&= e^{-\frac{m}{2}}\sup_{\tau\in [-m,m]}\left|\partial^{2m+1}_{\pmb{r}}f(\tau \pmb{x})\right|\nrm{\pmb{x}}^{2m+1}\\
			&\leq B e^{-\frac{m}{2}}\nrm{\pmb{x}}^{2m+1}.
		\end{align*}
	\end{proof}

With this corollary in hand, we now show how to calculate the gradient of $f:\R^d\rightarrow \R$ under a bounded higher derivative~condition.

\lagrangeAlg*
\begin{proof}
Let $r_{\text{opt}}:=\frac{2}{\sqrt{d}}\sqrt[2m]{\frac{\eps e^{\frac{m}{2}}}{B \sqrt{d}\cdot 4\cdot 42\cdot \pi}}$, and let $r:=\min\left(r_{\text{opt}},\frac{2R}{m}\right)$. 
By Corollary~\ref{cor:lagrangeBound} we get that whenever $\nrm{\pmb{x}}_\infty\leq r/2$ we have
\begin{align*}
\left|f_{(2m)}(\pmb{x})-\nabla f(\pmb{0})\cdot\pmb{x}\right|
&\leq B e^{-\frac{m}{2}}\nrm{\pmb{x}}^{2m+1}\nonumber\\
&\leq B e^{-\frac{m}{2}}\left(r\frac{\sqrt{d}}{2}\right)^{2m+1}\\
&= B e^{-\frac{m}{2}}\left(r_{\text{opt}}\frac{\sqrt{d}}{2}\right)^{2m+1}\left(\frac{r}{r_{\text{opt}}}\right)^{2m+1}\\
&= \frac{\eps r_{\text{opt}}}{8\cdot 42\cdot \pi}\left(\frac{r}{r_{\text{opt}}}\right)^{2m+1}\\
&\leq \frac{\eps r}{8\cdot 42\cdot \pi}.
\end{align*}

Assume without loss of generality that $\frac{1}{\eps r}=2^n$ for some $n\in\N$.
In Theorem~\ref{thm:genericJordan}, we showed that $\bigOt{\log\left(\frac{d}{\rho}\right)}$ queries to the phase oracle $\mathrm{O}_{\!f_{(2m)}}^{2^{n+1}\pi}$ suffice to compute an $\eps$-precise approximation of the gradient with probability $\geq 1-\rho$. Now observe that the phase oracle $$\mathrm{O}_{\!f_{(2m)}}^{2^{n+1}\pi}(\pmb{x})
=\prod_{\ell=-m}^{m}\mathrm{O}_{\!f}^{2^{n+1}\pi a^{\!(2m)}_\ell}(\ell\pmb{x})
=\prod_{\ell=-m}^{m}\mathrm{O}_{\!f}^{a^{\!(2m)}_\ell\frac{2\pi}{\eps r}}(\ell\pmb{x}),$$
can be implemented using 

\begin{align*}
\sum_{\ell=-m}^{m}\left\lceil a^{\!(2m)}_\ell \frac{2\pi}{\eps r}\right\rceil
&\leq 2m+ \frac{2\pi}{\eps r}\sum_{\ell=-m}^{m} a^{\!(2m)}_\ell \\
&\overset{\eqref{eq:diffCoeffs}}{\leq} 2m+ \frac{2\pi}{\eps r}\left(2\log(m)+2\right)\\
& = \bigO{m+\max\left(\frac{\sqrt{d}}{\eps}\sqrt[2m]{\frac{B\sqrt{d}}{\eps}},\frac{m}{\eps R}\right)\log(2m)}
\end{align*}
 fractional phase queries to $\mathrm{O}_f$.
\end{proof}

	Let us elaborate on the above cost by making some strong regularity assumptions on $f$. Suppose that for every $k\in [2m+1]$, index-sequence $\alpha\in[d]^k$ and $\pmb{x}\in\R^d$, we have $|\partial_\alpha f(\pmb{x})|\leq 1$ (implying also $B=1$). What can we say by using the above corollary?
	
	Well, it could happen\footnote{An example for such a function is $f(\pmb{x}):=\sin(x_1+x_2+\ldots+x_d)$.} that for every $\beta\in[d]^{2m+1}$, we have $\partial_\beta f (\pmb{0})=1$. 
	Then by Eq.~\ref{eq:rthderivateoff}, by picking $\pmb{r}:=\pmb{1}/\sqrt{d}$ we have $\partial^{(2m+1)}_{\pmb{r}}f(\pmb{0})=d^{\frac{2m+1}{2}}$. This is actually the worst possible case under our assumptions,
	it is easy to show that whenever $\nrm{\pmb{r}}\leq 1$, we must have $|\partial^{(2m+1)}_{\pmb{r}}f(\pmb{x})|\leq B=d^{\frac{2m+1}{2}}$ for all $\pmb{x}\in\R^d$. 
	In this case the best complexity we can get from Theorem~\ref{thm:lagrangeAlg} is 
	by choosing $m=\log(d/\eps)$ which yields an overall query complexity upper bound of

	$$
	\bigO{\frac{d}{\eps}\log(d/\rho)\log\log(d/\eps) }.
	$$
		
	This bound achieves the desired $O(1/\eps)$-scaling precision parameter $\eps$, but fails to grasp the $\sqrt{d}$ scaling. 	This failure is mainly due to the loose upper bound bound on $B$. Also as we discussed in Section~\ref{sec:regularFunctionsDiscussion}, we cannot really hope to achieve a $\sqrt{d}$ scaling with an algorithm that implements a phase oracle for an approximate affine function that uniformly approximates an affine function for all points of the hypergrid. 
	But fortunately as we showed in Theorem~\ref{thm:genericJordan}, it is sufficient if the approximation works for a constant fraction of the evaluation points.
 	
 	In order to rigorously prove $\sqrt{d}$ scaling with the dimension we assume that the function is analytic.
  The following lemma we will use for answering the question:
	Given a (complex) analytic function with its multi-dimensional Taylor series as in \eqref{eq:multidimensionalAnalytic}, 
	where do we need to truncate its Taylor series if we want to get a good approximation on the $d$-dimensional hypercube $[-1,1]^d$?
	
	\begin{lemma}\label{lemma:tensorChebyshev}
		Let $d,k \in\mathbb{N}_+$, and suppose $H\in\left(\mathbb{R}^{d}\right)^{\!\otimes k}$ is an order $k$ tensor of dimension $d$, having all elements bounded by $1$ in absolute value, i.e., $\nrm{H}_\infty\leq 1$. Suppose $\{x_1,\ldots,x_d\}$ are i.i.d. symmetric random variable bounded in $[-1/2,1/2]$ and satisfying $\mathbb{E}[(x_i)^{2k-1}]=0$ for every $k\in\mathbb{N}_+$. Then $\mathbb{P}\left[\left|\sum_{\alpha\in[d]^k} H_\alpha x^\alpha  \right|\geq \sqrt{2}\left(r\sqrt{\frac{dk}{2}}\right)^{\!\!k}\right]\leq \frac{1}{r^{2k}}$. 
	\end{lemma}
	\begin{proof}
		\begin{align*}
		\mathbb{E}\left[\left(\sum_{\alpha\in[d]^k} H_\alpha x^\alpha \right)^{\!\!2\,}\right]&=
		\mathbb{E}\left[\sum_{(\alpha,\beta)\in[d]^{2k}} H_\alpha H_\beta x^{(\alpha,\beta)} \right]\\
		&\leq \mathbb{E}\left[\sum_{(\alpha,\beta)\in[d]^{2k}} x^{(\alpha,\beta)} \right] \tag{$x_i$ is symmetric i.i.d. and $\nrm{H}_\infty\leq 1$}\\
		&= \mathbb{E}\left[ \left(x_1+x_2+\ldots+x_d\right)^{2k} \right]\\
		&= \int_{0}^{\infty}\mathbb{P}\left(\left(x_1+x_2+\ldots+x_d\right)^{2k} \geq t \right) dt\\			
		&= \int_{0}^{\infty}\mathbb{P}\left(\left|x_1+x_2+\ldots+x_d\right| \geq t^{1/2k} \right) dt\\			
		&\leq \int_{0}^{\infty}2 e^{-\left(\frac{2}{d} t^{\frac{1}{k}}\right)}  dt \tag{by the Hoeffding bound}\\	
		&=\int_{0}^{\infty}2 \left(\frac{d}{2}\right)^{\!\!k} k y^{k-1} e^{-y}  dy \tag*{$\left(\text{by change of variables }y:=\left(\left(\frac{2}{d}\right)^{\!k}t\right)^{\!\!\frac{1}{k}}\right)$}\\
		&=2 \left(\frac{d}{2}\right)^{\!\!k} k\Gamma(k) \tag{by definition of $\Gamma(x)$}\\				
		&=2 \left(\frac{d}{2}\right)^{\!\!k} k! \tag{main property of $\Gamma(x)$}\\									
		&\leq 2 e\sqrt{k} \left(\frac{dk}{2e}\right)^{\!\!k} \tag{Stirling's approximation}\\
		&< 2 \left(\frac{dk}{2}\right)^{\!\!k}. \tag{for all $k\geq 1: \sqrt{k}e^{1-k}\leq k e^{1-k}\leq 1$}
		\end{align*}
		Now use Chebyshev inequality to conclude.
			\end{proof}

\begin{remark}
	The dependence on $d$ in Lemma~\ref{lemma:tensorChebyshev} cannot be improved, as can be seen using the central limit theorem:
	by choosing $H\equiv 1$ (the all $1$ tensor) it is not hard to see, that for a fixed $k$ the typical value of $\left|\sum_{\alpha\in[d]^k} H_\alpha x^\alpha \right|=\left|\left(x_1+x_2+\ldots+x_d\right)^{k} \right| =\Theta( \sqrt{d}^k)$.
	A natural follow-up question is, if we can improve the $k$-dependence, in particular the $k^{k/2}$ factor? While it is possible that one can improve the above result we show in the next paragraph that the typical value eventually becomes much larger than $\sim d^{\frac{k}{2}}$. (An interesting regime, where our discussion does not imply a lower bound, is when $k\sim\sqrt{d}$.)
			
	Counterexample to the $\sim d^{\frac{k}{2}}$ scaling: suppose $\mathbb{N}\ni a\geq 5$, $d\geq 2^{a}$ and $k=ad$, then let $H$ be the tensor which is $1$ for index-sequences containing each index with even multiplicity, and $0$ otherwise. There are at least $d^{(a-1)d}$ such index sequences since there are $d^{(a-1)d}$ index-sequences of length $(a-1)d$ and each such index-sequence can be extended to an even multiplicity index-sequence of length $ad$. Also suppose that $P(|X_i|\geq 1/4)\geq 1/2$, then this tensor evaluated at every possible value of the random vector will yield at least $d^{(a-1)d}2^{-k}=d^{k-d}2^{-a(k/a)}\geq d^{\left(1-\frac{1}{a}\right)k}d^{-\frac{k}{a}} = d^{\left(1-\frac{2}{a}\right)k} \gg  d^{\frac{k}{2}}$. 
	\anote{Does our query lower bound say something about the asymptotics regarding $k$?}
    \end{remark}

	Now we are ready to prove Theorem~\ref{thm:analyticBound}. We restate the theorem for convenience.
	\analyticBound*
	\begin{proof} 
		Because $f$ is analytic it coincides with its $d$-dimensional Taylor series:
		\begin{align} \label{eq:d-dimensionaltaylor}
		f(\pmb{y}) &=\sum_{k=0}^{\infty}\sum_{\alpha\in[d]^k}\frac{ \pmb{y}^{\alpha} \cdot \partial_\alpha f(\pmb{0})}{k!}.
		\end{align}
		We are now going to use the central differences formula defined earlier in Definition~\ref{def:centralDiffernce}:
		\begin{align*}
		f_{\!(2m)}\!(\pmb{y})
		&=\sum_{\ell=-m}^{m} a_\ell^{(2m)}f\left(\ell \pmb{y} \right)\tag{using Definition~\ref{def:centralDiffernce} }\\
		&=\sum_{\ell=-m}^{m} a_\ell^{(2m)}\sum_{k=0}^{\infty}\frac{1}{k!}\sum_{\alpha\in[d]^k}(\ell \pmb{y})^{\alpha} \cdot \partial_\alpha f(\pmb{0}) \tag{using Eq.~\ref{eq:d-dimensionaltaylor}}\\
		&=\sum_{k=0}^{\infty}\frac{1}{k!}\sum_{\alpha\in[d]^k} \pmb{y}^{\alpha} \cdot \partial_\alpha f(\pmb{0}) 
		\underset{*}{\underbrace{\sum_{\ell=-m}^{m} a_\ell^{(2m)}\ell^k}}.
		\end{align*}
		
		Now, we apply Lemma~\ref{lemma:lagrangeBound} to the function $x^k$ with the choice $\delta=1$, to conclude that $(*)$ is $0$ if $k\leq 2m$ except for $k=1$, in which case it is $1$.
		This implies that 
		\begin{align*}
		\kern-2mm	\left|\nabla f(\pmb{0})\pmb{y}-f_{\!(2m)}\!(\pmb{y})\right|
		&=\left|\sum_{k=2m+1}^{\infty}\frac{1}{k!}\sum_{\alpha\in[d]^k}\!\pmb{y}^{\alpha}	\cdot \partial_\alpha f(\pmb{0}) \sum_{\ell=-m}^{m} a_\ell^{(2m)}\ell^k\right| \\
		&\leq\sum_{k=2m+1}^{\infty}\left(\frac{e}{k}\right)^{\!\!k}\!\frac{1}{\sqrt{4\pi m}}\left|\sum_{\alpha\in[d]^k}\!\pmb{y}^{\alpha} \cdot \partial_\alpha f(\pmb{0})\right|	\left|\sum_{\ell=-m}^{m} a_\ell^{(2m)}\ell^k\right| \tag{Stirling bound \eqref{eq:StirlingBounds}}\\		
		&\leq\sum_{k=2m+1}^{\infty}\left|\sum_{\alpha\in[d]^k}\!\pmb{y}^{\alpha} \cdot \partial_\alpha f(\pmb{0})\right|\left(\frac{e}{k}\right)^{\!\!k}\frac{3 e^{-\frac{7m}{6}}m^{k+\frac{1}{2}}}{\sqrt{\pi m}}	\tag{Lemma~\ref{lemma:coeffBoundSimple} with $k'\!\!:=\!k-\!1$}\\	
		&\leq\sum_{k=2m+1}^{\infty}\left|\sum_{\alpha\in[d]^k}\!\pmb{y}^{\alpha} \cdot \partial_\alpha f(\pmb{0})\right|\frac{1}{\sqrt{2}}\left(\frac{em}{k}\right)^{\!\!k}.	\tag{using $3\sqrt{2/\pi}e^{\!\!-\frac{7m}{6}}\leq 1$}	
		\end{align*}	
		If we take a uniformly random $\pmb{y}\in R\cdot G_d^{(n)}$, then $\pmb{y}$ has coordinates symmetrically distributed around zero, therefore by Lemma~\ref{lemma:tensorChebyshev} (choosing $r:=4$) we know that for all $k\in\mathbb{N}_+$ the ratio of $\pmb{y}$ vectors for which
		\begin{equation}\label{eq:typicalBound}
		\left|\sum_{\alpha\in[d]^k}\frac{\pmb{y}^{\alpha}}{R^k} \cdot \frac{\partial_\alpha f(\pmb{0})}{c^k k^{\frac{k}{2}}}\right|\geq \sqrt{2}\left(4\sqrt{\frac{dk}{2}}\right)^{\!\!k}
		\end{equation}
		is at most $4^{-2k}$. Since $\sum_{k=2m+1}^{\infty}4^{-2k}\leq \sum_{k=3}^{\infty}4^{-2k} < 1/1000$, it follows that apart from a $1/1000$-th fraction of the $\pmb{y}$ vectors, the other $\pmb{y}$s satisfy the following:
		\begin{align*}
		\left|\nabla f(\pmb{0})\pmb{y}-f_{\!(2m)}\!(\pmb{y})\right| &\leq\sum_{k=2m+1}^{\infty}\left|\sum_{\alpha\in[d]^k}\!\pmb{y}^{\alpha} \cdot \partial_\alpha f(\pmb{0})\right|\frac{1}{\sqrt{2}}\left(\frac{em}{k}\right)^{\!\!k} \\
		&\leq \!\! \sum_{k=2m+1}^{\infty} \sqrt{2}\left(\!4\sqrt{\frac{dk}{2}}\right)^{\!\!k} R^kc^k k^{\frac{k}{2}}\frac{1}{\sqrt{2}}\left(\frac{em}{k}\right)^{\!\!k}  \tag{using Eq.~\eqref{eq:typicalBound}}\\
		&= \sum_{k=2m+1}^{\infty}\left(\frac{4\sqrt{d}R c e m}{\sqrt{2}}\right)^{\!\!k} \\
		&< \sum_{k=2m+1}^{\infty}\left(8 R c m \sqrt{d}\right)^{\!\!k}. \tag{using $4e< 8\sqrt{2}$}
		\end{align*}
	\end{proof}		

	\section{Interconversion between phase and probability oracles}\label{apx:oracleConversions}
	In this appendix we show how to convert a phase oracle to a probability oracle efficiently.
	This is useful because it translates our lower bound proof (Theorem~\ref{thm:queryLowerBound}) to the original probability formalism (with a $\log(1/\eps)$ loss due to conversion).
	Also it shows how to implement generalized fractional queries using a phase oracle: We convert the phase oracle to a probability oracle, then
	we reduce the probability using an ancilla qubit, then we convert back the reduced probability oracle to a phase oracle. 
	
	For the conversion we are going to the following lemma from \cite[Lemma 37]{van2017quantum}:
	\begin{lemma}\label{lemma:LowWeightAPX}
	  Let $\delta,\eps\in\!(0,1)$ and $f:\mathbb{R}\rightarrow \mathbb{C}$ s.t. $\left|f(x)\!-\!\sum_{k=0}^K a_k x^k\right|\leq \eps/4$ for all $x\in\![-1+\delta,1-\delta]$.
	  Then $\exists\, c\in\mathbb{C}^{2M+1}$ such that 
	  $$
	  \left|f(x)-\sum_{m=-M}^M c_m e^{\frac{i\pi m}{2}x}\right|\leq \eps
	  $$ 
	  for all $x\in\![-1+\delta,1-\delta]$, where $M=\max\left(2\left\lceil \ln\left(\frac{4\nrm{a}_1}{\eps}\right)\frac{1}{\delta} \right\rceil,0\right)$ and $\nrm{c}_1\leq \nrm{a}_1$. Moreover $c$ can be efficiently calculated on a classical computer in time $\text{poly}(K,M,\log(1/\eps))$.
	\end{lemma}

	Now we are ready to show how to convert a phase oracle to a probability oracle, and for convenience we restate the result here:
	\phaseToProbability*
	\begin{proof}
		First we prove the claim for $\delta=1/4$, because it is conceptually simpler and it avoids the use of phase estimation. Then we generalize the result for arbitrary $\delta\in\!(0,1/2)$.		
		For ease of notation we fix a vector $x$, and denote $p:=p(x)\in\left[1/4 ,3/4 \right]$.
		The basic idea is that we implement the function $\sqrt{p(x)}=\sqrt{p}$ in the amplitude using LCU techniques~\cite{childs2015quantum}.
		We will use the modified phase $p':=p-1/2$, since the corresponding phase oracle can be trivially implemented using $V_p$ and an additional phase gate $e^{-i/2}$.
		The basis of our method is the following Taylor series representation\footnote{
			For concise notation of the coefficients we use generalized binomial coefficients $\binom{p}{k}$, which for a natural number $k$ and a real $p$, is a shorthand for
			$\binom{p}{k}= \frac{p(p-1)\cdots (p-k+1)}{k!}$.	
		}, which is convergent for all $y\in[-1,1]$:
		$$\sqrt{1+y}
		=\sum_{k=0}^{\infty}\binom{1/2}{k}y^k=\sum_{k=0}^{\infty}a_k y^k, \text{ where  }a_k:=\binom{1/2}{k}.$$
		Note that $\sum_{k=1}^{\infty}-|a_k|=\sum_{k=1}^{\infty}a_k(-1)^k=-1$ since $a_0=1$ and $\sqrt{1-1}=0$.
		
		Now observe that for $p'\in [-1/2,1/2]$ and $\theta=10/\pi p'$ we have
		$$\sqrt{1/2+p'}
		=\sqrt{1/2}\sqrt{1+2p'}
		=\sqrt{1/2}\sum_{k=0}^{\infty}a_k(2p')^k=\sum_{k=0}^{\infty}\check{a}_k\theta^k, \text{ where  }\check{a}_k:=\sqrt{1/2}a_k\left(\frac{\pi}{5}\right)^{\!\!k}\!.$$		
		One can verify that $\nrm{\check{a}}_1\leq 1$	
		therefore by choosing $K=\Theta(\log(1/\eps))$ we can ensure that for all $\theta\in(-4/5,4/5)$:
		$$
		\left|\sqrt{1/2+p'}-\sum_{k=0}^{\infty}\check{a}_k \theta^k\right|\leq \eps/4.
		$$
		As Lemma~\ref{lemma:LowWeightAPX} shows there exists $c\in\mathbb{C}^{2M+1}$ such that 
		$$
		\left|\sqrt{1/2+p'}-\sum_{m=-M}^M c_m e^{\frac{i\pi m}{2}\theta}\right|\leq \eps
		$$ 
		for all $\theta\in(-4/5,4/5)$, where $M=\bigO{K}$ and $\nrm{c}_1\leq \nrm{\check{a}}_1\leq 1$.
		By substituting $p'=p-1/2$ and $\theta=10/\pi p'$ we get that for all $p'\in[1/4,3/4]$:
		$$
		\Bigg|\sqrt{p}-\underset{\sqrt{\tilde{p}}}{\underbrace{\sum_{m=-M}^M c_m \left(e^{i5p'}\right)^{\!\!m}}}\Bigg|
		\leq \eps.
		$$ 
		This approximation makes it possible to use the LCU Lemma \cite[Lemma 4]{BerryChildsSim15} in the special case when all unitary is a power of $e^{\pm i5(p(x-1/2))}I$, see e.g. \cite[Lemma 8]{childs2015quantum},
		to implement a unitary $U'_{p}$ such that for $\check{k}=\bigO{\log(K)}$ it performs the map
		$$
		U'_{p}:\ket{x}\ket{0}^{\!\otimes \check{k}}\rightarrow \ket{x}\otimes\left(\frac{\sqrt{\tilde{p}(x)}}{\nrm{c}_1}\ket{0}^{\!\otimes \check{k}}+\sqrt{1-\frac{\tilde{p}(x)}{\nrm{c}^2_1}}\ket{\Phi'^\perp}\right).
		$$
		Adding an extra qubit initialized to $\nrm{c}_1\ket{0}+\sqrt{1-\nrm{c}^2_1}\ket{0}$ and then calculating the OR function of all auxiliary qubits and storing its output in the last qubit amounts to a unitary
		implementing the amplitude $\sqrt{\tilde{p}(x)}$ which is an $\eps$ approximation of $\sqrt{p(x)}$.
		
		In order to generalize the above approach for arbitrary $\delta\in (0,1/2)$ we essentially use \cite[Corollary 42]{van2017quantum}.	
		
		Let $N=\lceil\frac{3\pi}{\delta}\rceil$ and suppose $y_0\in (\delta,1-\delta)$, $y\in [-\pi/(4N),\pi/(4N)]$ and let $\theta:=y2N/\pi$, then
		$$
		\sqrt{y_0-y}=\sqrt{y_0}\sqrt{1+\frac{y}{y_0}}
		=\sqrt{1+\frac{y}{y_0 2N/\pi}}
		=\sqrt{y_0}\sum_{k=0}^{\infty}\binom{1/2}{k}  \left(\frac{\pi}{y_0 2N}\right)^{\!\!k} \theta^k
		=\sum_{k=0}^{\infty}a'_k \theta^k,
		$$	
		where $a'_k=\sqrt{y_0}\binom{1/2}{k}\left(\frac{\pi}{y_0 2N}\right)^{\!\!k} $. 
		Observe that 
		\begin{align*}
		\nrm{a'}_1
		&\leq\sqrt{y_0}+\sum_{k=1}^{\infty}\left|\binom{1/2}{k}\right|\frac{\pi}{y_0 2N} \tag*{$\left(\frac{\pi}{y_0 2N}\leq 1\right)$}\\
		&=\sqrt{y_0}+ \frac{\pi}{\sqrt{y_0} 2N}\tag*{$\left(\sum_{k=1}^{\infty}\left|\binom{1/2}{k}\right|=1\right)$}\\
		&\leq 1. \tag*{$\left(y_0\in[\delta,1-\delta]\text{ and } N\geq 3\pi/\delta\right)$}
		\end{align*}
		Since $|a'_k|\leq (1/2)^k$ for $K'=\Theta(\log(1/\eps))$ we have that $\left|\sum_{k=K'}^{\infty}a'_k \theta^k\right|\leq \eps/4$ for all $\theta\in[-1,1]$. 
		Thus Lemma~\ref{lemma:LowWeightAPX} shows that there exists $\gamma\in\mathbb{C}^{2M'+1}$ such that 
		$$
		\left|\sum_{k=0}^{\infty}a'_k \theta^k-\sum_{m=-M'}^{M'} \gamma_m e^{\frac{i\pi m}{2}\theta}\right|\leq \eps
		$$ 		
		for all $\theta\in\![-1/2,1/2]$, where $M'=\bigO{K}$ and $\nrm{\gamma}_1\leq \nrm{a'}_1\leq 1$.
		This implies that 
		$$
		\left|\sqrt{y_0-y}-\sum_{m=-M'}^{M'} \gamma_m e^{i m N y}\right|\leq \eps/2
		$$ 		
		for all $y\in [-\pi/(4N),\pi/(4N)]$. 
		Supposing that $p(x)\in[y_0-\pi/(4N),y_0+\pi/(4N)]$, setting $p'=p(x)-y_0$ 
		and using the LCU Lemma \cite[Lemma 4]{BerryChildsSim15} in the special case when all unitary is a power of $e^{\pm\frac{i\pi m}{2}p}I$, see e.g. \cite[Lemma 8]{childs2015quantum},
		we can implement a unitary $U'_{p}$ such that for $k'=\bigO{\log(M'N)}$ it performs the map
		$$
		U'_{p}:\ket{x}\ket{0}^{\!\otimes k'}\rightarrow \ket{x}\otimes\left(\sqrt{p}\ket{0}^{\!\otimes k'}+\sqrt{1-\sqrt{p}}\ket{\Phi'^\perp}\right).
		$$
		with $\bigO{M'N}$ uses of the phase oracle $P_{p'}=e^{-iy_0}P_p$.
		
		At this point we know how to implement the amplitude oracle piecewise. One can apply these piecewise 
		implementations in superposition to implement the amplitude for all $p\in[\delta,1-\delta]$ with the help of phase estimation
		as shown in the proof of \cite[Corollary 42]{van2017quantum},
		without increasing the query complexity and using $\bigOt{1/\delta}$ additional gates.
		%
	\end{proof}	

	We have shown how to convert phase oracles and probability oracles back and forth
	with logarithmic overhead in the precision, given that the probabilities are bounded away from $0$ and~$1$.

	Note that when $p$ is close to $0$ or $1$ we actually lose some information when we convert from probability oracles to phase oracles, 
	since the probability is the amplitude squared. (One could also convert the amplitude to phase preventing this loss, 
	but then one needs to be careful because the amplitude can be complex, and the absolute value function is ugly. But one can implement inner product oracles using the Hadamard test as in Section~\ref{sec:variational}, which is sort of an amplitude oracle.)
		
	\section{Robust oblivious amplitude amplification} \label{apx:oblivious}
	In this appendix we introduce a slightly generalized version of robust oblivious amplitude amplification~\cite{berry:simHamTaylor}, and give an analysis with improved constants in the error bounds.
	First we invoke Jordan's result on orthogonal projectors: (In the statement all ket vectors are normalized.)
	\begin{theorem}[Jordan's theorem]\label{thm:JordanSubsp}
		Let $\mathcal{H}$ be a finite dimensional complex Euclidian (i.e., Hilbert) space. If $\Pi_1,\Pi_2$ are orthogonal projectors acting on this space,
		then $\mathcal{H}$ can be decomposed to a direct sum of orthogonal subspaces
		$$
			\mathcal{H}=\bigoplus_{j\in[J]}\mathcal{H}_j^{(1)}\oplus\bigoplus_{k\in[K]}\mathcal{H}_k^{(2)},
		$$ 
		such that for all $j\in[J]: \mathcal{H}_j^{(1)}=\text{Span}(\ket{\varphi_j})$ is a $1$-dimensional subspace satisfying $\nrm{\Pi_1\ket{\varphi_j}}\in \{0,1\}$, $\nrm{\Pi_2\ket{\varphi_j}}\in \{0,1\}$.
		Moreover, for all $k\in[K]:\mathcal{H}_k^{(2)}=\text{Span}(\ket{\psi_k},\ket{\psi^\perp_k})=\text{Span}(\ket{\phi_k},\ket{\phi^\perp_k})$, is a $2$-dimensional subspace satisfying
		$\Pi_1\ket{\psi_k}=\ket{\psi_k}$, $\Pi_2\ket{\phi_k}=\ket{\phi_k}$, $\Pi_1\ket{\psi^\perp_k}=0=\Pi_2\ket{\phi^\perp_k}$,
		moreover $|\braket{\psi_k}{\phi_k}|\notin  \{0,1\}$.
	\end{theorem}

	Inspired by Kothari's ``2D subspace lemma''\cite[Lemma 2.3]{kothariPhDThesis}, we prove a generalization of amplitude amplification using terminology related to Jordan's result.
	
	\begin{lemma}[2D subspace lemma]\label{lemma:2D}
		Let $W$ be a unitary such that $W\ket{\psi}=\sin(\theta)\ket{\phi}+\cos(\theta)\ket{\phi^\perp}$, where $\braket{\phi}{\phi^\perp}=0$.
		Suppose $\Pi_1,\Pi_2$ are orthogonal projectors, such that $\Pi_2\ket{\phi}=\ket{\phi}$ and $\Pi_2\ket{\phi^\perp}=0$
		and similarly $\Pi_1\ket{\psi}=\ket{\psi}$ and	$\Pi_1 W^\dagger \left(\cos(\theta)\ket{\phi}-\sin(\theta)\ket{\phi^\perp}\right)=0$.
		(Note that the requirements on $\Pi_1$ are not contradicting, because $\bra{\psi}W^\dagger\left({\cos(\theta)\ket{\phi}-\sin(\theta)\ket{\phi^\perp}}\right)=0$.\footnote{
			This observation shows that this last condition is trivially satisfied when the rank of $\Pi_1$ is $1$, which is the case for Grover search and amplitude amplification.
		})
		Let $G:=W(2\Pi_1-I)W^\dagger(2\Pi_2-I)$, then 
		$$G^k W \ket{\psi}=\sin((2k+1)\theta)\ket{\phi}+\cos((2k+1)\theta)\ket{\phi^\perp}.$$
	\end{lemma}
	\begin{proof}
		Observe that the subspace $V=\text{Span}(\ket{\phi},\ket{\phi^\perp})$ is invariant under both $(2\Pi_2-I)$ and $W(2\Pi_1-I)W^\dagger$, therefore it is also invariant under $G$.
		Moreover on this subspace $G$ acts a product of two reflections, therefore it is a rotation. It is easy to verify, that its angle of rotation is indeed $2\theta$. Finally note that by definition $W\ket{\psi}\in V$.
	\end{proof}

	\begin{lemma}[Generic robust oblivious amplitude amplification]\label{lemma:oblivious}
		Let $\theta\in (0,\frac{\pi}{6}]$, $\eps\in[0,\frac{1}{2}]$, $\Pi_1,\Pi_2$ orthogonal projectors and $W$ a unitary such that for all $\ket{\psi}\in \mathrm{Im}(\Pi_1)$ it satisfies
		\begin{equation}\label{eq:apxUnifNorm}
			\nrm{\Pi_2 W \ket{\psi}}\in[\sin(\theta)(1-\eps),\sin(\theta)(1+\eps)].
		\end{equation} 
		Let $G:=W(2\Pi_1-I)W^\dagger(2\Pi_2-I)$, then for all $k\in\mathbb{Z}_0^+$ and $\ket{\psi}\in \mathrm{Im}(\Pi_1)$ we have that
		$$\nrm{G^k W \ket{\psi}-\left(\!\frac{\sin((2k\!+\!1)\theta)}{\sin(\theta)}\Pi_2 W \ket{\psi}\!+\!\frac{\cos((2k\!+\!1)\theta)}{\cos(\theta)}(I\!-\!\Pi_2) W \ket{\psi}\!\right)}\leq \eps\left(\!1\!+\!\theta\left(\frac{5}{2}\!+4(2k\!+\!1)\!\right)\!\right).$$
		\anote{Maybe simplify formula.}
	\end{lemma}
	\begin{proof}
		Let $\widetilde{\Pi}_1:=W \Pi_1 W^\dagger$. We apply Jordan's Theorem~\ref{thm:JordanSubsp} on the projectors $\widetilde{\Pi}_1,\Pi_2$ to decompose $\mathcal{H}$ to $1$ and $2$ dimensional subspaces. By \eqref{eq:apxUnifNorm} we know that for each $\ket{\varphi}\in\mathrm{Im}(W \Pi_1 W^\dagger)$ we have $\nrm{\Pi_2 \ket{\varphi}}\in[\sin(\theta)(1-\eps),\sin(\theta)(1+\eps)]\subseteq (0,1)$, and thus no $\ket{\varphi}\in \mathrm{Im}(W \Pi_1 W^\dagger)$ can satisfy $\nrm{\Pi_2\ket{\varphi}}\in \{0,1\}$. 
		Therefore using the notation of Theorem~\ref{thm:JordanSubsp}, we know that 
		\begin{equation}\label{eq:WPi1Span}
			\mathrm{Im}(W \Pi_1 W^\dagger)=\underset{k\in K}{\mathrm{Span}}\{\ket{\psi_k}\}.
		\end{equation}
		Let $\ket{\tilde{\psi}_k}:=W^\dagger\ket{\psi_k}$ and $\theta'=\arcsin(\lVert\Pi_2 W \ket{\tilde{\psi}_k}\rVert)$,
		then we can assume without loss of generality that 
		$W\ket{\tilde{\psi}_k}=\sin(\theta')\ket{\phi_k}+\cos(\theta')\ket{\phi^\perp_k}$, because in the definition of $\mathcal{H}_k^{(2)}=\text{Span}(\ket{\phi_k},\ket{\phi^\perp_k})$ we can multiply the vectors with a unit length complex number (i.e., phase).
		This shows that that $\ket{\tilde{\psi}_k},\theta'$ satisfy the requirements of Lemma~\ref{lemma:2D}, and so 
		$$G^k W \ket{\tilde{\psi}_k}=\sin((2k\!+\!1)\theta')\ket{\phi_k}+\cos((2k\!+\!1)\theta')\ket{\phi^\perp_k} .$$
		Let us define the sub-normalized state
		\begin{align*}
		&\ket{v_k}:=G^k W \ket{\tilde{\psi}_k}-\left(\frac{\sin((2k\!+\!1)\theta)}{\sin(\theta)}\Pi_2 W \ket{\tilde{\psi}_k}+\frac{\cos((2k\!+\!1)\theta)}{\cos(\theta)}(I-\Pi_2) W \ket{\tilde{\psi}_k}\right)\\
		&=\!\left(\!\frac{\sin((2k\!+\!1)\theta')}{\sin(\theta')}\!-\!\frac{\sin((2k\!+\!1)\theta)}{\sin(\theta)}\!\right)\Pi_2 W \ket{\tilde{\psi}_k}\!+\!\left(\!\frac{\cos((2k\!+\!1)\theta')}{\cos(\theta')}\!-\!\frac{\cos((2k\!+\!1)\theta)}{\cos(\theta)}\!\right)(I\!-\!\Pi_2) W \ket{\tilde{\psi}_k}.
		\end{align*}
		It is easy to see using the triangle inequality and the definition of
		$\theta'$, that \anote{Could use orthogonality for a slightly better bound.}
		\begin{align*}
			\nrm{\ket{v_k}}
			&\leq \left|\sin((2k\!+\!1)\theta')-\sin((2k\!+\!1)\theta)\frac{\sin(\theta')}{\sin(\theta)}\right|
			+\left|\cos((2k\!+\!1)\theta')-\cos((2k\!+\!1)\theta)\frac{\cos(\theta')}{\cos(\theta)}\right|\\
			&\leq \underset{\leq (2k\!+\!1)|\theta-\theta'|\leq (2k\!+\!1)2\eps\theta}{\underbrace{\left|\sin((2k\!+\!1)\theta')-\sin((2k\!+\!1)\theta)\right|}}
			+\underset{\leq \eps}{\underbrace{\left|\sin((2k\!+\!1)\theta)\eps\right|}}\\
			&+ \underset{\leq (2k\!+\!1)|\theta-\theta'|\leq (2k\!+\!1)2\eps\theta}{\underbrace{\left|\cos((2k\!+\!1)\theta')-\cos((2k\!+\!1)\theta)\right|}}
			+\underset{\leq \frac{2|\theta'-\theta|}{\sqrt{3}}\leq \frac{4\eps\theta}{\sqrt{3}}\leq \frac{5}{2}\eps\theta }{\underbrace{\left|\cos((2k\!+\!1)\theta)\frac{\cos(\theta')-\cos(\theta)}{\cos(\theta)}\right|}}\\
			&\leq \eps\left(\!1\!+\!\theta\left(\frac{5}{2}\!+4(2k\!+\!1)\!\right)\!\right).
		\end{align*}
		Finally we observe that due to \eqref{eq:WPi1Span} we have 
		$$\mathrm{Im}(\Pi_1 )=\underset{k\in K}{\mathrm{Span}}\{\ket{\tilde{\psi}_k}\},$$
		and that $\braket{v_k}{v_\ell}=0$ for $k\neq \ell$, since $\ket{v_k}\in \mathcal{H}_k^{(2)}$ and $\ket{v_\ell}\in \mathcal{H}_\ell^{(2)}$.
		Therefore we can conclude that for a general $\ket{\psi}\in \mathrm{Im}(\Pi_1 )$ we can write it as 
		$\ket{\psi}=\sum_{k=1}^{K}c_k\ket{\tilde{\psi}_k}$ therefore 
		\begin{align*}
		&\nrm{G^k W \ket{\psi}-\left(\frac{\sin((2k\!+\!1)\theta)}{\sin(\theta)}\Pi_2 W \ket{\psi}+\frac{\cos((2k\!+\!1)\theta)}{\cos(\theta)}(I-\Pi_2) W \ket{\psi}\right)}=\\
		&\qquad=\nrm{\sum_{k=1}^{K}c_k\ket{v_k}}
		=\sqrt{\sum_{k=1}^{K}c^2_k\nrm{\ket{v_k}}^2}
		\leq \eps\left(\!1\!+\!\theta\left(\frac{5}{2}\!+4(2k\!+\!1)\!\right)\!\right).
		\end{align*}
	\end{proof}
	\begin{corollary}[Robust oblivious amplitude amplification]\label{cor:oblivious}
		Let $k\in\mathbb{N}_+$, $\eps\in[0,\frac{1}{2}]$, $\Pi_1,\Pi_2$ orthogonal projectors and $W,U$ unitaries such that for all $\ket{\psi}\in \mathrm{Im}(\Pi_1)$ they satisfy
		\begin{equation}\label{eq:apxUnifNorm2}
		\nrm{\Pi_2 W \ket{\psi}-\sin\left(\frac{\pi}{2(2k+1)}\right)U\ket{\psi}}\leq \sin\left(\frac{\pi}{2(2k+1)}\right)\eps.
		\end{equation} 
		Let $G:=W(2\Pi_1-I)W^\dagger(2\Pi_2-I)$, then for all $\ket{\psi}\in \mathrm{Im}(\Pi_1)$ we have that
		$$\nrm{G^k W \ket{\psi}-U\ket{\psi}}\leq 10 \eps.$$
		\anote{Could improve the constant 10 to 6.}
	\end{corollary}	
	\begin{proof}
		We apply Lemma~\ref{lemma:oblivious} with choosing $\theta:=\frac{\pi}{2(2k+1)}$, and noting that \eqref{eq:apxUnifNorm2} implies \eqref{eq:apxUnifNorm} as can be seen using the triangle inequality.
		Therefore 
		\begin{align*}
			\nrm{G^k W \ket{\psi}-U\ket{\psi}}
			&\leq \nrm{G^k W \ket{\psi}-\Pi_2 W \ket{\psi}/\sin\left(\theta\right)}
			+ \nrm{\Pi_2 W \ket{\psi}/\sin\left(\theta\right)-U\ket{\psi}}\\
			&\leq \eps\left(\!1\!+\!\theta\left(\frac{5}{2}\!+4(2k\!+\!1)\!\right)\!\right) + \eps
			=\eps \left(2+2\pi + \frac{5\pi}{4(2k+1)}\right)\leq 10 \eps.
		\end{align*}
	\end{proof}
\end{document}